%% file: main-Clan.tex
\documentclass[11pt,english]{article}
\usepackage{amsmath,amssymb,amsthm}
\usepackage{multicol}
\usepackage[margin=1in]{geometry}
\usepackage{graphicx,color}
\usepackage{enumitem}
\usepackage{fullpage}
\usepackage[noblocks]{authblk}
\usepackage{tcolorbox}
\usepackage{babel}
\usepackage{wrapfig}
\usepackage{MnSymbol}
\usepackage{mdframed}
\usepackage{mathtools}
\usepackage{floatrow}

\usepackage[ruled,linesnumbered,vlined]{algorithm2e}
\usepackage{tablefootnote}
\usepackage{thm-restate}
\usepackage{bbding}
\usepackage{pifont}
\usepackage{nicefrac}
\usepackage{undertilde}

\definecolor{Darkblue}{rgb}{0,0,0.4}
\definecolor{Brown}{cmyk}{0,0.61,1.,0.60}
\definecolor{Purple}{cmyk}{0.45,0.86,0,0}
\definecolor{Darkgreen}{rgb}{0.133,0.543,0.133}

\usepackage[colorlinks,linkcolor=Darkblue,filecolor=blue,citecolor=blue,urlcolor=Darkblue,pagebackref]{hyperref}
\usepackage[nameinlink]{cleveref}

\usepackage[colorinlistoftodos,prependcaption,textsize=tiny]{todonotes}

\newif\ifdraft 
\draftfalse

\newcommand{\namedref}[2]{\hyperref[#2]{#1~\ref*{#2}}}
\newcommand{\propref}[1]{\hyperref[#1]{property~(\ref*{#1})}}
\newcommand{\theoremref}[1]{\namedref{Thm.}{#1}}

\newtheorem{theorem}{Theorem}
\newtheorem{lemma}{Lemma}
\newtheorem{definition}{Definition}
\newtheorem{claim}{Claim}
\newtheorem{observation}{Observation}
\newtheorem{corollary}{Corollary}

\newtheorem{remark}{Remark}

\newcommand{\poly}{\mathrm{poly}}
\newcommand{\polylog}{\mathrm{polylog}}

\newcommand{\tw}{\mathrm{tw}}
\newcommand{\pw}{\mathrm{pw}}

\newcommand{\R}{\mathbb{R}}
\newcommand{\N}{\mathbb{N}}

\newcommand{\opt}{\mathrm{opt}}

\newcommand{\home}{\mbox{\bf home}}
\newcommand{\diam}{\mathrm{diam}}
\newcommand{\dm}{\mathrm{diam}}
\newcommand{\lo}{\mathrm{lo}}
\newcommand{\hi}{\mathrm{hi}}
\newcommand{\midP}{\mathrm{mid}}

\newtheorem{fact}{Fact}
\definecolor{forestgreen}{rgb}{0.13, 0.55, 0.13}

\SetCommentSty{mycommfont}

\def\eps{\epsilon}

\DeclareMathAlphabet{\mathpzc}{OT1}{pzc}{m}{it}
\newcommand{\etal}{{\em et al. \xspace}}

\newlength{\dhatheight}

\newcommand {\ignore} [1] {}

\newcommand{\initOneLiners}{%
	\setlength{\itemsep}{0pt}
	\setlength{\parsep }{0pt}
	\setlength{\topsep }{0pt}
}

\title{Clan Embeddings into Trees, and Low Treewidth Graphs}
\author{Arnold Filtser\thanks{The research was supported by the Simons Foundation.}}
\affil{Columbia University, \texttt{arnold273@gmail.com}}

\author{Hung Le\thanks{The research was supported by the start-up grant of UMass Amherst.}}
\affil{University of Massachusetts at Amherst, \texttt{hungle@cs.umass.edu}}

\date{}
\begin{document}
\maketitle
\begin{abstract}
In low distortion metric embeddings, the goal is to embed a host ``hard'' metric space into a ``simpler'' target space while approximately preserving pairwise distances. A highly desirable target space is that of a tree metric. Unfortunately, such embedding will result in a huge distortion. 
A celebrated bypass to this problem is stochastic embedding with logarithmic expected distortion. Another bypass is Ramsey-type embedding, where the distortion guarantee  applies only to a subset of the points. 
However, both these solutions fail to provide an embedding into a single tree with a worst-case distortion guarantee on all pairs.
In this paper, we propose a novel third bypass called \emph{clan embedding}. Here each point $x$ is mapped to a subset of points $f(x)$, called a \emph{clan}, with a special \emph{chief} point $\chi(x)\in f(x)$. The clan embedding has multiplicative distortion $t$ if for every pair $(x,y)$ some copy $y'\in f(y)$ in the clan of $y$ is  close to the chief of $x$: $\min_{y'\in f(y)}d(y',\chi(x))\le t\cdot d(x,y)$. Our first result is a clan embedding into a tree with multiplicative distortion $O(\frac{\log n}{\epsilon})$ such that each point has $1+\epsilon$ copies (in expectation).
In addition, we provide a ``spanning'' version of this theorem for graphs  and use it to devise the first compact routing scheme with constant size routing tables.

We then focus on minor-free graphs of diameter prameterized by $D$, which were known to be stochastically embeddable into bounded treewidth graphs with expected additive distortion $\epsilon D$. We devise  Ramsey-type embedding and clan embedding analogs of the stochastic embedding. We use these embeddings to construct the first (bicriteria quasi-polynomial time) approximation scheme for the metric $\rho$-dominating set and metric $\rho$-independent set problems in minor-free graphs.	 
\end{abstract}

\newpage

%\vfill
%\begin{multicols}{2}
%	{\small \setcounter{tocdepth}{2} \tableofcontents}
%\end{multicols}

\setcounter{tocdepth}{2} 
\tableofcontents
    \newpage
    \pagenumbering{arabic}

\input{intro}

\input{Related-work}

\input{preliminaries}

\input{ClanTree}

\addtocontents{toc}{\protect\setcounter{tocdepth}{1}}
\input{SpanningClantree}

\input{LowerBoundTree}

\input{ramsey}

\input{clan}

\addtocontents{toc}{\protect\setcounter{tocdepth}{2}}
\input{applications}

\section*{Acknowledgments}
The authors are grateful to Philip Klein for suggesting the  metric $\rho$-dominating/independent set problems, which eventually led to this project. We thank Vincent Cohen-Addad for useful conversations and for pointing out the proof of \Cref{thm:CKM19} to the first author. 
The first author would like to thank Alexandr Andoni for helpful discussions. The second author would like to thank Michael Lampis for discussing dynamic programming algorithms for metric independent set/dominating set on bounded treewidth graphs.

	\bibliographystyle{alphaurlinit}
	\bibliography{RamseyTreewidthBib,RPTALGbib}
\appendix
\addtocontents{toc}{\protect\setcounter{tocdepth}{1}}
\input{AppendixPathDistortion}
\input{CKM19}

\end{document}

%% file: intro.tex
\section{Introduction}
Low distortion metric embeddings provide a powerful algorithmic toolkit, with applications ranging from approximation/sublinear/online/distributed algorithms \cite{LLR95,AMS99Sketch,BCLLM18,KKMPT12} to  machine learning \cite{GKK17}, biology \cite{HBKKW03}, and vision \cite{AS03}.
Classically, we say that an embedding $f$ from a metric space $(X,d_X)$ to a metric space $(Y,d_Y)$ has multiplicative distortion $t$, if for every pair of points $u,v\in X$ it holds that $d_X(u,v)\le d_Y(f(u),f(v))\le t\cdot d_X(u,v)$.
Typical applications of metric embeddings naturally have the following structures: take some instance of a problem in a ``hard'' metric space $(X,d_X)$; embed $X$ into a ``simple'' metric space $(Y,d_Y)$ via a low-distortion metric embedding $f$; solve the problem in $Y$, and ``pull-back'' the solution in $X$. Thus, the objectives are  low distortion and ``simple'' target space.

Simple target spaces that immediately come to mind are Euclidean space and tree metric, or --- even better --- an ultrametric. \footnote{Ultrametric is a metric space satisfying a strong form of the triangle inequality: 
	$d(x,z)\le\max\left\{ d(x,y),d(y,z)\right\}$ (for all $x,y,z$). Ultrametrics embed isometrically into both Euclidean space \cite{Lemin01}, and tree metric. See \Cref{def:ultra}.}
In a celebrated result, Bourgain \cite{Bou85} showed that every $n$-point metric space embeds into Euclidean space with multiplicative distortion $O(\log n)$ (which is tight \cite{LLR95}). 
On the other hand, any embedding of the $n$-vertex cycle graph $C_n$ into a tree metric will incur multiplicative distortion $\Omega(n)$ \cite{RR98}.
Karp \cite{Karp89} observed that deleting a random edge from $C_n$ results in an embedding into a line with expected distortion $2$ (see \hyperref[fig:circleExample]{Figure 1(a)}).
This idea was developed by Bartal \cite{Bar96,Bartal98} (improving over \cite{AKPW95}), 
and culminating in the celebrated work of Fakcharoenphol, Rao, and Talwar \cite{FRT04} (see also \cite{Bartal04}) who showed that every $n$-point metric space stochastically embeds into trees (actually ultrametrics) with expected multiplicative distortion $O(\log n)$.
Specifically, there is a distribution $\mathcal{D}$, over dominating metric embeddings \footnote{Metric embedding $f:X\rightarrow Y$ is dominating if  $\forall u,v\in X$, $d_X(u,v)\le d_Y(f(u),f(v))$.\label{foot:dominating}} into trees (ultrametrics), such that $\forall u,v\in X$, $ \mathbb{E}_{(f,T)\sim\mathcal{D}}d_T(f(u),f(v))\le O(\log n)\cdot d_X(u,v)$.
The $O(\log n)$ multiplicative distortion is known to be optimal \cite{Bar96}. Stochastic embeddings into trees are widely successful and have found numerous applications (see e.g. \cite{Ind01}).  

In many applications of metric embeddings, a worst-case distortion guarantee is required. A different type of compromise (compared to expected distortion) is provided by \emph{Ramsey-type} embeddings.
The classical Ramsey problem for metric spaces was introduced by Bourgain \etal \cite{BFM86}, and is concerned with finding "nice" structures in arbitrary metric spaces. Following \cite{BBM06,BLMN05b}, Mendel and Naor \cite{MN07} showed that  for every integer parameter $k\ge 1$, every $n$-point metric $(X,d)$ has a subset $M\subseteq X$ of size at least $n^{1-1/k}$ that embeds into a tree (ultrametric) with multiplicative distortion $O(k)$ (see \cite{NT12,BGS16,ACEFN20} for improvements). 
In fact, the embedding has multiplicative distortion $O(k)$ for any pair in $M\times X$. We say that the vertices in $M$ are \emph{satisfied}
(see \hyperref[fig:circleExample]{Figure 1(b)} for an illustration).
As a corollary, every $n$-point metric space $(X,d_X)$ admits a collection ${\cal T}$ of $k\cdot n^{1/k}$ dominating trees over $X$ and a mapping $\home:X\to{\cal T}$, such that for every $x,y\in X$, it holds that $d_{\home(x)}(x,y)\le O(k)\cdot d_X(x,y)$. These are called Ramsey trees, and they have found applications to online algorithms \cite{BBM06}, approximate distance oracles \cite{MN07,C15}, and routing \cite{ACEFN20}.

\begin{figure}[t]
	\centering
	\includegraphics[width=.9\textwidth]{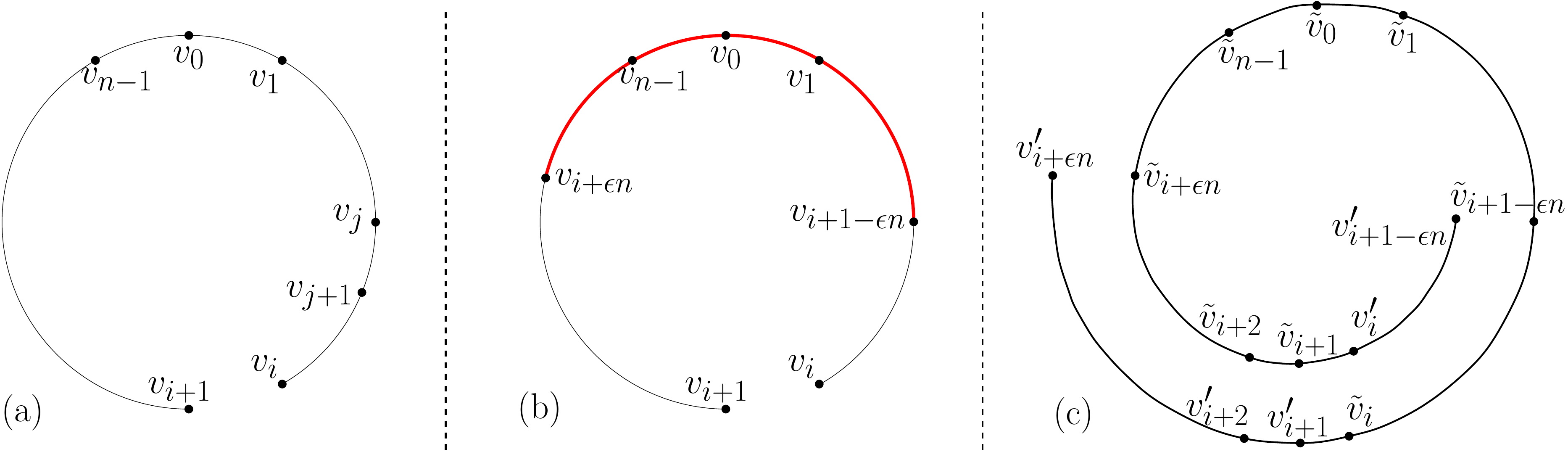}	
	{\caption{\small
			Three different types of embeddings of the cycle graph $C_n$ into a tree. 
			\textbf{(a)}~ On the left illustrated a stochastic embedding that is created by deleting an edge $\{v_i,v_{i+1}\}$ uniformly at random. The expected multiplicative distortion of a pair of neighboring vertices $v_j,v_{j+1}$ is $\mathbb{E}[d_{T}(v_{j},v_{j+1})]=\frac{n-1}{n}\cdot1+\frac{1}{n}\cdot(n-1)=\frac{2n-2}{n}<2$. By the triangle inequality and linearity of expectation, the expected multiplicative distortion is $\le 2$.\newline
			\textbf{(b)}~ In the middle illustrated a Ramsey type embedding: an arbitrary edge $\{v_i,v_{i+1}\}$ is deleted. The vertices in the subset $M$ (on the thick red line), which constitutes a $(1-2\eps)$ fraction of the vertex set, are satisfied. That is, they suffer from a multiplicative distortion at most $\frac{1}{\eps}$ w.r.t. any other vertex. \newline 
			\textbf{(c)}~ On the right illustrated a clan embedding, where $i$ is chosen uniformly at random. The chief of a vertex $v_j$ denoted $\tilde{v}_j$. Each vertex  $v_j\in\{v_{i+1-\eps n},\dots,v_{i+\eps n}\}$ has additional copy $v'_j$; thus the probability that a vertex has two copies is $2\eps$, implying  that  $\mathbb{E}[|f(v_a)|]=1+2\eps$.
			The distortion is $\min\{d(\tilde{v}_{a},\tilde{v}_{b}),d(v'_{a},\tilde{v}_{b})\}\le\frac{1}{\epsilon}\cdot d_{C_{n}}(v{}_{a},v_{b})$. 
		}
		\label{fig:circleExample}}
\end{figure}

\paragraph*{A new type of embedding: clan embedding}
Recall that our initial goal was to embed a general metric space into a ``simple'' target space, specifically a tree metric. A drawback of both the stochastic embedding and the Ramsey-type embedding is that the embeddings are actually into a collection of trees rather than into a single one; thus the target space is not as simple as one might desire.
Each embedding type makes a different type of compromise: the distortion guaranteed in stochastic embedding is only in expectation, while in the  Ramsey-type embedding, only a subset of the vertices enjoys a bounded distortion guarantee.
In this paper, we propose a novel type of compromise, which we call \emph{clan embedding}.
Here we will have a single embedding with a worst-case guarantee on all vertex pairs. The caveat is that each vertex might be mapped to multiple copies. 
This violates the classical paradigm of having a one-to-one relationship between the source and target spaces. However, we obtain a map into a single tree with a worst-case guarantee; this is beneficial and opens a new array of possibilities.

A \emph{one-to-many}  embedding $f:X\rightarrow 2^Y$ maps each point $x$ into a subset $f(x)\subseteq Y$ called the \emph{clan} of $x$. Each vertex $x'\in f(x)$ is called a \emph{copy} of $x$ (see \Cref{def:one-to-many}). 
Clan embedding is a pair $(f,\chi)$, where $f$ is a one-to-many embedding, and $\chi:X\rightarrow Y$ maps each vertex $x$ to a special vertex $\chi(x) \in f(x)$ called the \emph{chief}. Clan embeddings are \emph{dominating}, that is, for every $x,y\in X$, the distance between every two copies is at least the original distance: $\min_{x'\in f(x),y'\in f(y)}d_Y(x',y')\ge d_X(x,y)$. 
$(f,\chi)$ has multiplicative distortion $t$, if for every $x,y\in X$, some vertex in the clan of $x$ is close to the chief of $y$:  $\min_{x'\in f(x)}d_Y(x',\chi(y))\le t\cdot d_X(x,y)$ (see \Cref{def:clan}).
See \hyperref[fig:circleExample]{Figure 1(c)} for an illustration.

\paragraph*{Clan embeddings into trees}
One can easily construct an isometric clan embedding into a tree by allowing $n$ copies for each vertex. On the other hand, with a single copy per vertex, the clan embedding becomes a classic embedding, which requires a multiplicative distortion of $\Omega(n)$.
Our goal is to construct a low distortion clan embedding, while keeping the number of copies each vertex has as small as possible.
To this end, we construct a distribution over clan embeddings, where all the embeddings in the support have a worst-case distortion guarantee;  however, the expected number of copies each vertex has is bounded by a constant arbitrarily close to $1$.

\begin{restatable}[Clan embedding into ultrametric]{theorem}{ClanUltrametric}
	\label{thm:ClanUltrametric}
	Given an $n$-point metric space $(X,d_{X})$ and parameter $\epsilon\in(0,1]$, there is a uniform distribution $\mathcal{D}$ over $O(n\log n/\epsilon^2)$ clan embeddings   $(f,\chi)$ into ulrametrics	with multiplicative distortion $O(\frac{\log n}{\epsilon})$ such that for every point $x\in X$, $\mathbb{E}_{f\sim\mathcal{D}}[|f(x)|]\le1+\epsilon$.
	
	In addition, for every $k\in \N$, there is a uniform distribution $\mathcal{D}$ over $O(n^{1+\frac{2}{k}}\log n)$ clan embeddings $(f,\chi)$
	into ulrametrics with multiplicative distortion $16k$ such that for every point $x\in X$, $\mathbb{E}_{f\sim\mathcal{D}}[|f(x)|]= O(n^{\frac1k})$.
\end{restatable}

We fist show that there exists a distribution $\mathcal{D}$ of clan embeddings that has distortion and expected clan size via the minimax theorem. We then use the multiplicative weights update (MWU) method to explicitly construct a uniform distribution $\mathcal{D}$ of polynomial support as specified by~\Cref{thm:ClanUltrametric}. 

Our clan embedding into ultrametric is asymptotically tight (up to a constant factor in the distortion), and cannot be improved even if we embed into a general tree (rather than to the much more restricted structure of an ultrametric).
Additionally, our lower bound implies that the  ultra-sparse spanner construction of Elkin and Neiman \cite{EN19} is asymptotically tight. (Elkin and Neiman~\cite{EN19} constructed a spanner with stretch $O(\frac{\log n}{\eps})$ and $(1+\eps)n$ edges; see \Cref{rem:UltraSparseSpanners} for further details.)

\begin{restatable}[Lower bound for clan embedding into a tree]{theorem}{LBClanTree}
	\label{thm:ClnUltrametricLB}
	For every fixed $\epsilon\in(0,1)$ and large enough $n$, there is an $n$-point metric space $(X,d_X)$ such that for every clan embedding $(f,\chi)$ of $X$ into a tree with multiplicative distortion $O(\frac{\log n}{\eps})$, it holds that $\sum_{x\in X}|f(x)|\ge(1+\eps)n$.\\
	Furthermore, for every $k\in\N$, there is an $n$-point metric space  $(X,d_X)$ such that for every clan embedding $(f,\chi)$ of $X$ into a tree with multiplicative distortion $O(k)$, it holds that  $\sum_{x\in X}|f(x)|\ge\Omega(n^{1+\frac1k})$.
\end{restatable}

Often, we are given a weighted graph $G=(V,E,w)$, and the goal is to embed the shortest path metric of the graph $d_G$ into a tree $T$. However, if, for example, one is required to construct a network while using only pre-existing edges from $E$, 
it is desirable that the tree $T$ will be a subgraph of $G$, called a spanning tree.
Abraham and Neiman \cite{AN19} (improving over \cite{EEST08}) constructed a stochastic embedding of general graphs into spanning trees with expected distortion $O(\log n\log\log n)$ (losing a $\log\log n$ factor compared to general trees \cite{FRT04}).
Later, Abraham \etal \cite{ACEFN20} constructed Ramsey spanning trees, showing that for every $k\in \N$, every graph can be embedded into a spanning tree with a subset $M$  of at least $n^{1-\frac1k}$ satisfied vertices which suffers a distortion at most $O(k\log\log n)$ w.r.t. any other vertex (again losing a $\log\log n$ factor compared to general trees).
Here we provide a ``spanning'' analog of \Cref{thm:ClanUltrametric}. Similar to \cite{AN19,ACEFN20}, we also lose a $\log\log n$ factor compared to general trees (see the introduction to \Cref{sec:SpanningClan} for further discussion). 
In particular, by \Cref{thm:ClnUltrametricLB}, our spanning clan embedding is optimal up to second-order terms.
As an application, we construct the first compact routing scheme with routing tables of constant size in expectation; see \Cref{subsec:CompactRoutingSchemeIntro}.
We say that a clan embedding $(f,\chi)$ of a graph $G$ into a graph $H$ is \emph{spanning} if $f(V(G))=V(H)$ (i.e., every vertex in $H$ is an image of a vertex in $G$) and for every edge $\{v',u'\}\in E(H)$ where $v'\in f(v),u'\in f(u)$, it holds that $\{v,u\}\in E(G)$ (see \Cref{def:one-to-many,def:clan}).

\begin{restatable}[Spanning clan embedding into trees]{theorem}{ClanSpanningTree}
	\label{thm:ClanSpanningTree}
	Given an $n$-vertex weighted graph $G=(V,E,w)$ and parameter $\epsilon\in(0,1]$, there is a distribution $\mathcal{D}$ over spanning clan embeddings $(f,\chi)$ into trees with multiplicative distortion $O(\frac{\log n\log\log n}{\epsilon})$ such that for every vertex $v\in V$, $\mathbb{E}_{f\sim\mathcal{D}}[|f(v)|]\le1+\epsilon$.
	
	In addition, for every $k\in\N$, there is a distribution $\mathcal{D}$ over spanning clan embeddings $(f,\chi)$
	into trees with multiplicative distortion $O(k\log\log n)$, where for every vertex $v\in V$, \mbox{$\mathbb{E}_{f\sim\mathcal{D}}[|f(v)|]= O(n^{\frac1k})$}.
\end{restatable}

\paragraph{Clan embedding of minor-free graphs into bounded treewidth graphs}
As \cite{Bou85} and \cite{FRT04} are tight, a natural question arises: by embedding from a simpler space (than general $n$-point metric space) into a richer space (than trees), could the distortion be reduced?
The family of low-treewidth graphs is an excellent candidate for a target space: it is a much more expressive target space than trees, while many hard problems remain tractable.
Unfortunately, by the work of Chakrabarti \etal \cite{CJLV08} (see also \cite{CG04}), there are $n$-vertex planar graphs such that every (stochastic) embedding into $o(\sqrt{n})$-treewidth graphs must incur expected multiplicative distortion $\Omega(\log n)$.
Bypassing this roadblock,
Fox-Epstein \etal \cite{FKS19} (improving over \cite{EKM14}), showed
how to embed planar metrics into bounded treewidth graphs while incurring only a small \emph{additive} distortion.
Specifically,  given a planar graph $G$ and a parameter $\epsilon$, they constructed a deterministic dominating embedding $f$ into a graph $H$ of treewidth $\poly(\frac1\eps)$, such that $\forall u,v\in G$, $d_H(f(u),f(v))\le d_G(u,v)+\eps D$, where $D$ is the diameter of $G$. 
While $\eps D$ looks like a crude additive bound, it suffices to obtain approximation schemes for several classic problems: $k$-center, vehicle routing, metric $\rho$-dominating set, and metric $\rho$-independent set.

Following the success in planar graphs, Cohen-Addad \etal  \cite{CFKL20} wanted to generalize to minor-free graphs.
Unfortunately, they showed that already obtaining additive distortion $\frac{1}{20}D$ for $K_6$-minor-free graphs requires  the host graph to have treewidth $\Omega(\sqrt{n})$.
Inspired by the case of trees, \cite{CFKL20} bypass this barrier by constructing a stochastic embedding from $K_r$-minor-free $n$-vertex graphs into a distribution over treewidth-$O_r(\frac{\log n}{\eps^2})$ graphs with expected additive distortion $\eps D$, \footnote{$O_r$ hides some function depending only on $r$. That is, there is some function $\chi:\N\rightarrow\N$ such that $O_r(x)\le \chi(r)\cdot x$.} 
that is 
$\forall u,v\in G$, $\mathbb{E}_{(f,H)\sim\mathcal{D}}[d_H(f(u),f(v))]\le d_G(u,v)+\eps D$.
Similar to the case in planar graphs, Cohen-Addad \etal \cite{CFKL20} used their embedding to construct an approximation scheme  for the capacitated vehicle routing problem  in $K_r$-minor-free graphs. 
However, due to the stochastic nature of the embedding, it was not strong enough to imply any results for the metric $\rho$-dominating/independent problems in minor-free graphs, which, prior to our work, remain wide open.

In this paper, similar to the case of trees, we construct Ramsey-type and clan embedding analogs to the stochastic embedding of \cite{CFKL20}.
Our Ramsey-type embedding bypasses the lower bound of $\Omega(\sqrt{n})$ from \cite{CFKL20} while guaranteeing a worst-case distortion (for a large random subset of vertices).
As an application, we obtain a bicriteria quasi-polynomial time approximation scheme (QPTAS) $^{\ref{foot:approximationSchemes}}$ for the metric $\rho$-independent set problem in minor-free graphs (see \Cref{subsec:Becker}). 

\begin{restatable}[Ramsey-type embedding for minor-free graphs]{theorem}{Ramsey}
	\label{thm:Ramsey-minor-Free}
	Given an $n$-vertex $K_r$-minor-free graph $G=(V,E,w)$ with diameter $D$ and parameters ${\eps\in(0,\frac14)}$, $\delta\in(0,1)$, there is a distribution over dominating embeddings $g:G\rightarrow H$ into graphs of treewidth $O_{r}(\frac{\log^2 n}{\eps\delta})$, such that there is a subset $M\subseteq V$ of vertices for which the following claims hold:
	\begin{enumerate}
		\item For every $u\in V$, $\Pr[u\in M]\ge 1-\delta$.
		\item For every $u\in M$ and $v\in V$, $d_H(g(u),g(v))\le d_G(u,v)+\eps D$.
	\end{enumerate}
\end{restatable}
By setting $\delta=\frac12$ and repeating $\log n$ times, a straightforward corollary is the following.
\begin{corollary}\label{cor:ramseyManyEmbeddings}
	Given a $K_r$-minor-free $n$-vertex graph $G=(V,E,w)$ with diameter $D$ and parameter $\eps\in(0,\frac14)$, there are $\log n$ dominating embeddings $g_1,\dots,g_{\log n}$ into graphs of treewidth $O_{r}(\frac{\log^2 n}{\eps})$, such that for every 
	vertex $v$, there is some embedding $g_{i_v}$, such that
	$$\forall u\in V,\qquad d_{H_{i_v}}(g_{i_v}(u),g_{i_v}(v))\le d_G(u,v)+\eps D~.$$
\end{corollary}

While Ramsey-type embedding is sufficient for the metric $\rho$-independent set problem (as we can restrict our search to  independent sets in $M$), we cannot use it for the metric $\rho$-dominating set problem (as every good solution might contain vertices outside $M$). To resolve this issue, we construct a clan embedding of minor-free graphs into bounded treewidth graphs. As we have a worst-case distortion guarantee for all vertex pairs, we obtain a QPTAS $^{\ref{foot:approximationSchemes}}$ for the metric $\rho$-dominating set problem in minor-free graphs (see \Cref{subsec:Becker}). 
\begin{restatable}[Clan embedding for minor-free graphs]{theorem}{Clan}\label{thm:Clan-Embedding}
Given a $K_r$-minor-free $n$-vertex graph $G=(V,E,w)$ of diameter $D$  and parameters $\eps\in(0,\frac14)$, $\delta\in(0,1)$, there is a distribution $\mathcal{D}$ over clan embeddings $(f,\chi)$ with additive distortion $\eps D$ into graphs of treewidth $O_{r}(\frac{\log^2 n}{\delta\eps})$  such that for every $v\in V$, $\mathbb{E}[|f(v)|]\le 1+\delta$.
\end{restatable}

\subsection{Applications}
\subsubsection{Compact Routing Scheme}\label{subsec:CompactRoutingSchemeIntro}
A \emph{routing scheme} in a network is a mechanism that allows packets to be delivered from any node to any other node. The network is represented as a weighted undirected graph, and each node can forward incoming data by using local information stored at the node,  called a \emph{routing table}, and the (short) packet's \emph{header}. The routing scheme has two main phases: in the preprocessing phase, each node is assigned a routing table and a short \emph{label}; in the routing phase, when a node receives a packet, it should make a local decision, based on its own routing table and the packet's header (which may contain the label of the destination, or a part of it), of where to send the packet.
The {\em stretch} of a routing scheme is the worst-case ratio between the length of a path on which a packet is routed to the shortest possible path.

Compact routing schemes were extensively studied \cite{PU89,ABLP90,AP92,Cowen01,EGP03,TZ01b,C13,ACEFN20}, starting with Peleg and Upfal \cite{PU89}. Using $\tilde{O}(n^{\frac1k})$ table size, Awerbuch et al. \cite{ABLP90} obtained stretch $O(k^29^k)$, which was improved later to  $O(k^2)$ by Awerbuch and Peleg  \cite{AP92}. In their  celebrated compact routing scheme,  Thorup and Zwick \cite{TZ01b} obtained stretch $4k-5$ while using $O(k\cdot n^{1/k})$ size tables and labels of size $O(k\log n)$. \footnote{Unless stated otherwise, we measure space in machine words, each word is $\Theta(\log n)$ bits.\label{foot:words}} 
The stretch was improved to roughly $3.68k$ by Chechik \cite{C13}, using a scheme similar to \cite{TZ01b} (while keeping all other parameters intact). 
Recently, Abraham \etal \cite{ACEFN20} devise a  compact routing scheme (using Ramsey spanning trees) with labels of size $O(\log n)$, tables of size $O(k\cdot n^{1/k})$, and stretch $O(k\log\log n)$. 

In all previous works, the guarantees on the table size are worst case. That is, the table size of every node in the network is bounded by a certain parameter. Here our guarantee is only in expectation. 
Note that such an expected guarantee makes a lot of sense for a central planner constructing a routing scheme for a network where the goal is to minimize the total amount of resources rather than the maximal amount of resources in a single spot.
Even though previous works analyzed worst-case guarantees, if one tries to analyze their expected bounds per vertex, the guarantees will not be improved. Our contribution is the following:

\begin{table}[t]
	\floatbox[{\capbeside\thisfloatsetup{capbesideposition={left,top},capbesidewidth=8.1cm}}]{table}[\FBwidth]
	{\caption{\footnotesize The table compares various routing schemes for $n$-vertex graphs. In rows 1-4, we compare different schemes in their full generality, here $k$ is an integer parameter.
	In rows 5,6,8,10, we fix $k=\log n$, while in rows 7 and 9, we fix $k=\frac{\log n}{\log\log n}$. Note that our result in line 9 is superior to all previous results: it has reduced label size compared to lines 5-6, reduced table size compared to line 7, and reduced stretch compared to line 8.
	Our result in line 10 is the first to obtain a constant table size.\newline
	The sizes of the table and label are measured in words, each word is $O(\log n)$ bits.   
	The header size is asymptotically equal to the label size in all the compared routing schemes.
	The main caveat is that, while in all previous results the table size is analyzed w.r.t. a worst-case guarantee, we only provide bounds in expectation (marked by (*)).
	The label size (as well as the stretch) is a worst-case guarantee in our work as well.}\label{table:routing}}
	{
\scalebox{0.83}{\begin{tabular}{|l|l|l|l|l|}
	\hline\multicolumn{2}{|l|}{\textbf{Routing s.}} 
	& \textbf{Stretch}  	& \textbf{Label}     	& \textbf{Table}             	\\ \hline
	1.&\cite{TZ01b}    	& $4k-5$           		& $O(k\log n)$ 		 	& $O(k n^{1/k})$ 			\\ \hline		
	2.&\cite{C13}   	& $3.68k$          		& $O(k\log n)$ 		 	& $O(k n^{1/k})$  	\\ \hline
	3.&\cite{ACEFN20} 	& $O(k\log\log n)$ 		& $O(\log n)$  		 	& $O(k n^{1/k})$ 	\\ \hline
	4.&\theoremref{thm:route}	    & $O(k\log\log n)$ 		& $O(\log n)$  		 	& $O(n^{1/k})^{(*)}$      	\\ \hline\hline
	5.&\cite{TZ01b} 	& $O(\log n)$         	& $O(\log^2 n)$         & $O(\log n)$  				\\ \hline
	6.&\cite{C13}    	& $O(\log n)$         	& $O(\log^2 n)$         & $O(\log n)$  				\\ \hline
	7.&\cite{ACEFN20} 	& $O(\log n)$ 			& $O(\log n)$           & $O(\log^2 n)$      \\  \hline				
	8.&\cite{ACEFN20} 	& $\widetilde{O}(\log n)$ 	& $O(\log n)$           & $O(\log n)$      \\  \hline\hline		9.&\theoremref{thm:route}   	& $O(\log n)$ 			& $O(\log n)$           & $O(\log n)^{(*)}$  				\\ \hline			
	10.&\theoremref{thm:route}   	& $\widetilde{O}(\log n)$ 	& $O(\log n)$           & $\boldsymbol{O(1)}^{(*)}$  						\\ \hline	
	
\end{tabular}}	
}
\end{table}
	
\begin{restatable}[Compact routing scheme]{theorem}{Routing}\label{thm:route}
	Given a weighted graph $G=(V,E,w)$ on $n$ vertices and integer parameter $k> 1$, there is a compact routing scheme with stretch $O(k\log\log n)$ that has (worst-case) labels (and headers) of size $O(\log n)$, and the expected size of the routing table of each vertex is $O(n^{1/k})$.
\end{restatable}
See \Cref{table:routing} for comparison of our and previous results.
We mainly focus on the very compact regime where all the parameters are at most poly-logarithmic.
A key result in \cite{TZ01b} is a stretch $1$ routing scheme for the special case of a tree, where a routing table has constant size, and logarithmic label size (see \Cref{thm:tree-routh}).
All the previous works are based on constructing a collection of trees. Specifically, in \cite{TZ01b,C13}, there are $n$ trees, where each vertex belongs to $O(\log n)$ trees, and for each pair of nodes, there is a tree that guarantees a small stretch. Routing is then done in that tree. 
This is the reason for their large label size of $\log^2n$ (as a label consists of $\log n$ labels in different trees).
\cite{ACEFN20} constructs $\log n$ (Ramsey spanning) trees in total, where each vertex $v$ has a home tree $T_v$, such that $v$ enjoys a small stretch w.r.t. any other vertex in $T_v$. The label then consists of the name of $T_v$ and the label of $v$ in $T_v$. However, the routing table is still somewhat large as one needs to store the routing information in $\log n$ different trees.

In contrast, our construction is based on the spanning clan embedding $(f,\chi)$ of \Cref{thm:ClanSpanningTree} into a single tree $T$, where the clan of each vertex consists of $O(1)$ copies (in expectation). The label of each vertex $v$ is simply the label of $\chi(v)$ in $T$. The routing table of $v$ contains the routing tables of all the corresponding copies in $f(v)$.

\subsubsection{Metric Baker Problems in Minor-free graphs}\label{subsec:Becker}
Baker \cite{Baker94} introduced a ``layering'' technique in order to  construct efficient polynomial approximation schemes (EPTAS) \footnote{A polynomial time approximation scheme (PTAS) is an algorithm that for any fixed $\eps\in(0,1)$, provides a $(1+\eps)$-approximation in polynomial time. A PTAS is an \emph{efficient} polynomial time approximation scheme (EPTAS) if running time is of the form $n^{O(1)}\cdot f(\eps)$ for some function $f(.)$ depending on $\epsilon$ only. A quasi-polynomial time approximation scheme (QPTAS) has running time $2^{\cdot\polylog(n)}$ for every fixed $\epsilon$.\label{foot:approximationSchemes}} for many ``local'' problems in planar graphs such as minimum-measure \emph{dominating set} and maximum-measure \emph{independent set}.
The key observation is that planar graphs have the ``bounded local treewidth'' property. Baker showed that for some problems solvable on bounded treewidth graphs, one can construct efficient approximation schemes for graphs possessing the bounded local treewidth property.
This approach was generalized by Demaine \etal \cite{DHK05} to minor-free graphs.

Eisenstat \etal \cite{EKM14} proposed metric generalizations of Baker problems: minimum measure $\rho$\emph{-dominating set}, and maximum measure $\rho$\emph{-independent set}. 
Given a metric space $(X,d_X)$, a $\rho$-independent set is a subset $S\subseteq X$ of points such that for every $x,y\in S$, $d_X(x,y)> \rho$. 
Similarly, a $\rho$-dominating set is a subset $S\subseteq X$ such that for every $x\in X$, there exists $y\in S$, such that $d_X(x,y)\le \rho$.
Given a measure $\mu:X\rightarrow \R_+$, the goal of the metric $\rho$-dominating (resp. independent) set problem is to find a $\rho$-dominating (resp. independent) set of minimum (resp. maximum) measure.
It is often the case that metric Baker problems are much easier under the uniform measure.
Sometimes, in addition, we are given a set of terminals ${\cal K}\subseteq X$, and required only that the terminals will be dominated ($\forall x\in {\cal K},~\exists y\in S$ s.t.  $d_X(x,y)\ge \rho$).
Note that the metric generalization of Becker problems in structured graphs (e.g. planar) is considerably harder than the non-metric problems. This is because the graph describing dominance/independence relations no longer posses the original structure (e.g. planarity).
 
An approximation scheme for the $\rho$-dominating (resp. independent) set problem returns a $\rho$-dominating (resp. independent) set $S$ such that for every $\rho$-dominating (resp. independent) set $S'$ it holds that $\mu(S)\le(1+\eps)\mu(S')$ (resp. $\mu(S)\ge(1-\eps)\mu(S')$).
A bicriteria approximation scheme for the $\rho$-dominating (resp. independent) set problem returns a $(1+\eps)\rho$-dominating (resp. $(1-\eps)\rho$-independent) set $S$ such that for every $\rho$-dominating (resp. independent) set $S'$ it holds that $\mu(S)\le(1+\eps)\mu(S')$ (resp. $\mu(S)\ge(1-\eps)\mu(S')$).

For unweighted graphs with treewidth $\tw$, Borradaile and Le \cite{BL16} provided an exact algorithm for the $\rho$-dominating set problem with $O((2\rho+1)^{\tw+1}n)$ running time (see also \cite{DFHT05}).
For general treewidth $\tw$ graphs, using dynamic programming technique, Katsikarelis \etal \cite{KLP19} designed a fixed parameter tractable (FPT) approximation algorithm for the metric $\rho$-dominating set problem with $(\tw/\eps)^{O(\tw)}\cdot\poly(n)$ runtime that returns a $(1+\eps)\rho$-dominating set $S$, such that for every $\rho$-dominating  set $S'$ it holds that $\mu(S)\le\mu(S')$. A similar result was also obtained for the  metric $\rho$-independent set problem \cite{KLP20}. 
In particular, for the very basic case of bounded treewidth graphs, no true approximation scheme (even with quasi-polynomial time) is known for these  problems.
Additional evidence was provided by Marx and Pilipczuk \cite{MP15} (see also \cite{FKS19}), who showed that the existence of EPTAS $^{\ref{foot:approximationSchemes}}$ for either  $\rho$-dominating/independent set problem in planar graphs would refute the exponential-time hypothesis (ETH).
Given this evidence, it is natural to settle for bicriteria approximation.

For unweighted planar graphs and constant $\rho$, there are linear time approximation schemes (not bicriteria) for the metric $\rho$-independent/dominating set problems \cite{EILM16,DFHT05}.
In weighted planar graphs, under the uniform measure, Marx and Pilipczuk \cite{MP15} gave exact $n^{O(\sqrt{k})}$ time solution to both metric $\rho$-dominating/isolated set problems, provided that the solution is guaranteed to be of size at most $k$.
Using their embedding of planar graphs into $\eps^{-O(1)}\log n$-treewidth graphs with additive distortion $\eps D$, Eisenstat \etal \cite{EKM14} provided a bicriteria PTAS $^{\ref{foot:approximationSchemes}}$ for both metric $\rho$-independent/dominating set problems in planar graphs.
Later, by constructing an improved embedding into $\eps^{-O(1)}$-treewidth graphs, Fox-Epstein \etal \cite{FKS19} obtained a  bicriteria EPTAS.$^{\ref{foot:approximationSchemes}}$

\begin{table}[t]
\scalebox{0.91}{\begin{tabular}{|l|l|l|l|l|}		
		\hline\multicolumn{2}{|l|}{\textbf{Reference}}			& \textbf{Family}        & \textbf{Result}                                              & \textbf{Technique}         \\ \hline
		1.&\cite{MP15}               								& planar                 & No EPTAS  under ETH                                          &                            \\ \hline
		2.&\cite{KLP19,KLP20}        								& treewidth			  	 & FPT with  approx   $(1+\eps)\rho$ 				& Dynamic programming        \\ \hline
		3.&\cite{EKM14}              								& planar                 & Bicriteria PTAS                                              & Deterministic embedding    \\ \hline
		4.&\cite{FKS19}              								& planar                 & Bicriteria EPTAS                                             & Deterministic embedding    \\ \hline
		5.&Theorems \ref{thm:CKM19}\&\ref{thm:Local-Sarch-Independent}	& minor-free             & PTAS (uniform measure)                                                       & Local search               \\ \hline
		6.&Theorems \ref{thm:isolated}\&\ref{thm:dominatingSet}     & minor-free             & Bicriteria QPTAS                                             & Clan/Ramsey type embedding \\ \hline
	\end{tabular}}
\caption{\small The table compares different approximation schemes for metric Becker problems on weighted graphs. All compared results apply to both metric $\rho$-dominating/independent set problems. All the results (other than in line 5) apply to the general measure case.\label{tab:Becker}}
\end{table}

Finally, we turn to the most challenging case of minor-free graphs. 
For the restricted uniform measure case, using local search (similarly to \cite{CKM19}), we construct PTAS for both metric $\rho$-dominating/independent set problems. See \Cref{thm:CKM19,thm:Local-Sarch-Independent} in \Cref{sec:LocalSearch} for details. However, the local search approach seems to be hopeless for general measures.
Alternately, one can try the metric embedding approach (for which bicriteria approximation is unavoidable). Unfortunately, unlike the classic embeddings in \cite{EKM14,FKS19}, Cohen-Addad \etal \cite{CFKL20} provided a stochastic embedding with an \emph{expected distortion} guarantee.
Such a stochastic guarantee is not strong enough to construct approximation schemes for the metric $\rho$-independent/dominating set problems.
Using our clan and Ramsey-type embeddings, we are able to provide the first bicriteria QPTAS $^{\ref{foot:approximationSchemes}}$ for these problems. See \Cref{tab:Becker} for a summary of previous and current results.

\begin{restatable}[Metric $\rho$-independent set]{theorem}{isolated}\label{thm:isolated}
	There is a bicriteria quasi-polynomial time approximation scheme (QPTAS) for the metric $\rho$-independent set problem in $K_r$-minor-free graphs.\\
	Specifically, given a weighted $n$-vertex $K_r$-minor-free graph $G=(V,E,w)$, measure $\mu:V\rightarrow \R_+$ and parameters $\eps\in(0,\frac14)$, $\rho>0$, in $2^{\tilde{O}_r(\frac{\log^2n}{\eps^2})}$ time, one can find a $(1-\eps)\rho$-independent set $S\subseteq V$ such that for every $\rho$-independent set $\tilde{S}$,
	$\mu(S)\ge (1-\eps)\mu(\tilde{S})$.
\end{restatable}

\begin{restatable}[Metric $\rho$-dominating set]{theorem}{dominating}\label{thm:dominatingSet}
	There is a bicriteria quasi-polynomial time approximation scheme (QPTAS) for the metric $\rho$-dominating set problem in $K_r$-minor-free graphs.\\
	Specifically, given a weighted $n$-vertex $K_r$-minor-free graph $G=(V,E,w)$, measure $\mu:V\rightarrow \R_+$, a subset of terminals ${\cal K}\subseteq V$, and parameters $\eps\in(0,\frac14)$, $\rho>0$, in $2^{\tilde{O}_r(\frac{\log^2n}{\eps^2})}$ time, one can find a $(1+\eps)\rho$-dominating set $S\subseteq V$ for ${\cal K}$ such that for every $\rho$-dominating set $\tilde{S}$ of ${\cal K}$ ,
	$\mu(S)\le (1+\eps)\mu(\tilde{S})$.
\end{restatable}

\subsection{Paper Overview}
The paper overview uses terminology presented in the preliminaries \Cref{sec:prelim}. 
\paragraph{Clan embedding into ultrametric}
The main task is to prove a ``distributional'' version of \Cref{thm:ClanUltrametric}. Specifically, given a parameter $k$, and a measure $\mu:X\rightarrow \R_{\ge1}$, we construct a clan embedding with distortion $16k$ such that $\sum_{x\in X}\mu(x)\cdot |f(x)|\le \mu(X)^{1+\frac1k}$, where $\mu(X)=\sum_{x\in X}\mu(x)$ (\Cref{lem:clanTreeMeasure}).  We show that the distributioal version implies  \Cref{thm:ClanUltrametric} by using the minimax theorem.

The algorithm to construct the distributional version is a deterministic recursive ball growing algorithm, which is somewhat similar to previous deterministic algorithms constructing Ramsey trees \cite{Bar11,ACEFN20}.
Let $D$ be the diameter of the metric space. We grow a ball $B(v,R)$ around a point $v$ and partition the space into two clusters: the interior $B(v,R+\frac{D}{16k})$ and exterior $X\setminus B(v,R-\frac{D}{16k})$ of the ball, while points at distance $\frac{D}{16k}$ from the boundary of the ball belong to both clusters.
 We then recursively create a clan embedding into ultrametrics for each of the two clusters. These two embeddings are later combined into a single ultrametric where the root has label $D$. See \Cref{fig:UltrametricClanEmbedding} for an illustration.
The $16k$ distortion guarantee follows from the wide ``belt'' around the boundary of the ball belonging to both clusters.
Note that the images of vertices in this ``belt'' contain copies in the clan embeddings of both clusters, while ``non-belt'' points have copies in a single embedding only.
However, the two clusters have cardinality smaller than $|X|$. The key is to  carve the partition while guaranteeing that the relative measure of points belonging to both clusters will be small compared to the reduction in cardinality.

\paragraph{Spanning clan embedding into trees}
In \Cref{thm:ClanSpanningTree}, the spanning version, we try to imitate the approach of \Cref{thm:ClanUltrametric}. However, we cannot simply carve balls and continue recursively. The reason is that the diameter of a cluster could grow unboundedly after deleting some vertices. In particular, there is no clear upper bound on the distance between separated points.

To imitate the ball growing approach nonetheless, we use the petal-decomposition framework that was previously applied to create stochastic embedding into spanning trees \cite{AN19}, and Ramsey spanning trees \cite{ACEFN20}. 
The petal decomposition framework enables one to iteratively construct a spanning tree for a given graph. In each level, the current cluster is partitioned into smaller diameter pieces (called \emph{petals}), which have properties resembling balls. The algorithm continues recursively on the petals. Later, the petals are connected back to create a spanning tree.
The key property is that while creating a petal, we have a certain degree of freedom to chose its ``radius'', which enables us to use the ball growing approach from above. Crucially, the framework guarantees that for every choice of radii (within the sepecified limits), the diameter of the resulting tree will be only constant times larger than that of the original graph.
However, the petal decomposition framework does not provide us with the freedom to choose the center of the petal. This makes the task of controlling the number of copies more subtle.

\paragraph{Lower bound for clan embedding into a tree}
We provide here a proof sketch for the first assertion in \Cref{thm:ClnUltrametricLB}.
We begin by constructing an $n$-vertex graph $G=(V,E)$ with $(1+\eps)n$ edges and girth $g=\Omega(\frac{\log n}{\eps})$; the girth is  the length of the shortest cycle.
Consider an arbitrary clan embedding of $G$ into a tree $T$ with distortion $\frac{g}{c}=O(\frac{\log n}{\eps})$ (for some constant $c$) and $\kappa$ copies overall. 
We create a new graph $H$ by merging all the copies of each vertex into a single vertex. There is a naturally defined classic embedding from $G$ to $H$ with distortion $\le\frac{g}{c}$.
The Euler characteristic of the graph $G$ equals $\chi(G)=|E|-|V|+1=\eps n+1$, while the Euler characteristic of $H$ is at most $\chi(H)\le \kappa-n$. However, Rabinovich and Raz~\cite{RR98} showed that, if an embedding from a girth-$g$ graph $G$ has distortion $\le\frac{g}{c}$, the host graph must have the Euler characteristic at least as large as that of $G$. Thus, we conclude that $\kappa\ge (1+\eps)n+1$ as required.

\paragraph{Ramsey type embedding for minor-free graphs}
The structure theorem of Robertson and Seymour \cite{RS03} stated that every minor-free graph can be decomposed into a collection of graphs embedded on the surface of constant genus (with some vortices and apices), glued together into a tree structure by taking clique-sums. The stochastic embedding of minor free graphs into  a distribution over bounded treewidth graphs by Cohen-Addad \etal \cite{CFKL20} was constructed according to the layers of the structure theorem.
First, they constructed an embedding for a planar graph with a single vortex. Then, they generalized it to planar graphs with multiple vortices, subsequently to graphs embedded on the surface of constant genus with multiple vortices, and to surface embeddable graphs with multiple vortices and apices. Finally, they incorporated cliques-sums and generalized to minor-free graphs. Most crucially, for this paper, the only step requiring randomness was the incorporation of apices.
Specifically, \cite{CFKL20} constructed a deterministic embedding for graphs embedded on the surface of constant genus with multiple vortices.
This is the starting point of our embeddings.

Our first step is to incorporate apices, however, instead of guaranteeing that the distance of each pair is distorted by $\eps D$ in expectation, we will show that each vertex with probability $1-\delta$ enjoys a small distortion w.r.t. any other vertex.
We begin by deleting all the apices $\Psi$ and obtaining a surface embeddable graph with multiple vortices $G'=G[V\setminus\Psi]$. However, the diameter of the resulting graph is essentially unbounded. Pick an arbitrary vertex $r$, and partition $G'$ into layers of width $O(\frac{D}{\delta})$ w.r.t. distances from $r$ with a random shift
\footnote{Alternatively, one could use here a strong padded decomposition \cite{Fil19padded} (as in \cite{CFKL20}) into clusters of diameter $O_r(\frac{D}{\delta})$ such that each radius-$D$ ball is fully contained in a single cluster with probability $1-\delta$. However, this approach will not work for our clan embedding, as there is no bound on the number of copies we will need for failed vertices. We use the layering approach for the \Cref{thm:Ramsey-minor-Free} as well to keep the proofs of \Cref{thm:Ramsey-minor-Free,thm:Clan-Embedding} similar.}. 
It follows that every vertex $v$ is $2D$-padded (that is, the ball $B(v,2D)$ is fully contained in a single layer) with probability $1-\delta$.
The set $M$ of \emph{satisfied} vertices defined to be the set of all $D$-padded vertices. We then use the deterministic embedding from \cite{CFKL20} on every layer with distortion parameter $\eps'=\Theta(\eps\delta)$ to incur additive distortion $\eps D$. Finally, we combine all these embeddings together into a single embedding, which also contains the apices. 

The next step is to incorporate clique-sums. This is done recursively w.r.t. the clique-sum decomposition tree $\mathbb{T}$. In each step, we pick a central piece $\tilde{G}\in\mathbb{T}$ such that $\mathbb{T}\setminus \tilde{G}$ breaks into connected components $\mathbb{T}_1,\mathbb{T}_2,\dots$, where each $\mathbb{T}_i$ contains at most $|\mathbb{T}|/2$ pieces. 
We construct a Ramsey-type embedding for $\tilde{G}$ using the lemma above and obtain a set $\tilde{M}$ of satisfied vertices.
Recursively, we construct a Ramsey-type embedding for each $\mathbb{T}_i$ and obtain a set $M_i$ of satisfied vertices. 
We ensure that all these embeddings are clique-preserving. Thus even though eventually we will obtain a one-to-one embedding, during the process, we keep them \emph{one-to-many and clique-preserving}. This provides us with a natural way to combine all the embeddings of $\tilde{G},\mathbb{T}_1,\mathbb{T}_2,\dots$ into a single embedding into a graph of bounded treewidth (by identifying vertices of respective clique copies). All the vertices in $\tilde{M}$ will be satisfied. A vertex $v\in \mathbb{T}_i$ will be satisfied if $v\in M_i$ and all the vertices in the clique $Q_i$, used in the clique sum of $\tilde{G}$ with $\mathbb{T}_i$, are satisfied $Q_i\subseteq\tilde{M}$. Analyzing the entire process, we show that each vertex is satisfied with probability at least $(1-\delta)^{\log n}$. The theorem follows by setting the parameter $\delta'=\Theta(\frac{\delta}{\log n})$.

\paragraph{Clan embedding for minor-free graphs}
The construction here follows similar lines to our Ramsey-type embedding. However,  we cannot simply ``give-up'' on vertices, as we required to provide a worst-case distortion guarantee on all vertex pairs.
Similarly to the Ramsey-type case, we build on the deterministic embedding of surface embeddable graphs with vortices from \cite{CFKL20}, and generalize it to a clan embedding of graphs including the apices.
However, there is one crucial difference in creating the ``layering'' (with the random shift). In the Ramsey-type embedding, vertices near the boundary between two layers simply failed and did not join $M$. Here, instead, the layers will somewhat overlap such that copies of vertices near boundary areas will be split into two unrelated sets. In particular, cliques that lie near boundary areas will have two separated clique copies w.r.t. each corresponding layer (at most two). Even though that actually each vertex will have an essentially unbounded number of copies (due to the clique-preservation requirement), the copies of each vertex will be divided to either one or two sets, such that in the final embedding, it will be enough to pick an arbitrary single copy from each set.
The copies of a vertex will split into two sets only if it is in the area of the boundary, the probability of which is bounded by $\delta$.

The generalization to clique-sums also follows similar lines to the Ramsey-type embedding. We create a clan embedding for $\tilde{G}$ into treewidth graph $\tilde{H}$ as above, and recursively clan embeddings $H_1,H_2,\dots$ for $\mathbb{T}_1,\mathbb{T}_2,\dots$. For each $\mathbb{T}_i$, we will make the vertices of the clique $Q_i$, used for the clique-sum between $\tilde{G}$ and $\mathbb{T}_i$, into apices, thereby  ensuring that $H_i$ will succeed on $Q_i$. In particular, every vertex $v\in Q_i$ will have a single copy in $H_i$.
When combining $H_i$ with $\tilde{H}$, there are two cases. If the embedding $\tilde{H}$ was successful w.r.t. $Q_i$ we will simply identify between the two clique copies and done. Otherwise, $\tilde{H}$ will contain two vertex-disjoint clique copies $\tilde{Q}_i^1,\tilde{Q}_i^2$ of $Q_i$. We will create two disjoint copies of the embedding $H_i$: $H_i^1,H_i^2$, and identify the two copies of $Q_i$ in $H_i^1,H_i^2$ with $\tilde{Q}_i^1,\tilde{Q}_i^2$, respectively. It follows that for a vertex $v\in \mathbb{T}_i$, with probability at least $1-\delta$, the number of copies it will have is the same as in $H_i$, while with probability at most $\delta$ it will be doubled.
Analyzing the entire process (and picking a single copy from each relevant set as above), we show that each vertex is expected to have at most  $(1+\delta)^{\log n}$ copies. The theorem follows by using the parameter $\delta'=\Theta(\frac{\delta}{\log n})$.

%% file: Related-work.tex
\subsection{Related Work}\label{sec:related}
\textbf{Path-distortion} 
A closely related notion to clan embeddings is multi-embedding studied by Bartal and Mendel \cite{BM04multi}. A multi-embedding is a dominating one-to-many embedding. The distortion guarantee, however, is very different. We say that a multi-embedding $f:X\rightarrow 2^Y$ between metric spaces $(X,d_X)$, $(Y,d_Y)$ has \emph{path distortion} $t$, if for every ``path'' in $X$, i.e., a sequence of points $x_0,x_1,\dots,x_q$, there are copies $x'_i\in f(x_i)$ such that $\sum_{i=0}^{q-1}d_Y(x'_i,x'_{i+1})\le t\cdot \sum_{i=0}^{q-1}d_X(x_i,x_{i+1})$.
For $n$ point metric space $(X,d)$ with aspect ratio $\Phi$ 
\footnote{The aspect ratio of a metric space $(X,d)$ is the ratio between the maximal and minimal distances $\frac{\max_{x,y}d(x,y)}{\min_{x\not= y}d(x,y)}$.\label{foot:aspectRatio}}, 
and parameter $k\ge1$, 
Bartal and Mendel~\cite{BM04multi} constructed a multi-embedding into ultrametric with $O(n^{1+\frac1k})$ vertices and distortion $O(k\cdot\min\{\log n\cdot\log\log n,\log \Phi\cdot\log\log \Phi\})$.
Formally, path distortion and multiplicative distortion of clan embedding are incomparable, as clan embedding guarantees small distortion with respect to a single chief vertex (which is crucial to our applications), while the multi-embedding \cite{BM04multi} distortion guarantee is w.r.t. arbitrary copies, but preserve entire ``paths''.
Interestingly, a small modification to our clan embedding provides the path distortion guarantee as well! See \Cref{thm:ClanUltrametricAlt} in \Cref{app:PathDistortion}.
Specifically, we obtain embedding into  ultrametric with $O(n^{1+\frac1k})$ (resp. $(1+\eps)n$) vertices and distortion $O(k\cdot\min\{\log n,\log \Phi\})$ (resp.  $O(\frac{\log n}{\eps}\cdot\min\{\log n,\log \Phi\})$), shaving a $\log\log$ factor compared with \cite{BM04multi}.
In a private communication, Bartal told us that he obtained the exact same path distortion guarantees more than a decade ago; Bartal's manuscript is made public recently~\cite{Bar21}.

In a concurrent paper, Haeupler \etal \cite{HHZ21} studied a closely related notion of tree embeddings with copies. They construct a one-to-many embedding of a graph $G$ into a tree $T$ where every vertex has at most $O(\log n)$ copies, and such that every connect subgraph $H$ of $G$ has a connected copy $H'$ in $T$, of weight at most $O(\log^2n)\cdot w(H)$.
Using the path distortion gurantee in our embedding (or \cite{Bar21}), one will obtain an embedding such that every connect subgraph $H$ of $G$ has a connected copy $H'$ in $T$, of weight at most $O(\log n)\cdot w(H)$, however the bound on the maximal number of copies will be only polynomial.

\textbf{Tree covers.} The constructions of  Ramsey trees are asymptotically tight \cite{BBM06}. Furthermore, as was shown by Bartal \etal \cite{BFN19Ramsey} that they cannot be substantially improved even for planar graphs with a constant doubling dimension.  \footnote{Specifically, for every $\alpha>0$, \cite{BFN19Ramsey} constructed planar graph with constant doubling dimension, such that for every tree embedding, the subset of vertices  enjoying distortion $\le \alpha$ is of size at most $n^{1-\Omega(\frac{1}{\alpha\log\alpha})}$, which is almost as bad as general graphs.}
Therefore \cite{BFN19Ramsey} suggested studying a weaker gurantee provided by tree covers. 
Here the goal is to construct a small collection of dominating embeddings into trees such that every pair of vertices has a small distortion in some tree in the collection.
For  $n$-vertex minor-free graph \cite{BFN19Ramsey} constructed $1+\eps$ tree covers of size  $O_r(\frac{\log^2n}{\eps^2})$ (or a $O(1)$-tree cover $O(1)$ size). For metrics with doubling dimension $d$, \cite{BFN19Ramsey} constructed $1+\eps$-tree covers of size $(\frac1\eps)^{O(d)}$. 
Recently, the authors \cite{FL22Relaible} showed that for doubling metrics, we can replace the trees by ultrametrics.%\atodo{Added citation to our new paper!}

\textbf{Minor free graphs.} Different types of embedding were studied for minor-free graphs. $K_r$-minor-free graphs embed into $\ell_p$ space with multiplicative distortion $O_r(\log^{\min\{\frac12,\frac1p\}}n)$ \cite{Rao99,KLMN04,AGGNT19,AFGN18}. In particular, they embed into $\ell_\infty$ of dimension $O_r(\log^2n)$
with a constant multiplicative distortion. They also admit spanners with multiplicative distortion $1+\eps$ and $\tilde{O}_r(\eps^{-3})$ lightness \cite{BLW17}.
On the other hand, there are other graph families that embed well into bounded treewidth graphs.
Talwar \cite{Talwar04} showed that graphs with doubling dimension $d$ and aspect ratio $\Phi$ $^{\ref{foot:aspectRatio}}$, stochastically embed into graphs with treewidth $\eps^{-O(d\log d)}\cdot\log^d\Phi$ with expected distortion $1+\eps$.
Similar embeddings are known for graphs with highway dimension $h$ \cite{FFKP18} (into treewidth $(\log \Phi)^{-O(\log^2\frac h\eps)}$ graphs), and graphs with correlation dimension $k$ \cite{CG12} (into treewidth $\tilde{O}_{k,\eps}(\sqrt{n})$ graphs).

%% file: preliminaries.tex
	\section{Preliminaries} \label{sec:prelim}
	$\tilde{O}$ notation hides poly-logarithmic factors, that is $\tilde{O}(g)=O(g)\cdot\polylog(g)$, while 
	$O_r$ notation hides factors in $r$, e.g. $O_r(m)=O(m)\cdot f(r)$ for some function $f$ of $r$. All logarithms are at base $2$ (unless specified otherwise).
	
	We consider connected undirected graphs $G=(V,E)$ with edge weights $w_G: E \to \R_{\ge 0}$.
	A graph is called unweighted if all its edges have unit weight.
	Additionally, we denote $G$'s vertex set and edge set by $V(G)$ and $E(G)$, respectively. Often, we will abuse notation and write $G$ instead of $V(G)$. $d_{G}$ denotes the shortest path metric in $G$, i.e., $d_G(u,v)$ is the shortest distance between $u$ to $v$ in $G$. 
	Note that every metric space can be represented as the shortest path metric of a weighted complete graph. We will use the notions of metric spaces, and weighted graphs interchangeably.
	When the graph is clear from the context, we might use $w$ to refer to $w_G$, and $d$ to refer to $d_G$. 
	$G[S]$ denotes the induced subgraph by $S$. The diameter of $S$, denoted by $\dm(S)$, is $\max_{u,v \in S}d_{G[S]}(u,v)$.
\footnote{This is often called \emph{strong} diameter. A related notion is the \emph{weak} diameter of a cluster $S$ , defined to be $ \max_{u,v \in S}d_{G}(u,v)$. Note that for a metric space, weak and strong diameters are equivalent.}

An ultrametric $\left(X,d\right)$ is a metric space satisfying a
strong form of the triangle inequality, that is, for all $x,y,z\in X$,
$d(x,z)\le\max\left\{ d(x,y),d(y,z)\right\}$. The following definition is known to be an equivalent one (see \cite{BLMN05}).
\begin{definition}\label{def:ultra}
	An ultrametric is a metric space $\left(X,d\right)$ whose elements
	are the leaves of a rooted labeled tree $T$. Each $z\in T$ is associated
	with a label $\ell\left(z\right)\ge0$ such that if $x\in T$ is a
	descendant of $z$ then $\ell\left(x\right)\le\ell\left(z\right)$
	and $\ell\left(x\right)=0$ iff $x$ is a leaf. The distance between leaves $x,y\in X$ is defined as $d_{T}(x,y)=\ell\left(\mbox{lca}\left(x,y\right)\right)$
	where $\mbox{lca}\left(x,y\right)$ is the least common ancestor of
	$x$ and $y$ in $T$.
\end{definition}
 
\subsection{Metric Embeddings} 
	Classically, a metric embedding is defined as a function $f:X\rightarrow Y$ between the points of two metric spaces $(X,d_X)$ and $(Y,d_Y)$.
	A metric embedding $f$ is said to be \emph{dominating} if for every pair of points $x,y\in X$, it holds that $d_X(x,y)\le d_Y(f(x),f(y))$. 
	The distortion of a dominating embedding $f$ is $\max_{x \not= y\in X}\frac{d_Y(f(x),f(y))}{d_X(x,y)}$.	
	Here we will study a  more permitting generalization of metric embedding introduced by Cohen-Addad \etal \cite{CFKL20}, which is called \emph{one-to-many} embedding . 
\begin{definition}[One-to-many embedding]\label{def:one-to-many}
	A \emph{one-to-many embedding} is a function $f:X\rightarrow2^Y$  from the points of a metric space $(X,d_X)$ into non-empty subsets of points of a metric space $(Y,d_Y)$, where the subsets $\{f(x)\}_{x\in X}$ are disjoint.	
	$f^{-1}(x')$ denotes the unique point $x\in X$ such that $x'\in f(x)$. If no such point exists, $f^{-1}(x')=\emptyset$.
	A point $x'\in f(x)$ is called a \emph{copy} of $x$, while $f(x)$ is called the \emph{clan} of $x$.
	For a subset $A\subseteq X$ of vertices, denote $f(A)=\cup_{x\in A}f(x)$.
	
	We say that $f$ is \emph{dominating} if for every pair of points $x,y\in X$, it holds that $d_X(x,y)\le \min_{x'\in f(x),y'\in f(y)}d_Y(x',y')$.
	We say that $f$	has multiplicative distortion $t$, if it is dominating and $\forall x,y\in X$, it holds that $\max_{x'\in f(x),y'\in f(y)}d_Y(x',y')\le t\cdot d_X(x,y)$.
	Similarly, $f$ has additive distortion $\eps D$ if $f$ is dominating and $\forall x,y\in X$, $\max_{x'\in f(x),y'\in f(y)}d_Y(x',y')\le d_X(x,y)+\eps D$.
	
	A stochastic one-to-many embedding is a distribution $\mathcal{D}$ over dominating one-to-many embeddings. We say that a stochastic one-to-many embedding has expected multiplicative distortion $t$ if $\forall x,y\in X$, $\mathbb{E}[\max_{x'\in f(x),y'\in f(y)}d_Y(x',y')]\le t\cdot d_X(u,v)$.
	Similarly, $f$ has  expected additive distortion $\eps D$, if $\forall x,y\in X$, $\mathbb{E}[\max_{x'\in f(x),y'\in f(y)}d_Y(x',y')]\le d_X(x,y)+\eps D$.
	
	For a one-to-many embedding $f$ between weighted graphs $G=(V,E,w)$ and $H=(V',E',w')$, we say that $f$ is spanning if $V'=f(V)$ (i.e. $f$ is ``onto''), and for every edge $(u,v)\in E'$, it holds that $\left(f^{-1}(u),f^{-1}(v)\right)\in E$ and $w'(u,v)=w\left(f^{-1}(u),f^{-1}(v)\right)$.
\end{definition}

This paper is mainly devoted to the new notion of clan embeddings.
\begin{definition}[Clan embedding]\label{def:clan}
	A clan embedding from metric space $(X,d_X)$ into a metric space $(Y,d_Y)$ is a pair $(f,\chi)$ where $f:X\rightarrow2^{Y}$ is a dominating one-to-many embedding, and $\chi:X\rightarrow Y$ is a classic embedding. For every $x\in X$, we have that $\chi(x)\in f(x)$; here $f(x)$ called the clan of $x$, while $\chi(x)$ is referred to as the chief of the clan of $x$ (or simply the chief of $x$).

	\sloppy 
	We say that clan embedding $f$ has multiplicative distortion $t$ if for every
	$x,y\in X$, \mbox{$\min_{y'\in f(y)}d_{Y}(y',\chi(x))\le t\cdot d_{X}(x,y)$}.
	Similarly, $f$ has additive distortion $\epsilon D$ if for every $x,y\in X$,
	\mbox{$\min_{y'\in f(y)}d_{Y}(y',\chi(x))\le d_{X}(x,y)+\epsilon D$}. 
		
	A clan embedding $(f,\chi)$ is said to be spanning if $f$ is a spanning one-to-many embedding.
\end{definition}

We will construct embeddings for minor-free graphs using a divide-and-concur approach. First, we will construct embedding on each piece (see below). Then, in order to combine different embeddings into a single one, it will be important that these embeddings are \emph{clique-preserving}. 
\begin{definition}[Clique-copy]
	Consider a one-to-many embedding $f:G\rightarrow 2^H$, and a clique $Q$ in $G$. A subset $Q'\subseteq f(Q)$ is called clique copy of $Q$ if $Q'$ is a clique in $H$, and for every vertex $v\in Q$, $Q'\cap f(v)$ is a singleton.
\end{definition}
\begin{definition}[Clique-preserving embedding]
	A one-to-many embedding $f:G\rightarrow2^H$ is called clique-preserving embedding if for every clique $Q$ in $G$, $f(Q)$ contains a clique copy of $Q$.
	A clan embedding $(f,\chi)$ is clique-preserving if $f$ is clique preserving.
\end{definition}

	\subsection[Minor Structure Theorem]{Robertson-Seymour Structure Theorem 
	}\label{subsec:RobertsonSeymour}
	
   In this section, we review notation used in graph minor theory by Robertson and Seymour. 
	Informally speaking, the celebrated theorem of Robertson and Seymour (\Cref{thm:RS}, \cite{RS03}) said that every minor-free graph can be decomposed into a collection of graphs \emph{nearly embeddable} in the surface of constant genus, glued together into a tree structure by taking \emph{clique-sum}.  To formally state the Robertson-Seymour decomposition, we need additional notation.
	
	\begin{definition}[Tree/Path decomposition]
		A tree decomposition of $G(V,E)$, denoted by $\mathcal{T}$, is a tree satisfying the following conditions: 
	\begin{enumerate} [noitemsep,nolistsep]
	\item Each node $i \in V(\mathcal{T})$ corresponds to a subset of vertices $X_i$ of $V$  (called bags), such that $\cup_{i \in V(\mathcal{T})}X_i = V$.
	\item For each edge $uv \in E$, there is a bag $X_i$ containing both $u,v$.
	\item For a vertex $v \in V$, all the bags containing $v$ make up a subtree of $\mathcal{T}$.  
\end{enumerate}

The \emph{width} of a tree decomposition $\mathcal{T}$ is $\max_{ i \in V(\mathcal{T})}|X_i| -1$ and the treewidth of $G$, denoted by $\tw$, is the minimum width among all possible tree decompositions of $G$.  A \emph{path decomposition} of a graph $G(V,E)$ is a tree decomposition where the underlying tree is a path. The pathwidth of $G$, denoted by $\pw$, is defined accordingly.
	\end{definition}

	A \emph{vortex} is a  graph $W$ equipped with a pah decomposition $\{X_1,X_2,\ldots, X_t\}$ and a sequence of $t$ designated vertices $x_1,\ldots, x_t$, called the \emph{perimeter} of $W$, such that each $x_i \leq X_i$ for all $1\leq i \leq t$. The \emph{width} of the vortex is the width of its path decomposition.
	We say that a vortex $W$ is \emph{glued} to a face $F$ of a surface embedded graph $G$ if $W\cap F$ is the perimeter of $W$ whose vertices appear consecutively along the boundary of $F$. 
	
	\paragraph{Nearly $h$-embeddability}  A graph $G$ is nearly $h$-embeddable if  there is a set of at most $h$ vertices $A$, called \emph{apices}, such that $G\setminus A$ can be decomposed as $G_{\Sigma}\cup \{W_1, W_2,\ldots, W_{h}\}$ where $G_{\Sigma}$ is (cellularly) embedded on a surface $\Sigma$ of genus at most $h$ and each $W_i$ is a vortex of width at most $h$ glued to a face of $G_{\Sigma}$.

	\paragraph{$h$-Clique-sum} A graph $G$ is a $h$-clique-sum of two graphs $G_1,G_2$, denoted by $G  = G_1\oplus_h G_2$, if there are two cliques of size exactly $h$ each such that $G$ can be obtained by  identifying vertices of the two cliques and remove some clique edges of the resulting identification. 
	
	Note that clique-sum is not a well-defined operation since the clique-sum of two graphs is not unique due to the clique edge deletion step. We are ready now to state the decomposition theorem.

	\begin{theorem}[Theorem~1.3~\cite{RS03}] \label{thm:RS} There is a constant $h = O_r(1)$ such that any $K_r$-minor-free graph $G$ can be decomposed into a tree $\mathbb{T}$ where each node of $\mathbb{T}$ corresponds to a nearly $h$-embeddable graph such that $G = \cup_{X_iX_j \in E(\mathbb{T})} X_i \oplus_h X_j$.
	\end{theorem} 
	
	The graphs corresponding to the nodes in the clique-sum decomposition above are referred to as \emph{pieces}.
	Note that the pieces in $\mathbb{T}$ may not be
    subgraphs of $G$, as in the clique-sum, some edges of a node, namely some edges of a nearly $h$-embeddable subgraph associated to a node,
    may not be present in $G$. 
    We will slightly modify the graph to ensure that this never happens.
    Specifically, for any pair $u,v$ of vertices used in a clique-sum for a piece $X$ of $\mathbb{T}$, that are not present in $G$,
    we add an edge $(u,v)$ to $G$ and set its weight to be $d_G(u,v)$. 
    In the decomposition of the resulting graph, the clique-sum operation does not remove any edge.
    Note that this operation does not
    change the Robertson-Seymour decomposition of the graph, nor its shortest path metric. Thus from a metric point of view, the two graphs are equivalent.
        
    Cohen-Addad \etal \cite{CFKL20} showed that every $n$-vertex $K_r$-minor free graph has a stochastic one-to-many embedding with expected additive distortion $\eps D$ into a graph with treewidth $O(\frac{\log n}{\eps^2})$.
    The only reason \cite{CFKL20} used randomness is due to apices. The following lemma \cite{CFKL20} states that nearly $h$-embeddable graphs without apices embed deterministically into bounded treewidth graphs. We will use this embedding in a black box manner.
    \begin{restatable}[Multiple Vortices and Genus, \cite{CFKL20}]{lemma}{embedGenusVortex}
       	\label{lm:embed-genus-vortex}
       	Consider a graph $G = G_{\Sigma}\cup  W_1\cup\dots \cup W_{h}$ of diameter $D$,  where $G_{\Sigma}$ is (cellularly) embedded on a surface $\Sigma$ of genus $h$, and each $W_i$ is a vortex of width at most $h$ glued to a face of $G_{\Sigma}$. There is a one-to-many clique-preserving embedding $f$ from $G$ to a graph $H$ of treewidth at most $O_h\left(\frac{\log n}{\epsilon}\right)$ with additive distortion $\epsilon D$.
    \end{restatable}

%% file: ClanTree.tex
\section[Clan embedding into an ultrametric (\Cref{thm:ClanUltrametric})]{Clan embedding into an ultrametric}\label{sec:Ultrametric}
This section is devoted to proving \Cref{thm:ClanUltrametric}. We restate it for convenience.
\ClanUltrametric*

First, we will prove a ``distributional'' version of \Cref{thm:ClanUltrametric}. That is, we will receive a distribution $\mu$ over the points, and deterministically construct a single clan embedding $(f,\chi)$ such that $\sum_{x\in X}\mu(x)|f(x)|$ will be bounded. Later, we will use the minimax theorem to conclude \Cref{thm:ClanUltrametric}.
We begin with some definitions: a \emph{measure} over a finite set $X$, is simply a function $\mu:X\rightarrow \R_{\ge 0}$. The measure of a subset $A\subseteq X$, is $\mu(A)=\sum_{x\in A}\mu(x)$. Given some function $f:X\rightarrow \R$, it's expectation w.r.t. $\mu$ is  $\mathbb{E}_{x\sim\mu}[f]=\sum_{x\in X}\mu(x)\cdot f(x)$.
We say that $\mu$ is a \emph{probability measure} if $\mu(X)=1$. We say that $\mu$ is a \emph{$(\ge1)$-measure} if for every $x\in X$, $\mu(x)\ge1$.
\begin{lemma}\label{lem:clanTreeMeasure}
	Given an $n$-point metric space $(X,d_{X})$, $(\ge1)$-measure $\mu:X\rightarrow\mathbb{R}_{\ge1}$,
	and integer parameter $k\ge1$, there is a clan embedding $(f,\chi)$  into an ultrametric with multiplicative distortion $16k$
	such that $\mathbb{E}_{x\sim\mu}[|f(x)|]\le\mu(X)^{1+\frac1k}$.
\end{lemma}

\hspace{-18pt}\emph{Proof.}
	Our proof is inspired by Bartal's lecture notes \cite{Bar11}, who provided a deterministic construction of Ramsey trees. Specifically, \Cref{clm:ClanTreePartition} bellow is due to \cite{Bar11}. 
	\Cref{lem:clanTreeMeasure} could also be proved using the techniques of Abraham {\etal} \cite{ACEFN20} (and indeed we will use their approach for our clan embedding into a spanning tree, see \Cref{lem:SpanningClanTreeMeasure}); however the proof based on \cite{Bar11} we present here is shorter. 
	For a subset $A\subseteq X$, denote by $B_{A}(x,r)\coloneqq B_X(x,r)\cap A$ the ball in the metric space $(X,d_X)$ restricted to $A$.
	Set $\mu^{*}(A)\coloneqq \max_{x\in A}\mu\left(B_{A}(x,\frac{\diam(A)}{4})\right)$.
	Note that $\mu^{*}$ is monotone: i.e. $A'\subseteq A$ implies $\mu^{*}(A')\le \mu^{*}(A)$, and $\forall A,$
	$\mu^{*}(A)\le\mu(A)$. The following claim is crucial for our construction;
	its proof appears below. See \Cref{fig:UltrametricClanEmbedding} for an illustration of the claim.
	\begin{claim}\label{clm:ClanTreePartition}
		\sloppy There is a point $v\in X$ and radius $R\in(0,\frac{\diam(X)}{2}]$,
		such that the sets ${P=B_{X}(v,R+\frac{1}{8k}\cdot\diam(X))}$,
		$Q=B_{X}(v,R)$, and $\bar{Q}=X\setminus Q$ satisfy $\mu(P)\le\mu(Q)\cdot\left(\frac{\mu^{*}(X)}{\mu^{*}(P)}\right)^{\frac{1}{k}}$.
	\end{claim}

	The construction of the embedding is by induction on $n$, the number
	of points in the metric space. We assume that for a metric space $X$
	with strictly less than $n$ points, and arbitrary $(\ge1)$-measure $\mu$, we can construct a clan  embedding $(f,\chi)$ with distortion $16k$, such that $\mathbb{E}_{x\sim\mu}[|f(x)|]\le\mu(X)\mu^{*}(X)^{\frac1k}\le\mu(X)^{1+\frac1k}$.
	Find sets $P,Q,\bar{Q}\subseteq X$ using \Cref{clm:ClanTreePartition}.
	Let $\mu_{P}$ (resp. $\mu_{\bar{Q}}$) be the $(\ge1)$-measure $\mu$ restricted
	to $P$ (resp. $\bar{Q}$). 
	Using the induction hypothesis, construct clan embeddings
	$(f_{P},\chi_P)$ for $P$, and $(f_{\bar{Q}},\chi_{\bar{Q}})$ for $\bar{Q}$ into ultra-metrics $U_P,U_{\bar{Q}}$ respectively.
	Construct a new ultrametric $U$ by combining $U_P$ and $U_{\bar{Q}}$ by adding a new root node $r_U$ with label $\diam(X)$ and making roots of $U_P$ and $U_{\bar{Q}}$ children of $r_U$. For every $x\in X$ set $f(x)=f_P(x)\cup f_{\bar{Q}}(x)$. 
	If $d_X(v,x)\le R+\frac{1}{16k}\cdot\diam(X)$ set $\chi(x)=\chi_P(x)$, otherwise set  $\chi(x)=\chi_{\bar{Q}}(x)$. This finishes the construction; see \Cref{fig:UltrametricClanEmbedding} for an illustration. 
\begin{figure}[t]
	\centering
	\hspace{20pt}
	\begin{minipage}[b]{0.26\textwidth}
		\includegraphics[width=\textwidth]{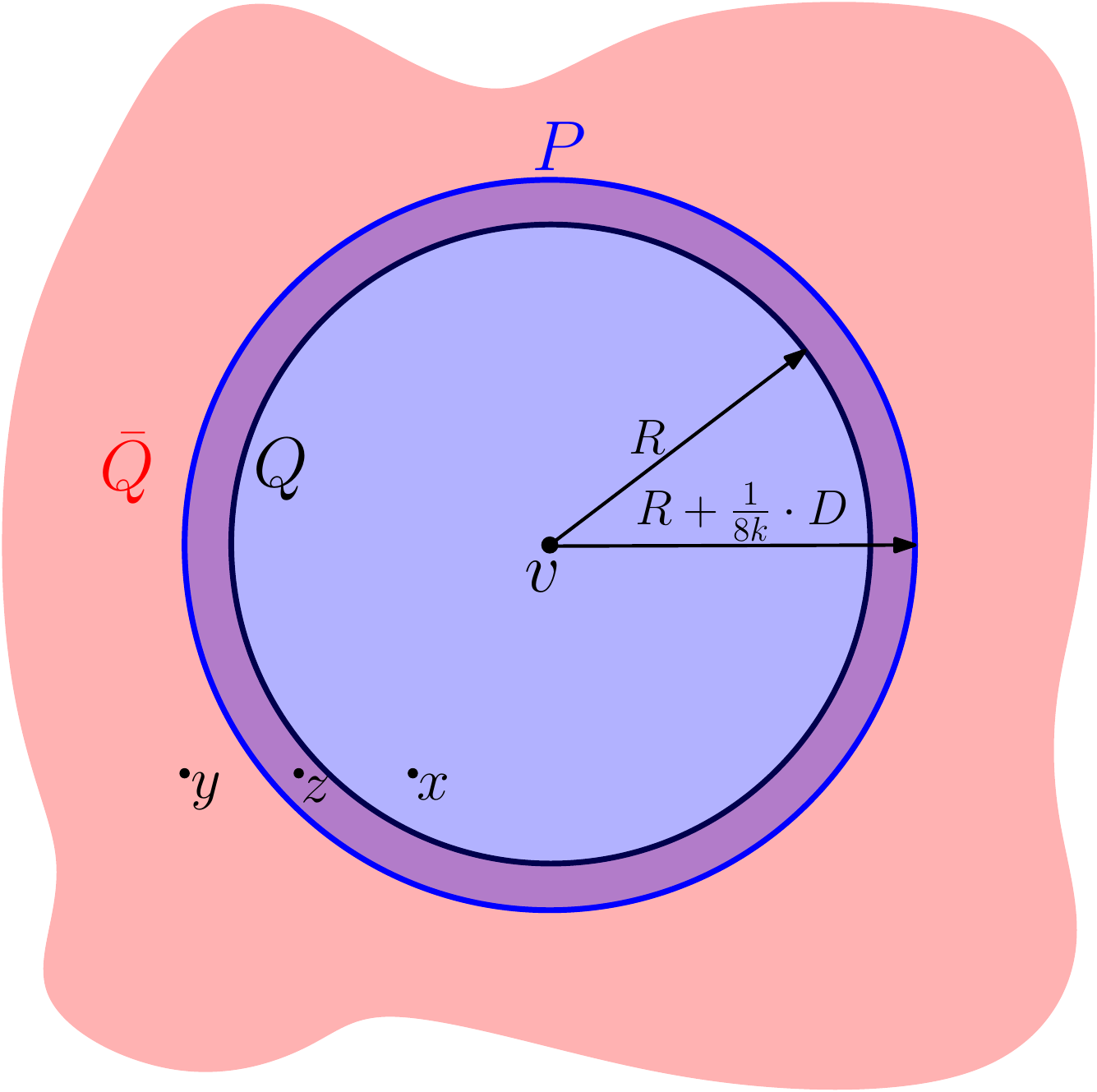}
	\end{minipage}
	\hspace{40pt}
	\begin{minipage}[b]{0.4\textwidth}
		\includegraphics[width=\textwidth]{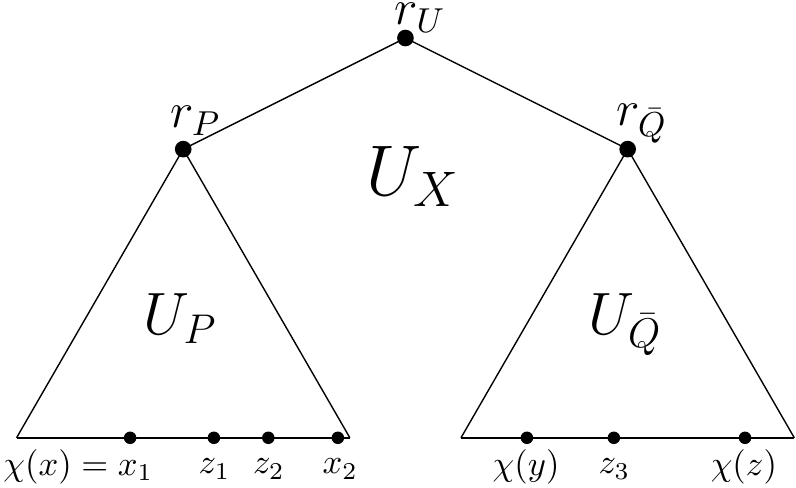}		
	\end{minipage}
	\begin{minipage}[b]{0.1\textwidth}\end{minipage}
	{\caption{\small
		On the left illustrated the clusters $P,Q,\bar{Q}$ from \Cref{clm:ClanTreePartition}.
		On the right we illustrate the clan embedding of the metric space $(X,d_X)$ into ultrametric $U$. 
		$r_U$ is the root of $U$, and its children are the roots of the ultrametrics  $U_P,U_{\bar{Q}}$ which were constructed recursively.
		The point $x\in P\cap Q$ has $f(x)=f_P(x)$ and $\chi(x)=\chi_P(x)$ (where $|f(x)|=2$). 
		The point $y$ is in $\bar{Q}\setminus P$ and thus $f(y)=f_{\bar{Q}}(y)$ and $\chi(y)=\chi_{\bar{Q}}(y)$ (there is a single copy of $y$). The point $z$ belongs to $P\cap\bar{Q}$, where $d_X(v,z)> R+\frac{1}{16}\cdot\diam(X)$, hence $f(z)=f_{P}(z)\cup f_{\bar{Q}}(z)$ and $\chi(z)=\chi_{\bar{Q}}(z)$. Note that $|f_P(z)|=|f_{\bar{Q}}(z)|=2$, and hence $|f(z)|=4$.
	}
	\label{fig:UltrametricClanEmbedding}}
\end{figure}

	Next, we argue that the clan embedding $(f,\chi)$ has multiplicative distortion $16k$. Consider a pair of points $x,y\in X$. We will show that $\min_{y'\in f(y)}d_{U}(y',\chi(x))\le16k\cdot d_{X}(x,y)$.
	Suppose first that $d_X(v,x)\le R+\frac{1}{16k}\cdot\diam(X)$. If $y\in P$, then by the induction hypothesis
	\[
	\min_{y'\in f(y)}d_{U}(y',\chi(x))\le\min_{y'\in f_{P}(y)}d_{U_{P}}(y',\chi_{P}(x))\le16k\cdot d_{P}(x,y)=16k\cdot d_{X}(x,y)~.
	\]
	Else, $y\notin P$, then $d_X(v,y)> R+\frac{1}{8k}\cdot\diam(X)$. Using the triangle inequality $d_{X}(x,y)\ge d_{X}(v,y)-d_{X}(v,x)\ge\frac{\diam(X)}{16}$.
	Note that the label of $r_U$ is $\diam(X)$, implying that $\min_{y'\in f(y)}d_{U}(y',\chi(x))\le\diam(X)\le16\cdot d_{X}(x,y)$.
	The case where  $d_X(v,x)> R+\frac{1}{16k}\cdot\diam(X)$ is symmetric (using $\bar{Q}$ instead of $P$).
	
	Next, we bound the weighted number of leafs in the ultrametric. Note that the process is deterministic and there is no probability involved. Using the induction hypothesis, it holds that 
	\begin{align*}
	\mathbb{E}_{x\sim\mu}[|f(x)|] & =\sum_{x\in X}\mu(x)\cdot\left(|f_{P}(x)|+|f_{\bar{Q}}(x)|\right)\\
	& =\mathbb{E}_{x\sim\mu_{P}}[|f_{P}(x)|]+\mathbb{E}_{x\sim\mu_{\bar{Q}}}[|f_{\bar{Q}}(x)|]\\
	& \le\mu_{P}(P)\mu_{P}^{*}(P)^{\frac{1}{k}}+\mu_{\bar{Q}}(\bar{Q})\mu_{\bar{Q}}^{*}(\bar{Q})^{\frac{1}{k}}\\
	& \le\mu(P)\mu^{*}(P)^{\frac{1}{k}}+\mu(\bar{Q})\mu^{*}(\bar{Q})^{\frac{1}{k}}\\
	& \overset{(*)}{\le}\mu(Q)\mu^{*}(X)^{\frac{1}{k}}+\mu(\bar{Q})\mu^{*}(X)^{\frac{1}{k}}=\mu(X)\mu^{*}(X)^{\frac{1}{k}}~,
	\end{align*}
	where in the inequality $(*)$ is due to \Cref{clm:ClanTreePartition} and the fact that $\mu^{*}(\bar{Q})\le\mu^{*}(X)$.
	\qed

\begin{proof}[Proof of \Cref{clm:ClanTreePartition}]
	Let $v$ be the point minimizing the ratio $\frac{\mu\left(B_{X}(v,\frac{\diam(X)}{4})\right)}{\mu\left(B_{X}(v,\frac{\diam(X)}{8})\right)}$.
	Set $\rho=\frac{\diam(X)}{8k}$, and for $i\in[0,k]$
	let $Q_{i}=B_{X}(v,\frac{\diam(X)}{8}+i\cdot\rho)$. Let $i\in[0,k-1]$
	be the index minimizing $\frac{\mu(Q_{i+1})}{\mu(Q_{i})}$. Then, 
	\[
	\left(\frac{\mu(Q_{k})}{\mu(Q_{0})}\right)^{\frac1k}=\left(\frac{\mu(Q_{1})}{\mu(Q_{0})}\cdot\frac{\mu(Q_{2})}{\mu(Q_{1})}\cdots\frac{\mu(Q_{k})}{\mu(Q_{k-1})}\right)^{\frac1k}\ge\left(\frac{\mu(Q_{i+1})}{\mu(Q_{i})}\right)^{k\cdot\frac{1}{k}}=\frac{\mu(Q_{i+1})}{\mu(Q_{i})}~.
	\]
	\sloppy Set $R=\frac{\diam(X)}{8}+i\cdot\rho$, then $P=B_{X}(v,R+\rho)$,
	$Q=B_{X}(v,R)$, $\bar{Q}=X\setminus Q$. Note that $\diam(P)\le2\cdot(\frac{\diam(X)}{8}+k\cdot\rho)=\frac{\diam(X)}{2}$.
	Let $u_{P}$ be the point defining $\mu^{*}(P)$, that is $\mu^{*}(P)	=\mu\left(B_{P}(u_{P},\frac{\diam(P)}{4}\right)\le\mu\left(B_{P}(u_{P},\frac{\diam(X)}{8}\right)$.
	Using the minimality of $v$, it holds that
	\[
	\frac{\mu(P)}{\mu(Q)}\le\left(\frac{\mu(Q_{k})}{\mu(Q_{0})}\right)^{\frac{1}{k}}=\left(\frac{\mu\left(B_{X}(v,\frac{\diam(X)}{4})\right)}{\mu\left(B_{X}(v,\frac{\diam(X)}{8})\right)}\right)^{\frac{1}{k}} \stackrel{(*)}{\le} \left(\frac{\mu\left(B_{X}(u_{P},\frac{\diam(X)}{4})\right)}{\mu\left(B_{X}(u_{P},\frac{\diam(X)}{8})\right)}\right)^{\frac{1}{k}}\le\left(\frac{\mu^{*}\left(X\right)}{\mu^{*}\left(P\right)}\right)^{\frac{1}{k}}~.
	\]
where $(*)$ is due to the choice of $v$.
\end{proof}

Next, we translate the language of $(\ge1)$-measures used in \Cref{lem:clanTreeMeasure} to probability measures:
\begin{lemma}\label{lem:clanTreeProbability}
	Given an $n$-point metric space $(X,d_{X})$, and probability measure
	$\mu:X\rightarrow\mathbb{R}_{\ge0}$, we can construct the two following clan embeddings $(f,\chi)$ into ultrametrics:
	\begin{enumerate}
		\item For every parameter $k\ge 1$, multiplicative distortion $16k$ such that $\mathbb{E}_{x\sim\mu}[|f(x)|]\le O(n^{\frac1k})$.
		\item For every parameter  $\epsilon\in(0,1]$, multiplicative distortion $O(\frac{\log n}{\eps})$ such that $\mathbb{E}_{x\sim\mu}[|f(x)|]\le1+\epsilon$.
	\end{enumerate}
\end{lemma}
\begin{proof}	
	We define the following probability
	measure $\widetilde{\mu}$: $\forall x\in X$, $\widetilde{\mu}(x)=\frac{1}{2n}+\frac{1}{2}\mu(x)$.
	Set the following $(\ge1)$-measure $\widetilde{\mu}_{\ge1}(x)=2n\cdot \tilde{\mu}(x)$. Note that $\widetilde{\mu}_{\ge1}(X)=2n$.
	We execute \Cref{lem:clanTreeMeasure} w.r.t. the $(\ge1)$-measure $\widetilde{\mu}_{\ge1}$, and parameter $\frac1\delta\in\N$ to be determined later. It holds that
	\[
	\widetilde{\mu}_{\ge1}(X)\cdot\mathbb{E}_{x\sim\widetilde{\mu}}[|f(x)|]=\mathbb{E}_{x\sim\widetilde{\mu}_{\ge1}}[|f(x)|]\le\widetilde{\mu}_{\ge1}(X)^{1+\delta}=\widetilde{\mu}_{\ge1}(X)\cdot(2n)^{\delta}~,
	\]
	implying 
	\[
	(2n)^{\delta}\ge\mathbb{E}_{x\sim\widetilde{\mu}}[|f(x)|]=\frac{1}{2}\cdot\mathbb{E}_{x\sim\mu}[|f(x)|]+\frac{\sum_{x\in X}|f(x)|}{2n}\ge\frac{1}{2}\cdot\mathbb{E}_{x\sim\mu}[|f(x)|]+\frac{1}{2}~.
	\]
	\begin{enumerate}
		\item  Set
		$\delta=\frac1k$, then we have multiplicative distortion $\frac{16}{\delta}=16k$,
		and $\mathbb{E}_{x\sim\mu}[|f(x)|]\le2\cdot(2n)^{\delta}=O(n^{\frac{1}{k}})$.
		\item Choose $\delta\in(0,1]$ such that $\frac{1}{\delta}=\left\lceil \frac{\ln(2n)}{\ln(1+\epsilon/2)}\right\rceil $,
		note that $\delta\le\frac{\ln(1+\epsilon/2)}{\ln(2n)}$.
		Then we have multiplicative distortion $O(\frac{1}{\delta})=O(\frac{\log n}{\eps})$, and		
		$\mathbb{E}_{x\sim\mu}[|f(x)|]\le2\cdot(2n)^{\delta}-1\le1+\eps$.		
	\end{enumerate}
\end{proof}
\begin{remark}\label{remark:MaxSizeClanEmbedding}
	\Cref{lem:clanTreeProbability}, note that for the clan embedding $(f,\chi)$ returned by \Cref{lem:clanTreeProbability} for input $k$, it holds that $|f(X)|\le\tilde{\mu}_{\ge1}(X)^{1+\frac{1}{k}}=(2n)^{1+\frac{1}{k}}$.
	In particular, every $x\in X$ has at most $(2n)^{1+\frac{1}{k}}$ copies.
	Similarly, for input $\eps$, $|f(X)|\le\tilde{\mu}_{\ge1}(X)^{1+\delta}\le(2n)^{1+\frac{\ln(1+\eps/2)}{\ln2n}}=2n\cdot(1+\frac{\epsilon}{2})$. As for every $y\in X$, $f(y)\ne\emptyset$, it follows that for every $x\in X$, its number of copies is bounded by $|f(x)|=|f(X)|-|f(X\setminus\{x\})|\le2n\cdot(1+\frac{\epsilon}{2})-(n-1)=(1+\epsilon)n+1$.	
\end{remark}

Using the minimax theorem, as shown bellow, we show that there exists a distribution $\mathcal{D}$ of clan embeddings with distortion and expected clan size as specified by \Cref{thm:ClanUltrametric}.
Afterwards, in \Cref{subsec:clan-MWU}, using the multiplicative weights update (MWU) method, we explicitly construct such distributions efficiently, and with small support size.

\begin{proof}[Proof of \Cref{thm:ClanUltrametric} (exsistential agrument)]
	Let $\mu$ be an arbitrary probability measure over the points, and
	$\mathcal{D}$ be any distribution over clan embeddings $(f,\chi)$ of $(X,d_{X})$
	intro trees with multiplicative distortion $O(\frac{\log n}{\epsilon})$.
	Using \Cref{lem:clanTreeProbability} and the minimax theorem we have that
	\[
	\min_{\mathcal{D}}\max_{\mu}\mathbb{E}_{(f,\chi)\sim\mathcal{D},x\sim\mu}[|f(x)|]=\max_{\mu}\min_{(f,\chi)}\mathbb{E}_{x\sim\mu}[|f(x)|]\le1+\epsilon~.
	\]
	Let $\mathcal{D}$ be the distribution from above, denote by $\mu_{z}$
	the probability measure where $\mu_z(z)=1$ (and $\mu_z(y)=0$ for $y\ne z$). Then for every $x\in X$
	\[
	\mathbb{E}_{(f,\chi)\sim\mathcal{D}}[|f(z)|]=\mathbb{E}_{(f,\chi)\sim\mathcal{D},x\sim\mu_{z}}[|f(x)|]\le\max_{\mu}\mathbb{E}_{(f,\chi)\sim\mathcal{D},x\sim\mu}[|f(x)|]\le1+\epsilon~.
	\]
	
	The second claim of \Cref{thm:ClanUltrametric} could be proven using exactly  the same argument.
\end{proof}

\subsection{Constructive Proof of \Cref{thm:ClanUltrametric}
}\label{subsec:clan-MWU}

In this section, we efficiently construct a uniform distribution $\mathcal{D}$ as stated in \Cref{thm:ClanUltrametric}. Our construction relies on the multiplicative weights update method (MWU) \footnote{For an excellent introduction of the MWU method and its historical account, see the survey by Arora, Hazan and Kale~\cite{AHK12}.} and the notion of a $(\rho,\alpha,\beta)$-bounded \textsc{Oracle}.

\begin{definition}[$(\rho,\alpha,\beta)$-bounded \textsc{Oracle}]\label{def:Oracle} Given a probability measure $\mu$ over the metric points, a $(\rho,\alpha,\beta)$-bounded \textsc{Oracle} returns a clan embedding $(f,\chi)$ with multiplicative distortion $\beta$ such that:
	\begin{enumerate}
		\item $\mathbb{E}_{x\sim\mu}[|f(x)|] \leq \alpha$.
		\item $\max_{x\in V}|f(x)| \leq \rho$. 
	\end{enumerate}
\end{definition}

 In \Cref{lem:MWU-Distribution} below, we show that one can construct a uniform distribution $\mathcal{D}$ by making a polynomial number of oracle calls.

\begin{lemma}\label{lem:MWU-Distribution} Given a $(\rho,\alpha,\beta)$-bounded \textsc{Oracle}, and parameter $\eps\in(0,\frac12)$ one can construct a uniform distribution $\mathcal{D}$ over $O(\frac{\rho \alpha \log(n)}{\epsilon^2})$  clan embeddings with multiplicative distortion $\beta$ such that:
		\begin{equation*}
			\mbox{For every }x\in X,~~	\mathbb{E}_{(f,\chi)\sim\mathcal{D}}[|f(x)|] \leq \alpha + \epsilon
		\end{equation*}
Furthermore, the construction only makes $O(\frac{\rho \alpha \log(n)}{\epsilon^2})$ queries to the $(\rho,\alpha,\beta)$-bounded \textsc{Oracle}. 
\end{lemma} 
\begin{proof}
	Let $\mathcal{O}$ be a $(\rho,\alpha,\beta)$-bounded \textsc{Oracle} and $\mathcal{O}(\mu)$ be the clan embedding returned by the oracle given a probability measure $\mu$. 
	We follow the standard set up of MWU: we have $n$ ``experts" where the $i$-th expert is associated with the $i$-th point $x_i \in X$. The construction happens in $T$ rounds. At the beginning of round $t$, we have a weight vector $\mathbf{w}^{t} =  (w_1^{t}, \ldots, w_n^{t})^\intercal$; at the first round, $\mathbf{w}^{1} = (1,1,\ldots, 1)^\intercal$. 
	
	The weight vector $\mathbf{w}^{t}$ induces a probability measure $\mu^{t} = (\frac{w^{t}_1}{W^{t}}, \ldots, \frac{w^t_n}{W^t})$ where $W^t  = \sum_{i=1}^n w^{t}_i$. We  construct a clan embedding $(f^t, \chi^t) = \mathcal{O}(\mu^t)$  by making an oracle call to $\mathcal{O}$ with $\mu_t$ as input. Let $g^{t}_i = \frac{|f^t(x_i)|}{\rho}$, and $\mathbf{g}^{t} = (g^t_1, \ldots, g^t_n)^\intercal$ be the ``penalty'' %\atodo{I changed reward to penalty, I think it is more intuitive} 
	vector for the  set of $n$ points (or experts). We then update:
	\begin{equation}\label{eq:update-MWU}
		w^{t+1}_i = (1+\delta)^{g^t_i}w^t_i \qquad \forall x_i \in X,
	\end{equation}
	for some small parameter $\delta$ chosen later. 

	The penalty for each additional copy of each point is proportional to the number of copies it has in the clan embeddings constructed in previous steps. This is because in the next round, we will increase the measure of points with a large number of copies. Hence the oracle will be ``motivated'' to reduce the number of copies of these points in the next outputted clan embedding.

	After $T$ rounds, we have a collection $\mathcal{D}_T = \{(f^1,\chi^1), \ldots, (f^T,\chi^T)\}$ of $T$ clan embeddings. The distribution $\mathcal{D}$ is constructed by  sampling an embedding from $\mathcal{D}_T$ uniformly at random. Note that the distortion bound follows directly from the fact that the distortion of every clan embedding returned by the oracle is $\beta$.	Our goal is to show that, by  setting $T = O(\frac{\rho \log(n)}{\epsilon^2})$, we have:
	
	\begin{equation}\label{eq:MWU-expected-size}
		\frac{1}{T}\cdot\sum_{t=1}^T |f^t(x_i)| ~\leq~ \alpha+\epsilon \qquad \forall x_i \in V
	\end{equation}	 
	 To that end, we first observe that:
	 \[
	 W^{t+1}~=~\sum_{i=1}^{n}w_{i}^{t+1}~=~\sum_{i=1}^{n}(1+\delta)^{g_{i}^{t}}w_{i}^{t}~\stackrel{(*)}{\leq}~\sum_{i=1}^{n}(1+\delta g_{i}^{t})w_{i}^{t}~=~(1+\sum_{i=1}^{n}\delta g_{i}^{t}\mu_{i}^{t})W^{t}~\leq~ e^{\delta\langle\mathbf{g}^{t},\mu^{t}\rangle}W^{t}
	 \]
where inequality $(*)$ follows from that $(1+x)^r \leq (1+rx)$ for any $x\geq 0$ and $r\in [0,1]$. 
Thus, we have:

\begin{equation}\label{eq:MWU-WT1}
	W^{T+1} ~\leq~ e^{\delta \sum_{t=1}^T \langle \mathbf{g}^t, \mu^t\rangle} W^1 ~=~ e^{\delta \sum_{t=1}^T \langle \mathbf{g}^t, \mu^t\rangle} n
\end{equation}
Observe that $W^{T+1} \geq w^{T+1}_i = (1+\delta)^{\sum_{t=1}^T g^t_i}w^1_i = (1+\delta)^{\sum_{t=1}^T g^t_i}$ and that:
\[
\sum_{t=1}^{T}\langle\mathbf{g}^{t},\mu^{t}\rangle~=~\sum_{t=1}^{T}\sum_{x\in X}\frac{|f^{t}(x)|}{\rho}\cdot\mu^{t}(v_{i})~=~\frac{1}{\rho}\cdot\sum_{t=1}^{T}\mathbb{E}_{x\sim\mu^{t}}[|f^{t}(x)|]~\leq~\frac{T\alpha}{\rho}~.
\]
Thus, by equation (\ref{eq:MWU-WT1}), it holds that:
\begin{align}\label{eq:MWU-gi-vs-overall}
	(1+\delta)^{\sum_{t=1}^T g^t_i} ~\leq~   e^{\frac{\delta T \alpha}{\rho}} n~.
\end{align}

Taking the natural logarithm from both sides we obtain that $\frac{\delta T\alpha}{\rho}+\ln n\ge\sum_{t=1}^{T}g_{i}^{t}\cdot\ln(1+\delta)=\frac{\ln(1+\delta)}{\rho}\cdot\sum_{t=1}^{T}|f^{t}(x_{i})|$,
and thus
\[
\frac{1}{T}\cdot\sum_{t=1}^{T}|f^{t}(x_{i})|~\leq~\frac{\rho}{T\cdot\ln(1+\delta)}\cdot\left(\frac{\delta T\alpha}{\rho}+\ln n\right)~=~\frac{\delta\alpha}{\ln(1+\delta)}+\frac{\rho\cdot\ln n}{T\cdot\ln(1+\delta)}~\leq~\alpha(1+\frac{\delta}{2})+\frac{2\rho\cdot\ln n}{T\cdot\delta}~,
\]
where the last inequality follows as $\frac{\delta}{\ln(1+\delta)}\leq(1+\frac{\delta}{2})$
and $\ln(1+\delta)\geq\frac{\delta}{2}$ for $\delta\in(0,\frac{1}{2})$.
By choosing $T=\frac{4\rho\alpha\ln n}{\epsilon^{2}}=O(\frac{\rho\alpha\log n}{\epsilon^{2}})$
and $\delta=\sqrt{\frac{4\rho\ln n}{T\alpha}}=\sqrt{\frac{\epsilon^{2}}{\alpha^{2}}}=\frac{\epsilon}{\alpha}<\frac{1}{2}$, we obtain that
\[
\frac{1}{T}\cdot\sum_{t=1}^{T}|f^{t}(x_{i})|~\leq~\alpha+\frac{\delta\cdot\alpha}{2}+\frac{\epsilon^{2}}{2\alpha\cdot\delta}~=~\alpha+\epsilon~,
\]
satisfying equation (\ref{eq:MWU-expected-size}), which completes our proof.
\end{proof}

Observe that \Cref{lem:clanTreeProbability}, combined with \Cref{remark:MaxSizeClanEmbedding}, provides an 
$(O(n), 1+\frac{\epsilon}{2}, O(\frac{\log n}{\epsilon}))$-bounded \textsc{Oracle} (when we apply \Cref{lem:clanTreeProbability} with parameter $\frac\eps2$).
Using \Cref{lem:MWU-Distribution} with parameter $\frac\eps2$  provides us with an efficiently computable distribution over clan embeddings with support size $O(\frac{n\log n}{\eps^2})$, distortion $O(\frac{\log n}{\eps})$, and such that for every $x\in X$, $\mathbb{E}_{(f,\chi)\sim\mathcal{D}}[|f(x)|] \leq 1 + \epsilon$.

Similarly, by applying \Cref{lem:clanTreeProbability} with parameter $k$, we get an 
$(O(n^{1+\frac{1}{k}}), O(n^{\frac{1}{k}}), 16k)$-bounded \textsc{Oracle}. Thus \Cref{lem:MWU-Distribution} will produce an efficiently computable distribution over clan embeddings with support size $O(n^{1+\frac{2}{k}}\log n)$, distortion $16k$, and such that for every $x\in X$, $\mathbb{E}_{(f,\chi)\sim\mathcal{D}}[|f(x)|] = O(n^{\frac1k})$. 
\Cref{thm:ClanUltrametric} now follows.

%% file: SpanningClantree.tex
\section[Clan Embedding into a Spanning Tree (\Cref{thm:ClanSpanningTree})]{Clan Embedding into a Spanning Tree}\label{sec:SpanningClan}
This section is devoted to proving \Cref{thm:ClanSpanningTree}. We restate it for convenience.

\ClanSpanningTree*

In this section, we construct spanning clan embeddings into trees. We will use the framework of petal decomposition proposed by Abraham and Neiman \cite{AN19}, who originally used it to construct a stochastic embedding of a graph into spanning trees with a bounded expected distortion. The framework was also previously used by Abraham \etal \cite{ACEFN20} to construct Ramsey spanning trees.
The petal decomposition is an iterative method to build a spanning tree of a given graph. At each level, the current graph is partitioned into smaller diameter pieces (called \emph{petals}), and a single central piece (called \emph{stigma}), which are then connected by edges in a tree structure. Each of the petals is a ball in a certain cone metric. When creating a petal from a cluster of diameter $\Delta$, one has the freedom to choose a radius from an interval of length $\Omega(\Delta)$. The crucial property is that, regardless of the radii choices during the execution of the algorithm, the framework guarantees that the diameter of the resulting tree will be $O(\Delta)$.

However, as we are constructing a clan embedding rather than a classical one, some vertices will have multiple copies. As a result, some mild changes will be introduced to the construction of \cite{AN19}. 
Once we establish the petal decomposition framework for clan embeddings, the proof of \Cref{thm:ClanSpanningTree} will follow the lines similar to \Cref{thm:ClanUltrametric}. The additional $\log\log n$ factor is a phenomenon also appearing in previous uses of the petal decomposition framework \cite{AN19,ACEFN20}. The reason is that, while similar embeddings into ultrametrics create clusters by growing balls around smartly chosen centers  (e.g. \cite{Bartal04,Bar11} and \Cref{thm:ClanUltrametric}), in the petal decomposition framework, we lack the freedom to choose the center of the petal.

\paragraph{Organization:} In \Cref{subsec:PetalDecompDesc}, we  describe the petal decomposition framework in general.
In \Cref{subsec:createPetal}, we describe our specific usage of it, i.e. the algorithm choosing the radii (with some leftovers in \Cref{subsec:MissingCreatePetal}).
Then in \Cref{subsec:SpanningClanCorrectness}, we prove \Cref{lem:SpanningClanTreeMeasure}, that appears below. \Cref{lem:SpanningClanTreeMeasure} is a ``distributional'' version of \Cref{thm:ClanSpanningTree}, and has a role parallel to \Cref{lem:clanTreeMeasure} in \Cref{sec:Ultrametric}. Finally, in \Cref{subsec:ProofOfSpanningTree}, we will deduce \Cref{thm:ClanSpanningTree} using \Cref{lem:SpanningClanTreeMeasure}.
\begin{lemma}\label{lem:SpanningClanTreeMeasure}
	Given an $n$-vertex weighted graph $G=(V,E,w)$, $(\ge1)$-measure $\mu:V\rightarrow\mathbb{R}_{\ge1}$,
	and integer parameter $k\ge1$, there is a spanning clan embedding $(f,\chi)$  into a tree with multiplicative distortion $O(k\log\log \mu(V))$
	such that $\mathbb{E}_{v\sim\mu}[|f(v)|]\le\mu(V)^{1+\frac1k}$.
\end{lemma}

\subsection{Petal Decomposition Framework}\label{subsec:PetalDecompDesc}
We begin with some notations specific to this section.
For a subset $S\subseteq G$ and a center vertex $x_0\in S$, the radius of $S$ w.r.t $x_0$, $\Delta_{x_0}(S)$, is the minimal $\Delta$ such that $B_{G[S]}(x_0,\Delta)=S$.
(If for every $\Delta$, $B_{G[S]}(x_0,\Delta)\neq S$ --- this can happen iff $G[S]$ is not connected --- we say that $\Delta_{x_0}(S)=\infty$.)
When the center $x_0$ is clear from the context or is not relevant, we will omit it.
Given two vertices $u,v$, $P_{u,v}(X)$ denotes the shortest path between them in $G[X]$, the graph induced by $X$ (we will assume that every pair has a unique shortest path; this can be arranged by tiny perturbation of the edge weights.). 

\begin{wrapfigure}{r}{0.26\textwidth}
	\begin{center}
		\vspace{-30pt}
		\includegraphics[scale=0.42]{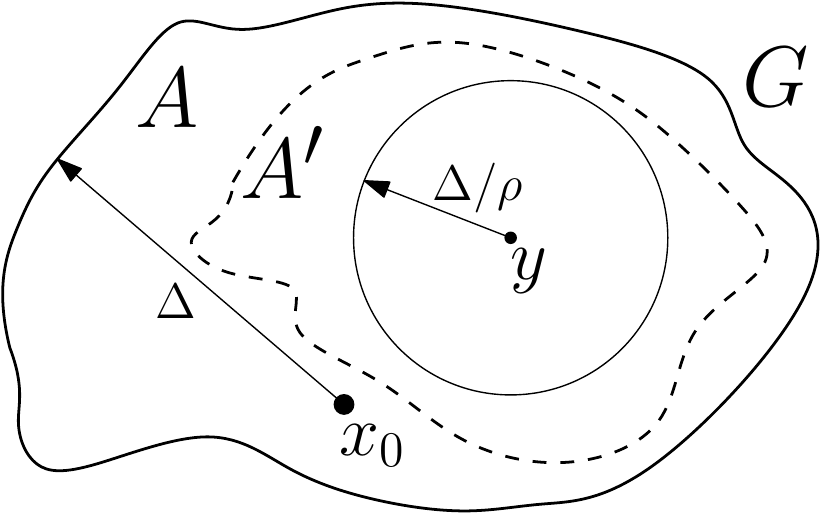}
		\vspace{-5pt}
	\end{center}
	\vspace{-10pt}
\end{wrapfigure}
Given a graph $G=(V,E,w)$ and a cluster $A\subseteq V$ (with center $x_0$), we say that a vertex $y\in A$ is {\em $\rho$-padded} by the cluster $A'\subseteq A$ (w.r.t $A$) if $B(y,\Delta_{x_0}(A)/\rho,G)\subseteq A'$. See an illustration on the right.

\sloppy Next, we provide a concise description of the petal decomposition algorithm, focusing on the main properties we will use.
For proofs and further details, we refer readers to \cite{AN19}.
The presentation here differs slightly from \cite{AN19} as our goal is to construct a spanning clan embedding into a tree rather than a classic one. However, the changes are straightforward, and no new ideas are required. 

The \texttt{hierarchical-petal-decomposition} (see \Cref{alg:h-petal}) is a recursive algorithm. The input is $G[X]$ (a graph $G=(V,E,w)$ induced over a set of vertices $X\subseteq V$), a center $x_0\in X$, a target $t\in X$, and the radius $\Delta=\Delta_{x_0}(X)$.\footnote{Rather than inferring $\Delta=\Delta_{x_0}(X)$ from $G[X]$ and $x_0$ as in \cite{AN19}, we will follow \cite{ACEFN20} and think of $\Delta$ as part of the input. We shall allow any $\Delta\ge\Delta_{x_0}(X)$. We stress that, in fact, in the algorithm, we always use $\Delta_{x_0}(X)$, and consider this degree of freedom only in the analysis.}  
The algorithm invokes the \texttt{petal-decomposition} procedure to create clusters $\widetilde{X}_0,\widetilde{X}_1,\dots,\widetilde{X}_s$ of $X$ (for some integer $s$), and also provides a set of edges $\{(x_1,y_1),\dots,(x_s,y_s)\}$ and targets $t_0,t_1,\dots,t_s$.
The \texttt{hierarchical-petal-decomposition} algorithm now recurses on each $(G[\widetilde{X}_j],x_j,t_j,\Delta_{x_j}(\widetilde{X}_j))$ for $0\le j\le s$, to get trees $\{T_j\}_{0\le j\le s}$ (and clan embeddings $\{(f_j,\chi_j)\}_{0\le j\le s}$), which are then connected by the edges $\{(x_j,y_j)\}_{1\le j\le s}$ to form a tree $T$ (the recursion ends when $X_j$ is a singleton).
The one-to-many embedding $f$ simply defined as the union of the one-to-many embeddings $\{f_j\}_{0\le j\le s}$.
Note, however, that the clusters  $\widetilde{X}_0,\widetilde{X}_1,\dots,\widetilde{X}_s$ are not disjoint. Therefore, in addition, for each cluster $\widetilde{X}_j$ the \texttt{petal-decomposition} procedure will also provide us with sub-clusters $\underline{X}_j\subseteq X_j\subseteq \widetilde{X}_j$ that will be used to determine the chiefs (i.e. $\chi$ part) of the clan embedding.

\begin{algorithm}[t]
	\caption{$(T,f,\chi)=\texttt{hierarchical-petal-decomposition}
		(G[X],x_{0},t,\Delta)$}	\label{alg:h-petal}
	\DontPrintSemicolon
	\If{$|X|=1$}{\Return $G[X]$}
	Let $\left(\left\{ \underline{X}_j,X_j,\widetilde{X}_j,x_j,t_{j},\Delta_{j}\right\} _{j=0}^{s},\left\{ (y_{j},x_{j})\right\} _{j=1}^{s}\right)=\texttt{petal-decomposition}(G[X],x_{0},t,\Delta)$\;
	\For{each $j\in[0,\dots,s]$}{
		$(T_{j},f_j,\chi_j)=\texttt{hierarchical-petal-decomposition}$ $(G[\widetilde{X}_j],x_{j},t_{j},\Delta_j)$\;
	}
	\For{each $z\in X$}{
		Set $f(x)=\cup_{j=0}^sf_j(z)$\;
		\If{$\exists j>0$ such that $z\in X_j$}{
			Let $j>0$ be the minimal index such that $z\in X_j$. Set $\chi(z)=\chi_j(z)$\;}
		\Else{Set $\chi(z)=\chi_0(z)$\;}
	}
	Let $T$ be the tree formed by connecting $T_{0},\dots,T_{s}$ using the edges $\{\chi(y_{1}),\chi(x_{1})\},\dots,\{\chi(y_{s}),\chi(x_{s})\}$\;
	\Return $(T,f,\chi)$\;
\end{algorithm}

Next, we describe the \texttt{petal-decomposition} procedure (see \Cref{alg:petal-d}).
Initially it sets $Y_0=X$, and for $j=1,2,\dots,s$, it carves out the petal $\widetilde{X}_j$ from the graph induced on $Y_{j-1}$, and sets $Y_j=Y_{j-1}\backslash \underline{X}_j$, where $\underline{X}_j$ is a sub-petal of $\widetilde{X}_j$, consisting of all the vertices which are padded by $\widetilde{X}_j$.
The idea is that $Y_j$ is defined w.r.t. to a smaller set than the petal itself;  thus, by duplicating some vertices, we will be able to guarantee that each vertex is padded somewhere. 
In order to control the radius increase, the first petal might be carved using different parameters (see \cite{AN19} for details and explanation of this subtlety \footnote{One may notice that in \cref{line:edgeWeightReduce} of the \texttt{petal-decomposition} procedure, 	the weight of some edges is changed by a factor of 2. This can happen at most once for each copy of every edge throughout the  \texttt{hierarchical-petal-decomposition} execution, thus it may affect the padding parameter by a factor of at most 2. This re-weighting is ignored here for simplicity. We again refer readers to \cite{AN19} for details and further explanation.}). 
The definition of petal guarantees that the radius $\Delta_{x_0}(Y_j)$ is non-increasing, and when at step $s$ it becomes at most $3\Delta/4$, define $X_0=Y_s$ and then the \texttt{petal-decomposition} routine ends. In carving of the petal $\widetilde{X}_j\subseteq Y_{j-1}$, the algorithm chooses an arbitrary target $t_j\in Y_{j-1}$ (at distance at least $3\Delta/4$ from $x_0$) and a range $[\lo,\hi]$ of size $\hi-\lo\in\{\Delta/8,\Delta/4\}$ which are passed to the sub-routine \texttt{create-petal}.

\begin{algorithm}[t]
	\caption{\mbox{$\left(\left\{\underline{X}_j,X_j,\widetilde{X}_j,x_j,t_{j},\Delta_{j}\right\} _{j=0}^{s},\left\{ (y_{j},x_{j})\right\} _{j=1}^{s}\right)=\texttt{petal-decomposition}(G[X],
		x_{0},t,\Delta)$}}	\label{alg:petal-d}
	\DontPrintSemicolon
	Let $Y_0=X$\;
	Set $j=1$\;
	
	\If{$d_{X}(x_{0},t)\ge\Delta/2$}
	{Let
		$(\underline{X}_1,X_1,\widetilde{X}_1)=\texttt{create-petal}(G[Y_0],[d_{X}(x_{0},t)-\Delta/2,d_{X}(x_{0},t)-\Delta/4],x_0,t)$\;
		$Y_{1}=Y_{0}\backslash \underline{X}_1$\;
		Let $\{x_1,y_{1}\}$ be the unique edge on the shortest path $P_{x_{0}t}$ from $x_0$ to $t$ in $Y_0$, where $x_1\in X_1$ and $y_1\in Y_1$\;
		Set $t_{0}=y_{1}$, $t_{1}=t$; $j=2$\;}
	\Else
	{set $t_{0}=t$\;}
	
	\While{$Y_{j-1}\backslash B_{X}(x_{0},\frac{3}{4}\Delta)\neq\emptyset$}{
		Let $t_{j}\in Y_{j-1}$ be an arbitrary vertex satisfying $d_{X}(x_{0},t_{j})>\frac{3}{4}\Delta$\;
		Let	$(\underline{X}_j,X_j,\widetilde{X}_j)=\texttt{create-petal}(G[Y_{j-1}],[0,\Delta/8],x_0,t_j)$\;
		$Y_{j}=Y_{j-1}\backslash \underline{X}_{j}$\;
		Let $\{x_j,y_{j}\}$ be the unique edge on the shortest path $P_{x_{j}t_j}$ from $x_0$ to $t_j$ in $Y_{j-1}$, where $x_j\in \widetilde{X}_j$ and $y_j\in Y_j$\;
		Consider $G_j=G[\widetilde{X}_j]$ the graph induced by $\widetilde{X}_j$. For each edge $e\in P_{x_{j}t_{j}}(\widetilde{X}_j)$, set its weight to be $w(e)/2$ \label{line:edgeWeightReduce}\;
		Let $j=j+1$\;
	}
	
	Let $s=j-1$\;
	Let $\underline{X}_0=X_{0}=\widetilde{X}_{0}=Y_{s}$\;
	\Return $\left(\left\{ \underline{X}_j,X_j,\widetilde{X}_j,x_{j},t_{j},\Delta_{x_j}(\widetilde{X}_j)\right\} _{j=0}^{s},\left\{ (y_{j},x_{j})\right\} _{j=1}^{s}\right)$\;
\end{algorithm}

Both \texttt{hierarchical-petal-decomposition} and \texttt{petal-decomposition} are essentially the algorithms that appeared in \cite{AN19}. The only technical difference is that in \cite{AN19} $\widetilde{X}_j=\underline{X}_j$ for every $j$ (as they created actually spanning tree, while we are constructing a clan embedding).
The more important difference lies in the \texttt{create-petal} procedure, depicted in \Cref{alg:create-petal}. It carefully selects a radius $r\in[\lo,\hi]$, which determines the petal $\widetilde{X}_j$ together with a connecting edge $(x_j,y_j)\in E$, where $x_j\in \widetilde{X}_j$ is the center of $\widetilde{X}_j$ and $y_j\in Y_j$. It is important to note that the target $t_0\in X_0$ of the central cluster $X_0$ is determined during the creation of the first petal $X_1$. 
The petals are created using an alternative metric on the graph, known as the {\em cone-metric}:
\begin{definition}[Cone-metric] Given a graph $G=(V,E)$, a subset $X\subseteq V$
	and points $x,y\in X$, define the $\emph{cone-metric}$ $\rho=\rho(X,x,y):X^{2}\to\mathbb{R}^{+}$
	as $\rho(u,v)=\left|\left(d_{X}(x,u)-d_{X}(y,u)\right)-\left(d_{X}(x,v)-d_{X}(y,v)\right)\right|$.
\end{definition}
The cone-metric is in fact a pseudo-metric, i.e., distances between distinct points are allowed to be 0.
The ball $B_{(X,\rho)}(y,r)$ in the cone-metric
$\rho=\rho(X,x,y)$, contains all vertices $u$ whose shortest path to $x$ is increased (additively) by at most $r$ if forced to go through $y$.
In the $\texttt{create-petal}$
algorithm, while working in a subgraph $G[Y]$ with two specified vertices: a center $x_{0}$
and a target $t$, we define $W_{r}\left(Y,x_{0},t\right)=\bigcup_{p\in P_{x_{0}t}:\ d_{Y}(p,t)\le r}B_{(Y,\rho(Y,x_{0},p))}(p,\frac{r-d_{Y}(p,t)}{2})$
which is union of balls in the cone-metric, where any vertex $p$
in the shortest path from $x_{0}$ to $t$ of distance at most $r$
from $t$ is a center of a ball with radius $\frac{r-d_{Y}(p,t)}{2}$. See \Cref{fig:PetalExample} for an illustration.
The parameters $(Y,x_{0},t)$ are usually clear from the context and hence  omitted.
The following fact from \cite{AN19} demonstrates that petals are similar to balls.

\begin{figure}[t]
	\centering
	\includegraphics[width=0.7\textwidth]{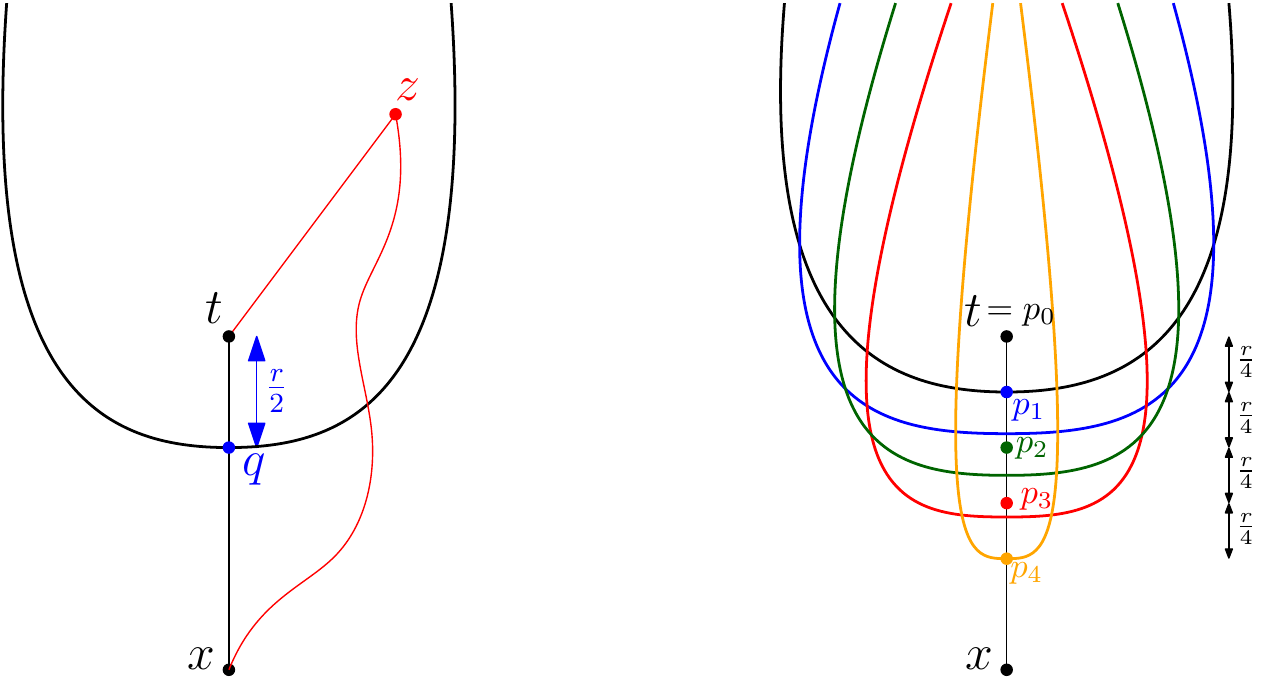}
	{\caption{\small
		On the left, we illustrate the ball $B_{(X,\rho)}(t,r)$ in the cone-metric
		$\rho=\rho(X,x,t)$ containing all vertices $u$ whose shortest path to $x$ is increased (additively) by at most $r$ if forced to go through $t$. The red vertex $z$ joins $B_{(X,\rho)}(t,r)$ as $d_X(z,t)+d_X(t,x)\le d_X(z,x)+r$. The blue point $q$ on the path $P_{t,x}$ at distance $\frac{r}{2}$ from $t$ is the last point on $P_{t,x}$ to join $B_{(X,\rho)}(t,r)$.\newline
		On the right, we illustrate the petal $W_{r}\left(X,x,t\right)=\bigcup_{p\in P_{xt}:\ d_{Y}(p,t)\le r}B_{(Y,\rho(X,x,p))}(p,\frac{r-d_{X}(p,t)}{2})$. In the illustration, the point $p_i$ is at distance $\frac{i}{4}r$ from $t$, and is the center of a ball of radius $\frac{4-i}{8}r$ in the respective cone metric.
		}
		\label{fig:PetalExample}}
\end{figure}

\begin{fact}[\cite{AN19}]\label{fact:W_rProp}
	For every $y\in W_{r}\left(Y,x_{0},t\right)$ and $l\ge 0$, $B_{G[Y]}(y,l)\subseteq W_{r+4l}\left(Y,x_{0},t\right)$.
\end{fact}
Note that \Cref{fact:W_rProp} implies that $W_r$ is monotone in $r$, i.e., for $r\le r'$, it holds that $W_{r}\subseteq W_{r'}$.

\sloppy For each $j$, the clusters $\underline{X}_j,X_j,\widetilde{X}_j$ returned by the \texttt{create-petal} procedure executed on $(G[Y_{j-1}],[\lo,\hi],x_0,t_j)$ will all be petals of the form $W_{r}(Y_{j-1},x_0,t_j)$ for $r\in [\lo,\hi]$. Specifically, we will chose some $r_1,r_2,r_3\in [\lo,\hi]$ such that $\underline{X}_j=W_{r_1}(Y_{j-1},x_0,t_j)$, $X_j=W_{r_2}(Y_{j-1},x_0,t_j)$ and $\widetilde{X}_j=W_{r_3}(Y_{j-1},x_0,t_j)$ while $r_2-r_1=r_3-r_2=\Theta(\frac{\hi-\lo}{k\log\log \mu(Y_{j-1})})$.

The following facts were proven in \cite{AN19} regarding the \texttt{petal-decomposition} procedure. They hold in our version of the algorithm using exactly the same proofs.
\begin{fact}[\cite{AN19}]\label{Fact:PetaDecomp} 
	Consider the \texttt{petal-decomposition} procedure executed on $X$ with center $x_0$, target $t$ and radius $\Delta$. It creates clusters $(\underline{X_0},X_0,\widetilde{X_0}),(\underline{X_1},X_1,\widetilde{X_1}),\dots,(\underline{X_s},X_s,\widetilde{X_s})$. During the process, we had temporary metrics $Y_0=X$, and $Y_j=Y_{j-1}\backslash \underline{X_j}$. For $j\ge 1$ the cluster $\widetilde{X}_j$ had center $x_j$ connected to $y_j\in Y_j$ and target $t_j\in \widetilde{X}_j$.
	Throughout the execution, the following hold:
	\begin{enumerate}
		\item For every $j$ and $z\in Y_j$, $P_{z,x_0}(X)\subseteq G[Y_j]$. In particular, the radius of the $Y_j$'s is monotonically non-increasing: $\Delta_{x_0}(Y_0)\ge\Delta_{x_0}(Y_1)\ge\dots\ge\Delta_{x_0}(Y_s)$.
		In particular $X_0$ is a connected cluster with radius at most $3\Delta/4$.
		\item For each $j\ge 0$, $\widetilde{X}_j$ is a connected cluster with center $x_j$, target $t_j$ such that $\Delta_{x_j}(X_j)\le 3\Delta/4$. In particular, the entire shortest path from $x_j$ to $t_j$ (in $Y_{j-1}$) is in $\widetilde{X}_j$.
		\item If a special first cluster is created, then $y_1\in X_0$ and $P_{x_0,t}(X)\subseteq G[X_0\cup X_1]$. If no special first cluster is created, then $P_{x_0,t}(X)\subseteq G[X_0]$.
	\end{enumerate}
\end{fact}

Next, we cite the relevant properties regarding the  \texttt{hierarchical-petal-decomposition} procedure. The proofs follow  almost the same lines as \cite{AN19}, with slight and natural adaptations due to the embedding being a clan embedding with duplicate copies for some vertices. In any case, no new ideas are required and we will skip the proof.
\begin{fact}[\cite{AN19}]\label{Fact:HierachicalPetaDecomp} 
	Consider the \texttt{hierarchical-petal-decomposition} procedure executed on $X$ with center $x_0$, target $t$ and radius $\Delta$.
	The following properties hold:
	\begin{enumerate}
		\item The algorithm returns a spanning clan embedding into a tree $T$.
		\item The tree $T$ has radius at most $4\Delta_{x_0}(X)$. That is 
		\[
		\Delta_{x_0}(T)\le 4\Delta_{x_0}(X)~.
		\]
	\end{enumerate}
\end{fact}
Note that it follows from \Cref{Fact:HierachicalPetaDecomp}, that the distance between every pair of vertices in the tree $T$ is at most $8\Delta_{x_0}(X)$.

We will need the following observation. Roughly speaking, it says that when the \texttt{petal-decomposition} algorithm is carving out $(\underline{X}_{j+1},X_{j+1},\widetilde{X}_{j+1})$, it is oblivious to the past petals, 
edges and targets -- it only cares about $Y_j$ and the original diameter $\Delta$.
\begin{observation}\label{ob:delta}
	Assume that \texttt{petal-decomposition} on input $(G\left[X\right],x_{0},t,\Delta_{x_0}(X))$ returns as output $\left(\left\{ \underline{X}_{j},X_{j},\widetilde{X}_{j},x_{j},t_{j},\Delta_{j}\right\} _{j\in\{0,\dots,s\}},\left\{ (y_{j},x_{j})\right\} _{j\in\{1,\dots,s\}}\right)$. \\
	Then running \texttt{petal-decomposition} on input $(G\left[Y_l\right],x_{0},t_0,\Delta_{x_0}(X))$ will output
	$\left(\left\{ \underline{X}_{j},X_{j},\widetilde{X}_{j},x_{j},t_{j},\Delta_{j}\right\} _{j\in\{0,l+1,\dots,s\}},\left\{ (y_{j},x_{j})\right\} _{j\in\{l+1,\dots,s\}}\right)$.
\end{observation}

\subsection{Choosing a Radius}\label{subsec:createPetal}

Fix some $1\le j\le s$, and consider carving the petal $(\underline{X}_{j},X_{j},\widetilde{X}_{j})$ from the graph induced on $Y=Y_{j-1}$. Our choice of radius bears similarities to the one in \cite{ACEFN20}. The properties of the petal decomposition described above (in \Cref{subsec:PetalDecompDesc}), together with \Cref{Fact:PetaDecomp} and \Cref{Fact:HierachicalPetaDecomp},  hold for any radius picked from a given interval. We will now describe the method to select a radius that suits our needs. The \texttt{petal-decomposition} algorithm provides an interval $[\lo,\hi]$ of size at least $\Delta/8$, and for each $r\in[\lo,\hi]$ let $W_{r}(Y,x_0,t)\subseteq Y$ denote the petal of radius $r$ (usually we will omit $(Y,x_0,t)$.).

Our algorithm will return three clusters: $\underline{X}_{j}\subseteq X_{j}\subseteq \widetilde{X}_{j}$ which will correspond to three petals $W_{r-\frac{R}{4Lk}}\subseteq W_r\subseteq W_{r+\frac{R}{4Lk}}$ respectively, where $\frac{R}{4Lk}=\Theta(\frac{\hi-\lo}{k\log\log \mu(Y)})=\Theta(\frac{\Delta}{k\log\log \mu(Y)})$. 
The algorithm will be executed recursively on $\widetilde{X}_{j}$, while $\underline{X}_j$ will be removed from $Y$. The cluster $X_j$ will only be used in order to define $\chi$ (during the \texttt{hierarchical-petal-decomposition} procedure). 
\Cref{fact:W_rProp} implies that the vertices in $X_j$ are padded by $\widetilde{X}_j$, while the vertices in $Y\backslash X_j$ are padded by $Y\backslash \underline{X}_j$.
If a pair of vertices $u,v$ do not belong to the same cluster (e.g. $u\in \underline{X}_j$ and $v\notin \widetilde{X}_j$) then $d_Y(u,v)=\Omega(\frac{\Delta}{k\log\log\mu(Y)})$. 
By \Cref{Fact:HierachicalPetaDecomp}, the diameter of the final tree will be $O(\Delta)$. In particular, the distance in the embedded tree between every copy of $u$ and $v$ will be bounded by $O(\Delta)=O(k\log\log\mu(Y))d_Y(u,v)$.
Note that only the vertices in $\widetilde{X}_j\backslash \underline{X}_j$ are duplicated. Thus, our goal  is to choose a radius $r$ such that the measure of the duplicated vertices would be small.

Our algorithm to select a radius is based on region growing techniques as in \cite{ACEFN20}, which is more involved than the region growing in \Cref{thm:ClanUltrametric}.
In the petal decomposition framework, we cannot pick as the center a vertex maximizing the "small ball" (as the target $t_j$ must be at distance $\frac34$ from $x_0$). We first choose an appropriate range that mimics that choice (see \cref{line:a-b} in \Cref{alg:create-petal}) --- this is the reason for the extra factor of $\log\log \mu(Y)$.
The basic idea in region growing is to charge the measure of the duplicated vertices (i.e. $\widetilde{X}_j\backslash\underline{X}_j$), to all the vertices in the cluster $\widetilde{X}_j$. 
In order to avoid a range in $[\lo,\hi]$ that contains more than half of the measure, we will cut either in $[\lo,\midP]$ or in $[\midP,\hi]$ where $\midP=(\hi+\lo)/2$.
Specifically, in the case where $W_\midP$ has measure at least $\mu(Y)/2$, we "cut backward" in the regime $[\midP,\hi]$, and charge the measure of duplicated vertices to the remaining graph $Y_j$, rather than to $\widetilde{X}_j$.

\begin{algorithm}[h]
	\caption{$(\underline{X},X,\widetilde{X})=\texttt{create-petal}(G[Y],\mu,[\lo,\hi],x_0,t)$}	\label{alg:create-petal}
	\DontPrintSemicolon
	$L=\lceil1+\log\log \mu(Y)\rceil$\;
	$R=\hi-\lo$; $\midP=(\lo+\hi)/2=\lo+R/2$\;
	For every $r$, denote $W_{r}=W_{r}(Y,x_0,t)$,   $w_r=\mu(W_r)$\;
	\If{$w_{\midP}\le \frac{\mu(Y)}{2}$}{
		Choose $\left[a,b\right]\subseteq\left[\lo,\midP\right]$\label{line:a-b}
		such that $b-a=\frac{R}{2L}$ and $w_a\ge w_b^2/\mu(Y)$
		\tcp*{see \Cref{lem:interval-choose}}	  
		
		Pick $r\in\left[a+\frac{b-a}{2k},b-\frac{b-a}{2k}\right]$ such that
		$w_{r+\frac{b-a}{2k}}\le w_{r-\frac{b-a}{2k}}\cdot\left(\frac{w_{b}}{w_{a}}\right)^{\frac{1}{k}}$\tcp*{see \Cref{lem:radiuosRG}}
	}
	\Else{
		For every $r\in[\lo,\hi]$, denote  $q_r=\mu(Y\backslash W_r)$\;
		Choose $\left[b,a\right]\subseteq\left[\midP,\hi\right]$
		such that $a-b=\frac{R}{2L}$ and $q_a\ge q_b^2/\mu(Y)$
		\tcp*{see \Cref{lem:interval-choose-Ver2}}
		
		Pick $r\in\left[b+\frac{b-a}{2k},a-\frac{b-a}{2k}\right]$ such that
		$q_{r-\frac{a-b}{2k}}\le q_{r+\frac{a-b}{2k}}\cdot\left(\frac{q_{b}}{q_{a}}\right)^{1/k}$
		\tcp*{see \Cref{lem:radiuosRG-Ver2}}
	}
	\Return $(W_{r-\frac{R}{4Lk}},W_r,W_{r+\frac{R}{4Lk}})$\;
\end{algorithm}

\subsection{Proof of \Cref{lem:SpanningClanTreeMeasure}: the distributional case}\label{subsec:SpanningClanCorrectness}
Let $u,v\in V$ be a pair of vertices, let $(f,\chi)$ be the spanning clan embedding into a tree $T$  returned by calling \texttt{hierarchical-petal-decomposition} on $(G[V],z,z,\Delta_z(V))$ for arbitrary $z\in V$.

\begin{lemma}\label{lem:paddedImpliesStretch}
	The clan embedding $(f,\chi)$ has distortion $O(\rho)=O(k\log\log \mu(V))$.
\end{lemma}
\begin{proof}
	The proof is by induction on the radius $\Delta$ of the graph (w.r.t. the center). The basic case is where the graph is a singleton and $\Delta=0$ is trivial.
	For the general case, consider a pair of vertices $u,v$. 
	Let $\left(\left\{\underline{X}_{j},X_{j},\widetilde{X}_{j},x_j,t_{j},\Delta_{j}\right\} _{j=0}^{s},\left\{ (y_{j},x_{j})\right\} _{j=1}^{s}\right)$ be the output of the call to the \texttt{petal-decomposition} procedure on $X,x_0$.
	For each $j\ge 1$, let $Y_{j-1}$ be the graph held during the $j$'th stage of the algorithm. Note that $Y_s=X_0$. Then we created the petals $(\underline{X}_{j},X_{j},\widetilde{X}_{j})=(W_{r_{j}-\frac{R}{4Lk}},W_{r_{j}},W_{r_{j}+\frac{R}{4Lk}})$, and $Y_{j}=Y_{j-1}\backslash \underline{X}_{j}$, where $L=\lceil1+\log\log \mu(Y_j)\rceil$, and $R\ge\frac{\Delta}{8}$. 
	Set $\rho=128\left\lceil 1+\log\log\mu(V)\right\rceil\cdot k=O(k\log\log\mu(V))$. Note that for every execution of the \texttt{create-petal} procedure at this stage, it holds that $\frac{\Delta}{\rho}\le\frac14\cdot\frac{R}{4Lk}$.
	
	First, consider the case where $d_G(u,v)\ge\frac\Delta\rho$. By \Cref{Fact:HierachicalPetaDecomp}, the distance between any pair of vertices in $T$ is $O(\Delta)$. In particular
	\[
	\min_{v'\in f(v)}d_{T}(v',\chi(u)) = O(\Delta)= O(\rho)\cdot d_{G}(u,v)~.
	\]
	
	Otherwise, $d_G(u,v)<\frac\Delta\rho$. Let  $B=B_X(u,\frac{\Delta}{\rho})$.
	For ease of notation, set $\underline{X}_{s+1}=X_{s+1}=\widetilde{X}_{s+1}=X_0=Y_s$.
	Let $j_u\in[1,s+1]$ be the minimal index such that $u\in X_j$. 
	We argue that $B\subseteq Y_{j_u-1}$. Assume otherwise, and let $j\in[1,j_u-1]$ be the minimal index such that $B\nsubseteq Y_{j}$. Thus, there is a vertex $u'\in B\cap \underline{X}_j\subseteq W_{r_j-\frac{R}{4Lk}}$, while by the minimality of $j$, it holds that $B\subseteq Y_{j-1}$.
	Using \Cref{fact:W_rProp}, it follows that 
	\[
	u\in B_{Y_{j-1}}(u',\frac{\Delta}{\rho})\subseteq W_{r_j-\frac{R}{4Lk}+4\cdot\frac{\Delta}{\rho}}\subseteq W_{r_j}=X_{j}~,
	\]
	a contradiction to the minimality of $j_u$.
	
	Next, we argue that $B\subseteq \widetilde{X}_{j_u}$.
	If $j_u=s+1$, then we have $B\subseteq Y_s=X_0=\widetilde{X}_{s+1}$ and done. Otherwise, as $u\in X_{j_u}=W_{r_{j_u}}$, using \Cref{fact:W_rProp} again we obtain
	\[
	B=B_{X}(u,\frac{\Delta}{\rho})=B_{Y_{j_u-1}}(u,\frac{\Delta}{\rho})\subseteq W_{r_{j_u}+4\cdot\frac{\Delta}{\rho}}\subseteq W_{r_{j_u}+\frac{R}{4Lk}}=\widetilde{X}_{j_u}~.
	\]
	
	In the \texttt{hierarchical-petal-decomposition} algorithm, we create a clan embedding $(f_{j_u},\chi_{j_u})$ of $\widetilde{X}_{j_u}$ into a tree $T_{j_u}$. The tree $T_{j_u}$ is incorporated into a global tree $T$, where $f(u)=\cup_jf_j(u)$,  $f(v)=\cup_jf_j(v)$, and $\chi(u)=\chi_{j_u}(u)$ by the definition of $j_u$.
	As $d_G(u,v)<\frac\Delta\rho$, it holds that $v\in B$. In particular, the shortest path from $v$ to $u$ in $G$ belongs to $B$, thus $d_{G[X_{j_u}]}(u,v)=d_G(u,v)$. By \Cref{Fact:PetaDecomp}, the radius of $X_{j_u}$ is at most $\frac34\Delta$; hence, using the induction hypothesis, we conclude that:
	\[
	\min_{v'\in f(v)}d_{T}(v',\chi(u))\le\min_{v'\in f_{j_{u}}(v)}d_{T_{j_{u}}}(v',\chi_{j_{u}}(u))=O(\rho)\cdot d_{G[X_{j_{u}}]}(u,v)=O(\rho)\cdot d_{G}(u,v)~.
	\]
	
\end{proof}

\begin{lemma}\label{lem:SavedTerminals}
	$\mathbb{E}_{v\sim\mu}[|f(v)|]\le \mu(V)^{1+1/k}$.
\end{lemma}
\begin{proof}	
	We prove by induction on $|X|$ and $\Delta$ that the one-to-many embedding $f$ constructed using the \texttt{hierarchical-petal-decomposition} algorithm  w.r.t. any  $(\ge1)$-measure $\mu$ fulfills $\mathbb{E}_{v\sim\mu}[|f(v)|]\le \mu(X)^{1+1/k}$.
	The base case where $X$ is a singleton is trivial. For the inductive step, assume we call \texttt{petal-decomposition} on $(G[X],x_0,t,\Delta)$ with $\Delta\ge\Delta_{x_0}(X)$ and measure $\mu$.
	
	Assume that the \texttt{petal-decomposition} algorithm does a non-trivial clustering of $X$ 
	to $\widetilde{X}_0,\widetilde{X}_1,\dots,\widetilde{X}_s$. (If it is the case that all vertices are sufficiently close to $x_0$, then no petal will be created, and the \texttt{hierarchical-petal-decomposition} will simply recurse on $(G[X],x_0,t,\Delta_{x_0}(X))$, so we can ignore this case.)
	Let $\widetilde{X}_1=W_{r+\frac{R}{4Lk}}$ be the first petal created by the \texttt{petal-decomposition} algorithm, and $Y_1=X\backslash \underline{X}_1$, where $\underline{X}_1=W_{r-\frac{R}{4Lk}}$.
	Denote by $\mu_{\widetilde{X}_j}$ the measure $\mu$ restricted to $\widetilde{X}_j$, and by $f_{\widetilde{X}_j}$ the one-to-many embedding our algorithm constructs for $\widetilde{X}_j$.

	By \Cref{ob:delta}, we can consider the remaining execution of \texttt{petal-decomposition} on $Y_1$ as a new recursive call of \texttt{petal-decomposition} with input $(G[Y_1],x_0,t_0,\Delta)$. 
	In particular, the recursive calls on $\widetilde{X}_0,\widetilde{X}_2,\dots,\widetilde{X}_s$ are completely independent from $\widetilde{X}_1$. Denote $f_{Y_1}=\cup_{j=0,2,\dots,s}f_{\widetilde{X}_j}$, and by $\mu_{{Y}_1}$ the measure $\mu$ restricted to $Y_1$. 
	Since $|\widetilde{X}_1|,|Y_1|<|X|$, the induction hypothesis implies that
	$\mathbb{E}_{v\sim\mu_{\widetilde{X}_1}}[|f_{\widetilde{X}_1}(v)|]\le \mu_{\widetilde{X}_1}(\widetilde{X}_1)^{1+\frac1k}=\mu(\widetilde{X}_1)^{1+\frac1k}$ and 	$\mathbb{E}_{v\sim\mu_{Y_1}}[|f_{Y_1}(v)|]\le \mu_{Y_1}(Y_1)^{1+\frac1k}=\mu(Y_1)^{1+\frac1k}$.
	Note that by our construction, 
	\[
	\mathbb{E}_{v\sim\mu}[|f(v)|]=\sum_{j=0}^{s}\mathbb{E}_{v\sim\mu_{\widetilde{X}_{j}}}[|f_{j}(v)|]=\mathbb{E}_{v\sim\mu_{\widetilde{X}_{1}}}[|f_{1}(v)|]+\mathbb{E}_{v\sim\mu_{Y_{1}}}[|f_{Y_{1}}(v)|]~.
	\]
	The rest of the proof is by case analysis according to the choice of radii in \Cref{alg:create-petal}. Recall that $w_{r'}=\mu(W_{r'})$ and $q_{r'}=\mu(Y\setminus W_{r'})=\mu(X\setminus W_{r'})$ for every parameter $r'$.

	\begin{enumerate}
		\item {\bf Case 1:} $w_{\midP}\le \mu(X)/2 $. In this case, we pick $a,b\in[\lo,\hi]$ where $b-a=R/(2L)$, and $r\in\left[a+\frac{b-a}{2k},b-\frac{b-a}{2k}\right]$ such that
		\begin{equation*}
		w_a>w_b^2/\mu(X)\qquad\qquad\mbox{and}\qquad\qquad w_{r+\frac{b-a}{2k}}\le w_{r-\frac{b-a}{2k}}\cdot\left(\frac{w_{b}}{w_{a}}\right)^{1/k}~.
		\end{equation*} 
		Here $\widetilde{X}_1=W_{r+\frac{b-a}{2k}}$, while $Y_1=X\backslash\underline{X}_1=X\backslash W_{r-\frac{b-a}{2k}}$.
		Using these two inequalities, we have that
		\[
		\mu(\widetilde{X}_1)^{1+\frac{1}{k}}=w_{r+\frac{b-a}{2k}}\cdot w_{r+\frac{b-a}{2k}}^{\frac{1}{k}}\le w_{r-\frac{b-a}{2k}}\cdot\left(\frac{w_{b}}{w_{a}}\right)^{\frac{1}{k}}\cdot w_{r+\frac{b-a}{2k}}^{\frac{1}{k}}\le w_{r-\frac{b-a}{2k}}\cdot\left(\frac{\mu(X)}{w_{b}}\right)^{\frac{1}{k}}\cdot w_{r+\frac{b-a}{2k}}^{\frac{1}{k}}\le w_{r-\frac{b-a}{2k}}\cdot\mu(X)^{\frac{1}{k}}~,
		\]
		where we used the fact that $r+\frac{b-a}{2k}\le b$ (and that $w_r$ is monotone). Using the induction hypothesis, we conclude that
		\begin{align*}
		\mathbb{E}_{x\sim\mu}[|f(x)|] & =\mathbb{E}_{x\sim\mu_{\widetilde{X}_{1}}}[|f_{\widetilde{X}_{1}}(x)|]+\mathbb{E}_{x\sim\mu_{Y_{1}}}[|f_{Y_{1}}(x)|]\\
		& \le\mu(\widetilde{X}_{1})^{1+\frac{1}{k}}+\mu(Y_{1})^{1+\frac{1}{k}}\\
		& \le w_{r-\frac{b-a}{2k}}\cdot\mu(X)^{\frac{1}{k}}+\mu(Y_{1})\cdot\mu(X)^{\frac{1}{k}}\\
		& =\left(\mu(W_{r-\frac{b-a}{2k}})+\mu(X\backslash W_{r-\frac{b-a}{2k}})\right)\cdot\mu(X)^{\frac{1}{k}}=\mu(X)^{1+\frac{1}{k}}~,
		\end{align*}
		where  the second inequality is because $\mu(Y_1)\le \mu(X)$.

		\item {\bf Case 2:} $w_{\midP}> \mu(X)/2$. This case is completely symmetric.
		Denoting $q_r=\mu(X\setminus W_r)$, we picked $b,a\in[\lo,\hi]$ so that $a-b=R/(2L)$ and $r\in\left[b+\frac{b-a}{2k},a-\frac{b-a}{2k}\right]$ such that
		\[q_a\ge q_b^2/\mu(X)\qquad\qquad\mbox{and}\qquad\qquad q_{r-\frac{a-b}{2k}}\le q_{r+\frac{a-b}{2k}}\cdot\left(\frac{q_{b}}{q_{a}}\right)^{1/k}~,
		\]
		Here $\widetilde{X}_1=W_{r+\frac{b-a}{2k}}$, while $Y_1=X\backslash W_{r-\frac{b-a}{2k}}$.
		Note that $\mu(Y_1)=q_{r-\frac{a-b}{2k}}$ while $\mu(\widetilde{X}_1)=\mu(X)-q_{r+\frac{a-b}{2k}}$.
		Using this two inequalities we have that
		\[
		\mu(Y_{1})^{1+\frac{1}{k}}=q_{r-\frac{b-a}{2k}}\cdot q_{r-\frac{b-a}{2k}}^{\frac{1}{k}}\le q_{r+\frac{b-a}{2k}}\cdot\left(\frac{q_{b}}{q_{a}}\right)^{\frac{1}{k}}\cdot q_{r-\frac{b-a}{2k}}^{\frac{1}{k}}\le q_{r+\frac{b-a}{2k}}\cdot\left(\frac{\mu(X)}{q_{b}}\right)^{\frac{1}{k}}\cdot q_{r-\frac{b-a}{2k}}^{\frac{1}{k}}\le q_{r+\frac{b-a}{2k}}\cdot\mu(X)^{\frac{1}{k}}
		\]
		where we used the fact that $b\le r-\frac{b-a}{2k}$. Following previous calculations, we conclude  that:
		\begin{align*}
		\mathbb{E}_{x\sim\mu}[|f(x)|] & \le\mu(\widetilde{X}_{1})^{1+\frac{1}{k}}+\mu(Y_{1})^{1+\frac{1}{k}}\\
		& \le\mu(\widetilde{X}_{1})\mu(X)^{\frac{1}{k}}+q_{r+\frac{b-a}{2k}}\cdot\mu(X)^{\frac{1}{k}}\\
		& =\left(\mu(W_{r+\frac{a-b}{2k}})+\mu(X\backslash W_{r+\frac{a-b}{2k}})\right)\cdot\mu(X)^{\frac{1}{k}}=\mu(X)^{1+\frac{1}{k}}~.
		\end{align*}
	\end{enumerate}
	
\end{proof}

\Cref{lem:SpanningClanTreeMeasure} follows by the combination of \Cref{lem:paddedImpliesStretch} and \Cref{lem:SavedTerminals}.

\subsection{Missing proofs from the \texttt{create-petal} procedure (\Cref{alg:create-petal})}\label{subsec:MissingCreatePetal}
		In this section we prove that the choices made in the \texttt{create-petal} procedure are all legal. In all lemmas in this section, we shall use the notation in \Cref{alg:create-petal}.
		\begin{lemma}\label{lem:interval-choose}
			If $w_{\midP}\le\mu(Y)/2$ 
			then there is $\left[a,b\right]\subseteq\left[\lo,\midP\right]$
			such that $b-a=\frac{R}{2L}$ and $w_a\ge w_b^2/\mu(Y)$.
		\end{lemma}
		\begin{proof}
			Seeking contradiction, assume that for every such $a,b$ with $b-a=\frac{R}{2L}$ it holds that $w_b>\sqrt{\mu(Y)\cdot w_a}$.
			Applying this on $b=\midP-\frac{i R}{2L}$ and $a=\midP-\frac{(i+1) R}{2L}$ for every $i=0,1,\dots,L-2$, we have that
			\[
			w_{\midP}>\mu(Y)^{1/2}\cdot w_{\midP-\frac{R}{2L}}^{1/2}>\dots>\mu(Y)^{1-2^{-(L-1)}}\cdot w_{\midP-\frac{(L-1)R}{2L}}^{2^{-(L-1)}}\ge\mu(Y)\cdot2^{-1}\cdot w_{\lo}^{2^{-(L-1)}}\ge\frac{\mu(Y)}{2}~,
			\]
			where we used that $\log\log \mu(Y)\le L-1$ and $\midP=\lo+R/2$. In the last inequality, we also used that $W_{\lo}$ contains at least one vertex, thus $w_{\lo}\ge 1$.
			The contradiction follows.
		\end{proof}
		\begin{lemma}\label{lem:radiuosRG}
			There is $r\in\left[a+\frac{b-a}{2k},b-\frac{b-a}{2k}\right]$ such that $w_{r+\frac{b-a}{2k}}\le w_{r-\frac{b-a}{2k}}\cdot\left(\frac{w_{b}}{w_{a}}\right)^{\frac{1}{k}}$.
		\end{lemma}
		\begin{proof}
			Seeking contradiction, assume there is no such choice of $r$. Then applying the inequality for $r=b-(i+1/2)\cdot\frac{b-a}{k}$ for $i=0,1,\dots,k-1$ we get
			\[
			w_{b}  >w_{b-\frac{b-a}{k}}\cdot\left(\frac{w_{b}}{w_{a}}\right)^{1/k}
			>\cdots>w_{b-k\cdot\frac{b-a}{k}}\cdot\left(\frac{w_{b}}{w_{a}}\right)^{k/k}=w_{a}\cdot\frac{w_{b}}{w_{a}}=w_{b}~,
			\]	
			a contradiction.
		\end{proof}
		The following two lemmas are symmetric to the two lemmas above.
		\begin{lemma}\label{lem:interval-choose-Ver2}
			If $w_{\midP}>\frac{\mu(Y)}{2}$ (implies $q_{\midP}\le\frac{\mu(Y)}{2}$), 
			then there is $\left[b,a\right]\subseteq\left[\midP,\hi\right]$ such that $a-b=\frac{R}{2L}$ and $q_a\ge q_b^2/\mu(Y)$.
		\end{lemma}
		\begin{lemma}\label{lem:radiuosRG-Ver2}
			There is $r\in\left[b+\frac{b-a}{2k},a-\frac{b-a}{2k}\right]$ such that
			$q_{r-\frac{a-b}{2k}}\le q_{r+\frac{a-b}{2k}}\cdot\left(\frac{q_{b}}{q_{a}}\right)^{1/k}$.
		\end{lemma}
		
		\subsection{Grand finale: proof of \Cref{thm:ClanSpanningTree}}\label{subsec:ProofOfSpanningTree}
		The proof of \Cref{thm:ClanSpanningTree} using \Cref{lem:SpanningClanTreeMeasure} follows the same lines as the proof of \Cref{thm:ClanUltrametric} from \Cref{lem:clanTreeMeasure}. 
		First we transform the language of $(\ge1)$-measure to that of probability measure.		
		\begin{lemma}\label{lem:clanTreeSpanningProbability}
			Given an $n$-point weighted graph $G=(V,E,w)$ and probability measure
			$\mu:V\rightarrow\mathbb{R}_{\ge0}$, we can construct the two following spanning clan embeddings $(f,\chi)$ into a tree:
			\begin{enumerate}
				\item For integer $k\ge 1$, multiplicative distortion $O(k\log\log n)$ such that $\mathbb{E}_{x\sim\mu}[|f(x)|]\le O(n^{\frac1k})$.
				\item For  $\epsilon\in(0,1]$, multiplicative distortion $O(\frac{\log n\log\log n}{\eps})$ such that $\mathbb{E}_{x\sim\mu}[|f(x)|]\le1+\epsilon$.
			\end{enumerate}
		\end{lemma}
	The proof of \Cref{lem:clanTreeSpanningProbability} is exactly identical to that of \Cref{lem:clanTreeProbability} and we will skip it. The only subtlety to note is the  $(\ge1)$-measure $\widetilde{\mu}_{\ge1}$ constructed during the proof of \Cref{lem:clanTreeProbability} fulfills $ \widetilde{\mu}_{\ge1}(V)=2n$, and thus the multiplicative distortion guarantee from \Cref{lem:SpanningClanTreeMeasure} will be $O(k\log\log n)$.
	\Cref{thm:ClanSpanningTree} now follows from the minimax theorem (in the exact same way as the proof of \Cref{thm:ClanUltrametric}).

%% file: LowerBoundTree.tex
\section[Lower Bound for Clan Embeddings into Trees (\Cref{thm:ClnUltrametricLB})]{Lower Bound for Clan Embeddings into Trees}\label{sec:LB}
This section is devoted to proving \Cref{thm:ClnUltrametricLB} that we restate below.
\LBClanTree*

The girth of an unweighted graph $G$ is the length of the shortest cycle in $G$. The Erd\H{o}s' girth conjecture states that for any $g$ and $n$, there exists an $n$-vertex graph with girth $g$ and $\Omega(n^{1+\frac{2}{g-2}})$ edges. The conjecture is known to holds for $g=4,6,8,12$ (see \cite{Benson66,Wenger91}). However, the best known  lower bound for general $k$ is due to Lazebnik \etal \cite{LUW95}.
\begin{theorem}[\cite{LUW95}]\label{thm:girth}
	For every even $g$, and $n$, there exists an unweighted graph with girth $g$ and $\Omega(n^{1+\frac43\cdot\frac{1}{g-2}})$ edges.
\end{theorem}
From the upper bound perspective, the (generalized) Moore's bound \cite{AHL02,BR10} states that every $n$ vertex graph with girth $g$ has at most $n^{1+\frac{2}{g-2}}$ edges for $g\le 2\log n$, and at most $n\left(1+(1+o(1))\frac{\ln(m-n+1)}{g}\right)$ edges for larger $g$; here  $m$ is the number of edges. 

We will be able to use \Cref{thm:girth} to prove the second assertion in \Cref{thm:ClnUltrametricLB}. That is, any clan embedding into a tree with distortion $O(k)$ must have $\sum_{x\in X}|f(x)|\ge\Omega(n^{1+\frac1k})$.
However, the first assertion requires a much tighter lower bound of $(1+\eps)n$ on the number of edges. Therefore, the asymptotic nature of \Cref{thm:girth} is unfortunately not strong enough for our needs.
We begin by showing that for large enough $n$ and $\eps\in(0,1)$, there exists an $n$-vertex graph with $(1+\eps)n$ edges and girth $\Omega(\frac{\log n}{\eps})$. We are not aware of this very basic fact to previously appear in the literature.
Note that \Cref{lem:lognepsGirth} matches Moore's upper bound (up to a constant dependency on the girth $g$).

\begin{lemma}\label{lem:lognepsGirth}
	For every fixed $\eps\in (0,1)$ and large enough $n$, there exists a graph with at least $(1+\eps)n$ edges and girth $\Omega(\frac{\log n}{\eps})$.
\end{lemma}

\begin{remark}[Ultra sparse spanners]\label{rem:UltraSparseSpanners}
	Given a graph $G=(V,E,w)$, a $t$-spanner is a subgraph $H$ of $G$ such that for every pair of vertices $u,v\in V$, $d_H(u,v)\le t\cdot d_G(u,v)$.
	For every fixed $\eps\in(0,1)$, Elkin and Neiman \cite{EN19} constructed ultra-sparse spanners with $(1+\eps)n$ edges and stretch $O(\frac{\log n}{\eps})$. 
	Even though they noted that the sparsity of their spanner matches the Moore's bound, it remained open whether one can construct better spanners. 
	As the only $(g-2)$-spanner of a graph with girth $g$ is the graph itself, 
	\Cref{lem:lognepsGirth} implies that the ultra sparse spanner from \cite{EN19} is tight (up to a constant in the stretch).
\end{remark}

For the case of girth $\Omega(\log n)$, the first step is to replace the asymptotic notation in the lower bound on the number of edges from \Cref{thm:girth} with explicit bound. 
\begin{claim}\label{clm:lognGirthGraph}
	For every $n\in N$, there exist an $n$-vertex graph with $2n$ edges and girth $\Omega(\log n)$.
\end{claim}
\begin{proof}
	Set $p=\frac{4n}{{n\choose 2}}=\frac{8}{n-1}$. Consider a graph $G=(V,E)$ sampled according to $G(n,p)$ (that is, each edge sampled to $G$ i.i.d. with probability $p$.). It holds that $\mathbb{E}[|E|]={n\choose 2}\cdot p=4n$. By Chernoff bound,
	\[
	\Pr\left[|E|<3n\right]\le e^{-\frac{1}{32}\mathbb{E}[E]}=e^{-\frac{n}{8}}~.
	\]
	On the other hand, for $t\ge 3$, denote by $C_t$ the set of cycles of length exactly $t$. Then,
	\[
	\mathbb{E}\left[|C_{t}|\right]\le n(n-1)\cdots(n-t+1)\cdot p^{t}=\frac{n(n-1)\cdots(n-t+1)}{(n-1)^{t}}\cdot4^{t}<4^{t}~.
	\]
	Denote by $\mathcal{C}$ the set of all cycles of length smaller than $\frac13\log n$. Then
	\[
	\mathbb{E}\left[|\mathcal{C}|\right]=\sum_{t=3}^{\frac{1}{3}\log n-1}\mathbb{E}\left[|C_{t}|\right]\le\sum_{t=3}^{\frac{1}{3}\log n-1}4^{t}<4^{\frac{1}{3}\log n}=n^{\frac23}~.
	\]
	By Markov inequality, $\Pr\left[|\mathcal{C}|\ge n\right]\le\frac{\mathbb{E}\left[|\mathcal{C}|\right]}{n}<n^{-\frac13}<\frac{1}{2}$.
	By union bound, there exists a graph $G$ with at least $3n$ edges, and at most $n$ cycles of length less than $\frac{1}{3}\log n$. Let $G'$ be the graph obtained by deleting an arbitrary single edge from each cycle. Continue deleting edges until $G'$ has exactly $2n$ edges. 
	We conclude that $G'$ has $2n$ edges and girth at least $\frac13\log n$ as required.
\end{proof}
\begin{proof}[Proof of \Cref{lem:lognepsGirth}]
	Fix $\delta=\frac{1-\eps}{2\eps}$. Set $n'=\eps n=\frac{n}{1+2\delta}$.
	We ignore issues of integrality during the proof. Such issues could be easily fixed as we don't state an explicit bound on the girth.	
	Using \Cref{clm:lognGirthGraph}, construct a graph $G'$ with $n'$ vertices, $2n'$ edges,
	and girth $\Omega(\log n')$.
	
	Let $G$ be the graph obtained from $G'$ by replacing each edge with 
	a path of length $\delta+1$. Then: 
	\begin{align*}
	|V(G)| & =|V(G')|+\delta\cdot|E(G')|=n'+\delta\cdot2n'=n'(1+2\delta)=n\\
	|E(G)| & =(\delta+1)\cdot|E(G')|=(\delta+1)\cdot2n'=n\cdot\frac{2(1+\delta)}{1+2\delta}=(1+\eps)n~,
	\end{align*}
	where the last equality follows by the definition of $\delta$.
	Note that the girth of $G$ is at least $\Omega((1+\delta)\log n')=\Omega(\frac{\log\eps n}{\eps})=\Omega(\frac{\log n}{\eps})$, for $n$ large enough.
\end{proof}

The \emph{Euler characteristic} of a graph $G$ is defined as $\chi(G)\coloneqq|E(G)|-|V(G)|+1$. Our lower bound is based on the following theorem by Rabinovich and Raz \cite{RR98}.
\begin{theorem}[\cite{RR98} ]\label{thm:RR98}
	Consider an unweighted graph $G$ with girth $g$, and consider a (classic) embedding $f:G\rightarrow H$ of $G$ into a weighted graph $H$, such that $\chi(H)<\chi(G)$. Then $f$ has multiplicative distortion at least $\frac{g}{4}-\frac{3}{2}$.
\end{theorem}
Next, we transfer the language of classic embeddings into graphs used in \Cref{thm:RR98} to that of clan embeddings into trees.
\begin{lemma}\label{lem:ClanRR98}
	Consider an unweighted, $n$-vertex graph $G=(V,E)$ with girth $g$, and let $(f,\chi)$ be a clan embedding of $G$ into a tree $T$ with multiplicative distortion $t<\frac{g}{4}-\frac{3}{2}$. Then necessarily $\sum_{v\in V}|f(v)|\ge n+\chi(G)$.
\end{lemma}
\begin{proof}
	Let $H$ be the graph obtained from $T$ by merging all the copies of each vertex. Specifically, arbitrarily order the vertices in $V$: $v_1,v_2,\dots,v_n$. Iteratively construct a series of graphs $H_0=T,H_1,H_2,\dots,H_n$ with one-to-many embeddings $f_i:G\rightarrow H_i$.
	In the $i$'th iteration, we create $H_i,f_i$ out of $H_{i-1},f_{i-1}$ by replacing all the vertices in $f_{i-1}(v_i)$ by a single vertex $\tilde{v}_i$. For a vertex $u\in H_{i-1}$, we add an edge from $u$ to $\tilde{v}_i$ if there was an edge from $u$ to some vertex in $f_{i-1}(v)$. If an edge $\{u,\tilde{v}_i\}$ is added, its weight is defined to be $\min_{v'\in f_{i-1}(v)}w_{H_{i-1}}(v',u)$.
	Set $H=H_n$, and $\tilde{f}=f_n$. 	
	Clearly, distances in $H$ can only decrease compared to $T$.
	This is because for every $u,v\in V$, $d_H(\tilde{u},\tilde{v})\le \min_{u'\in f(u),~v'\in f(v)}d_T(u',v')\le \min_{u'\in f(u)}d_T(u',\chi(v))\le t\cdot d_G(u,v)$.
	On the other hand, by induction (and the triangle inequality), since $f$ is a dominating embedding, one can show that $\tilde{f}$ is also dominating. That is $\forall u,v\in V$, $d_H(\tilde{u},\tilde{v})\ge d_G(u,v)$.
	
	We conclude that $\tilde{f}$ is a classic embedding of $G$ with a  multiplicative distortion at most $t<\frac{g}{4}-\frac32$. By \Cref{thm:RR98}, it follows that $\chi(H)\ge\chi(G)$.	
	For every $i$, it holds that 
	\[
	\chi(H_{i})=|E(H_{i})|-|V(H_{i})|-1\le|E(H_{i-1})|-\left(|V(H_{i-1})|-|f(v_{i})|+1\right)-1=\chi(H_{i-1})+|f(v_{i})|-1
	\]
	As the Euler characteristic of a tree equals $0$, we obtain 
	\[
	\chi(G)\le\chi(H)=\chi(H_{n})\le\sum_{i}(|f(v_{i})|-1)+\chi(T)=\sum_{v\in V}|f(v)|-n~,
	\]
as desired.
\end{proof}

We are now ready to prove \Cref{thm:ClnUltrametricLB}.
\begin{proof}[Proof of \Cref{thm:ClnUltrametricLB}]

	For the first assertion,  
	using \Cref{lem:lognepsGirth}, let $G$ be an unweighted graph with girth  $g=\Omega(\frac{\log n}{\eps})$ and $(1+\eps)n$ edges.
	Consider a clan embedding of $G$ into a tree with distortion smaller than $\frac{g}{4}-\frac32=\Omega(\frac{\log n}{\eps})$.
	By \Cref{lem:ClanRR98}, it holds that 
	\[
	\sum_{v\in V}|f(v)|\ge n+\chi(G)=|E(G)|+1>(1+\eps)n~.
	\]
		
	The second assertion follows similar lines. Set $g=2\cdot\left\lfloor \frac{\frac{4}{3}k+2}{2}\right\rfloor$. Note that $g$ is largest even number up to $\frac{4}{3}k+2$. Using \Cref{thm:girth}, let $G$ be an unweighted graph with girth $g$ and $\Omega(n^{1+\frac{4}{3}\cdot\frac{1}{g-2}})\ge\Omega(n^{1+\frac{1}{k}})$ edges.
	Consider a clan embedding of $G$ into a tree with distortion smaller than $\frac{g}{4}-\frac32=\Omega(k)$.
	By \Cref{lem:ClanRR98}, it holds that $$\sum_{v\in V}|f(v)|\ge n+\chi(G)=|E(G)|+1=\Omega(n^{1+\frac{1}{k}})~.$$
\end{proof}

%% file: ramsey.tex
\section[Ramsey Type Embedding for Minor-Free Graphs (\Cref{thm:Ramsey-minor-Free})]{Ramsey Type Embedding for Minor-Free Graphs}\label{sec:RamseyMinorFree}

This section is devoted to proving the following theorem,
\Ramsey*
We begin by proving \Cref{thm:Ramsey-minor-Free} for the special case of nearly-$h$-embeddable graphs.
\begin{lemma}\label{lem:Ramsey-almost-embeddable}
	Given a nearly $h$-embeddable $n$-vertex graph $G=(V,E,w)$ of diameter $D$, and parameters $\eps\in(0,\frac14)$, $\delta\in(0,1)$ , there is a distribution over one-to-many, clique preserving, dominating embeddings $f$ into treewidth $O_{h}(\frac{\log n}{\eps\delta})$ graphs, such that there is a subset $M\subseteq V$ of vertices for which the following claims hold:
	\begin{enumerate}
		\item For every clique $Q\subseteq V$, $\Pr[Q\subseteq M]\ge 1-\delta$. \label{property:MLargeAlmostEmbeddable}
		\item For every $u\in M$ and $v\in V$, $\max_{u'\in f(u),v'\in f(v)}d_H(u',v'))\le d_G(u,v)+\eps D$.\label{property:MSmallDistortionAlmostEmbeddable}
	\end{enumerate}
\end{lemma}
\begin{proof}
	Consider a nearly $h$-embedded graph $G=(V,E,w)$. Assume w.l.o.g. that $D=1$, otherwise we will scale accordingly. 
	We assume that $1/\delta$ is an integer, otherwise, we solve for $\delta'$ such that $\frac{1}{\delta'}=\lceil\frac{1}{\delta}\rceil$.
	Let $\varPsi$ be the set of apices. 
	We will construct $q=\frac{5}{\delta}$ embeddings, all satisfying  \propref{property:MSmallDistortionAlmostEmbeddable} of \Cref{lem:Ramsey-almost-embeddable}.
	The final embeddings will be obtained by choosing one of these  embeddings uniformly at random.
	We first create a new graph
	$G'=G[V\setminus\varPsi]$ by deleting all the apex vertices $\Psi$. 
	In the tree decomposition of $H$ to be constructed, the set $\Psi$ will belong to all the bags (with edges towards all the vertices). Thus we can assume that $G'$ is connected, since otherwise, we can simply solve the problem on each connected component separately  and combine the solutions by taking the union of all graphs/embeddings.
	
	Let $r\in G'$ be an arbitrary vertex. For $\sigma\in\left\{ 1,\dots,\frac{5}{\delta}\right\} $
	set $\mathcal{I}_{-1,\sigma}=[0,\sigma]$, $\mathcal{I}_{-1,\sigma}^{+}=[0,\sigma+1]$,
	and 
	$$\mbox{for $j\ge0$, set}\qquad
	\mathcal{I}_{j,\sigma}=\left[\frac{5j}{\delta}+\sigma,\frac{5(j+1)}{\delta}+\sigma\right),\quad\mbox{and}\quad
	\mathcal{I}_{j,\sigma}^{+}=\left[\frac{5j}{\delta}+\sigma-1,\frac{5(j+1)}{\delta}+\sigma+1\right),
	$$
	Set $U_{j,\sigma}=\left\{ v\in G'\mid d_{G'}(r,v)\in\mathcal{I}_{j,\sigma}\right\} $
	and similarly $U_{j,\sigma}^{+}$ 
	w.r.t. $I_{j,\sigma}^{+}$. 
	Let $G_{j,\sigma}$ be the graph
	induced by $U^{+}_{j,\sigma}$, plus the vertex $r$.
	In addition, for every vertex $v\in U^{+}_{j,\sigma}$ who has a neighbor in $\cup_{j'<j}U^{+}_{j',\sigma}\setminus U_{j,\sigma}^{+}$, we add an edge to $r$ of weight $d_G(v,r)$.
	Equivalently, $G_{j,\sigma}$  can be constructed by taking the graph induced by $\cup_{j'\le j}U^{+}_{j',\sigma}$, and contracting all the internal edges out of $U_{j,\sigma}^{+}$ into $r$. See \Cref{fig:GraphLayers} (in \Cref{sec:ClanMinorFree}) for an illustration.
	Note that all the edges towards $r$ have weight at most $D=1$, thus $G_{j,\sigma}$ is a nearly $h$-embedded graph with diameter
	at most $2\cdot(\frac{5}{\delta}+3)=O(\frac{1}{\delta})$ and no apices. 
	
	Fix some $\sigma$ and $j$. Using \Cref{lm:embed-genus-vortex} with parameter $\Theta(\eps\cdot\delta)$,
	we construct a one-to-many embedding $f_{j,\sigma}$, of $G_{j,\sigma}$ into a graph $H_{j,\sigma}$ with treewidth $O_{h}(\frac{\log n}{\eps\cdot\delta})$, such that $f_{j,\sigma}$ is clique preserving and has additive distortion $\Theta(\eps\cdot\delta)\cdot O(\frac1\delta)=\eps$.
	After the application of \Cref{lm:embed-genus-vortex}, we will merge all copies of $r$, and add edges from $r$ to all the other vertices (where the weight of a new edge $\left(r,v\right)$ is $d_G(r,v)$). Note that this increases the treewidth by at most $1$. Furthermore, we will assume that there is a bag containing only the vertex $r$ (as we can simply add such a bag). 
	Next, fix $\sigma$. Let $H'_{\sigma}$ be a union of the graphs $\cup_{j\ge -1} H_{j,\sigma}$. We identify the vertex $r$ with itself, but for all the other vertices that participate in more  than one graph, their copies in each graph remain separate. Formally, we define a one-to-many embedding $f_\sigma$, where $f_\sigma(r)$ equals to the unique $r$, and for every other vertex $v\in V\setminus \Psi$, $f_\sigma(v)=\bigcup_{j\ge -1} f_{j,\sigma}(v)$. 
	Note that $H'_{\sigma}$ has a tree decomposition of width $O_{h}(\frac{\log n}{\eps\cdot\delta})$, by identifying the bag containing
	only $r$ in all the graphs. 
	Finally, we create the graph $H_{\sigma}$ by adding the set $\Psi$ with edges towards all the vertices in $H'_{\sigma}$,
	where the weight of a new edge $\left(u',v\right)$ 
	is $d_{G}(u,v)$. For $v\in \Psi$, set $f_\sigma(v)=\{v\}$. 
	As $\Psi=O_h(1)$, $H_{\sigma}$ has treewidth $O_{h}(\frac{\log n}{\eps\cdot\delta})$.	Finally, set $M_{j,\sigma}=\left\{ v\in G'\mid d_{G'}(r,v)\in\big[\frac{5j}{\delta}+\sigma+2,\frac{5(j+1)}{\delta}+\sigma-2\big)\right\}$
	, and $M_\sigma=\Psi\cup\{r\}\cup\bigcup_{j\ge-1}M_{j,\sigma}$.	
	This finishes the construction.	
	
	Observe that the one-to-many embedding $f_\sigma$ is dominating. This follows from the triangle inequality since every edge $\{u',v'\}$ for $u'\in f_\sigma(u),v'\in f_\sigma(v)$ in the graph has weight $d_G(u,v)$. 
	Next we argue that $f_\sigma$ is clique-preserving. Consider a clique $Q$ in $G$, and let $\tilde{Q}=Q\setminus \Psi$ be the non apex vertices in $Q$.
	We will show that $f_\sigma$ contains a clique copy of $\tilde{Q}$. As the apices have edges towards all the other vertices, it will imply that $f_\sigma$ is clique-preserving.
	Let $v\in \tilde{Q}$ be some arbitrary vertex, and $j$ be the unique index such that  $v\in U_{j,\sigma}$. For every $u\in \tilde{Q}$, $d_{G'}(v,u)=d_{G}(v,u)\le 1$, implying $u\in U^+_{j,\sigma}$. We conclude that all $\tilde{Q}$ vertices belong to $G_{j,\sigma}$. As $f_{j,\sigma}$ is clique-preserving, it follows that there is a bag in $H_{j,\sigma}$, and thus also in $H_\sigma$, containing a clique copy of $\tilde{Q}$.
	
	Next, we argue that property (\ref{property:MLargeAlmostEmbeddable}) holds.
	We say that $f$ fails on a vertex $v\in V$ if $v\notin M$, and we say that $f$ fails on a clique $Q$ if $Q\nsubseteq M$.
	Consider some clique $Q$; we can assume w.l.o.g. that $Q$ does not contain any apex vertices (as $f$ never fails on  an apex vertex). 
	Let $s_Q,t_Q\in Q$ be the closest and farthest vertices from $r$ in $G'$, respectively.
	Then $d_{G'}(r,t_Q)-d_{G'}(r,s_Q)\le d_{G'}(s_Q,t_Q)\le 1$. $f_\sigma$ fails on $Q$ iff there is a non-empty intersection between the interval $[d_{G'}(r,t_Q),d_{G'}(r,s_Q))$ and the interval $[\frac{5j}{\delta}+\sigma-2,\frac{5j}{\delta}+\sigma+2)$ for some $j$. Note that there are at most $5$ values of $\sigma$ for which this intersection is non-empty.
	As we constructed $q=\frac5\delta$ embeddings,	
	\[
	\Pr_{\sigma}\left[Q\subseteq M_{\sigma}\right]=\frac{\left|\left\{ \sigma\in[q]\mid Q\subseteq M_{\sigma}\right\} \right|}{q}\geq\frac{q-5}{q}=1-\delta
	\]
	
	Finally, we show that $f_\sigma$ has additive distortion $\eps D$ w.r.t. $M_\sigma$. Consider a pair of vertices $u\in M_\sigma$ and $v\in V$. If one of $u,v$ belongs to $\Psi\cup \{r\}$ then for every $u'\in f_\sigma(u)$ and  $v'\in f_\sigma(v)$, $d_{H_\sigma}(u',v')=d_G(u,v)$.
	Otherwise, if $d_{G'}(u,v)>d_{G}(u,v)$, then it must be that the shortest path between $u$ to $v$ in $G$ goes through an apex vertex $z\in \Psi$. In $H_\sigma$, $f_\sigma(z)$ is a singleton that have an edge towards every other vertex. It follows that 
	$\max_{u'\in f_{\sigma}(u),v'\in f_{\sigma}(v)}d_{H_{\sigma}}(u',v')\le\max_{u'\in f_{\sigma}(u),v'\in f_{\sigma}(v)}d_{H_{\sigma}}(u',f_{\sigma}(z))+d_{H_{\sigma}}(f_{\sigma}(z),v')=d_{G}(u,z)+d_{G}(z,v)=d_{G}(u,v)~.$
	
	Else, $d_{G'}(u,v)=d_G(u,v)\le D=1$. 
	Let $j$ be the unique index such that $u\in U_{j,\sigma}$. As $u\in M_{j,\sigma}$, it implies that there is no index $j'\ne j$ such that $v\in U^+_{j',\sigma}$. In particular, all the vertices in the shortest path between $u$ to $v$ in $G$ are in $U_{j,\sigma}$. Thus, we have
	\[
	\max_{u'\in f_{\sigma}(u),v'\in f_{\sigma}(v)}d_{H_{\sigma}}(u',v')\le\max_{u'\in f_{j,\sigma}(u),v'\in f_{j,\sigma}(v)}d_{H_{j,\sigma}}(u',v')\le d_{G_{j,\sigma}}(u,v)+\eps D=d_{G}(u,v)+\eps D~,
	\]
as desired.
\end{proof}

	Consider a $K_r$-minor-free graph $G$, and let $\mathbb{T}$ be its clique-sum decomposition. That is
	$G = \cup_{(G_i,G_j) \in E(\mathbb{T})}G_i \oplus_{h(r)} G_j$
	where each $G_i$ is a nearly $h(r)$-embeddable graph.  
	We call the clique involved in the clique-sum of $G_i$ and $G_j$ the \emph{joint set} of the two graphs. Here $h(r)$ is a function depending on $r$ only. 
	Let $\phi_h$ be some function depending only on $h$ such that the treewidth of the graphs constructed in \Cref{lem:Ramsey-almost-embeddable} is bounded by $\phi_h\cdot \frac{\log n}{\eps\cdot\delta}$.
	
	The embedding of $G$ is defined recursively, where some vertices from former levels will be added to future levels as apices. In order to  control the number of such apices, we will use the following concept. 
	\begin{definition}[Enhanced minor-free graph]\label{def:enhancedMinorFree}
		A graph $G$ is called $(r,s,t)$\emph{-enhanced minor free graph} if there is a set $S$ of  at most $s$ vertices, called elevated vertices, such that every elevated vertex $u\in S$ has edges towards all the other vertices and $G\setminus S$ is a $K_r$-minor-free graph that has a clique-sum decomposition with $t$ pieces.
	\end{definition}
	 We will prove the following claim by induction on $t$:
	\begin{lemma}\label{lem:ramsey-enhanced-minor-free}
		Given an $n$-vertex $(r,s,t)$-enhanced minor-free graph $G$ of diameter $D$ with a set $S$ of elevated vertices, and parameter $\eps\in(0,\frac14)$, 
 		there is a distribution over one-to-many, clique-preserving, dominating embeddings $f$ into graphs $H$ of treewidth $\phi_{h(r)}\cdot\frac{\log n}{\eps\cdot\delta} + s+h(r)\cdot \log t$, such that there is a subset $M\subseteq V$ of vertices for which the following hold:
		\begin{enumerate}
			\item For every $v\in V$, $\Pr[v\in M]\ge 1-\delta\cdot \log2t$. \label{property:RamseyEnhancedFailureProbability}
			\item For every $u\in M$ and $v\in V$, $\max_{u'\in f(u),v'\in f(v)}d_{H_{\sigma}}(u',v')\le d_G(u,v)+\eps D$.\label{property:RamseyEnhancedDistoriton}
		\end{enumerate}
	\end{lemma}
	W now show that \Cref{lem:ramsey-enhanced-minor-free} implies \Cref{thm:Ramsey-minor-Free}:
	\begin{proof}[Proof of \Cref{thm:Ramsey-minor-Free}]
	Note that every $K_r$-minor-free graph is $(r,0,n)$-enhanced minor free. Apply \Cref{lem:ramsey-enhanced-minor-free} using parameters $\eps$ and $\delta'=\frac{\delta}{\log 2n}$ to obtain a distribution of embeddings. For each embedding $f$ in the distribution, define another embedding $g$ by setting $g(v)$ for each $v\in V$ to be an arbitrary vertex from $f(v)$.
	We obtain a distribution over embeddings into treewidth
	$\phi_{h(r)}\cdot\frac{\log n}{\eps\cdot\delta'} + 0+h(r)\cdot \log n=O_r(\frac{\log^2 n}{\eps \delta})$ graphs with distortion $\eps D$, such that for every vertex $v\in V$, $\Pr[v\in M]\ge 1-\delta'\cdot \log2n=1-\delta$.
	\end{proof}

	\begin{proof}[Proof of \Cref{lem:ramsey-enhanced-minor-free}]
	It follows from \Cref{lem:Ramsey-almost-embeddable} that the claim holds for the base case $t=1$.
	We now turn to the induction step. Consider an $(r,s,t)$-enhanced minor-free graph $G$. Let $G'$ be a $K_r$-minor-free graph obtained from $G$ by removing a set $S$ of elevated vertices. Let $\mathbb{T}$ be the clique-sum decomposition of $G'$ with $t$ pieces.
	We use the following lemma to pick a central piece $\mathcal{G}$ of $\mathbb{T}$.	
	\begin{lemma}[\cite{Jordan69}]\label{lm:tree-sep} Given a tree $T$ of $n$ vertices, there is a vertex $v$ such that every connected component of $T\setminus \{v\}$  has at most $\frac{n}{2}$ vertices. 
	\end{lemma} 	
	Let $G_1,\dots,G_p$ be the neighbors of $\tilde{G}$ in $\mathbb{T}$. Note that $\mathbb{T}\setminus \tilde{G}$ contains $p$ connected components $\mathbb{T}_1,\dots,\mathbb{T}_p$, where $G_i\in \mathbb{T}_i$, and  $\mathbb{T}_i$ contains at most $|\mathbb{T}|/2=t/2$ pieces.
	Let $Q_i$ be the clique used in the clique-sum of $G_i$ with $\tilde{G}$ in $\mathbb{T}$.  
	For every $i$, we will add edges between $Q_i$ vertices to all the vertices in $\mathbb{T}_i$. That is, we add $Q_i$ to the set of elevated vertices in the graph induced by pieces in $\mathbb{T}_i$.  Every new edge $\{u,v\}$ will have the weight $d_G(u,v)$.
	Let $\mathcal{G}_i$ be the graph induced on vertices of $\mathbb{T}_i\cup S$ (and the newly added edges). Note that $\mathcal{G}_i$ is an $(r,s',t')$-enhanced minor-free graph for $t'\le \frac t2$ and $|s'|\le |S|+|Q_i|\le s+h(r)$. Furthermore, for every $u,v\in  \mathcal{G}_i$, it holds that $d_{\mathcal{G}_i}(u,v)=d_{G}(u,v)$. Thus, each $\mathcal{G}_i$ has diameter at most $D$.
	Using the inductive hypothesis on $\mathcal{G}_i$, we sample a dominating embedding $f_i$ into $H_i$, and a subset $M_i\subseteq\mathcal{G}_i$ of vertices. Note that properties (\ref{property:RamseyEnhancedFailureProbability})-(\ref{property:RamseyEnhancedDistoriton}) hold, and $H_i$ has treewidth $\phi_{h(r)}\cdot\frac{\log n}{\eps\cdot\delta} + s'+h(r)\cdot \log 2t'\le \phi_{h(r)}\cdot\frac{\log n}{\eps\cdot\delta} + s+h(r)\cdot \log 2t$.
	
	Let $\tilde{\mathcal{G}}$ be the graph induced on $\tilde{G}\cup S$. Note that $\tilde{\mathcal{G}}$ has diameter at most $D$. We apply \Cref{lem:Ramsey-almost-embeddable} to $\tilde{\mathcal{G}}$ to sample a dominating embedding $\tilde{f}$ into $\tilde{H}$, and a subset $\tilde{M}$ of vertices. Note that properties (\ref{property:MLargeAlmostEmbeddable})-(\ref{property:MSmallDistortionAlmostEmbeddable}) hold, 
	in particular, the treewidth of $\tilde{H}$ is bounded by $\phi_{h(r)}\cdot\frac{\log n}{\eps\cdot\delta} + s$ (as the construction first will delete the elevated vertices and eventually add them to all the bags).
	
	As the embeddings $\tilde{f},f_1,\dots,f_p$ are clique-preserving embeddings into  $\tilde{H},H_1,\dots,H_p$, there is a natural way to combine them into a single graph $H$ of treewidth $\phi_{h(r)}\cdot\frac{\log n}{\eps\cdot\delta} + s+h(r)\cdot \log 2t$.
	In more detail, initially, we just take a disjoint union of all the graphs $\tilde{H},H_1,\dots, H_p$, keeping all copies of the different vertices separately. Next, we identify all the copies of the elevated vertices. Finally, for each $i$, as both $\tilde{f}$ and $f_i$ are clique-preserving, we simply take two clique copies of $Q_i$ from  $\tilde{f}$ and $f_i$, and identify the respective vertices in this two clique copies. Note that every vertex $v\in Q_i$ is an elevated vertex in $\mathcal{G}_i$, and thus $f_i(v)$ is unique.
	The embedding $f$ is defined as follows: 
	For $v\in\tilde{\mathcal{G}}$, set $f(v)=\tilde{f}(v)$, while for $v\in\mathcal{G}_i\setminus \tilde{\mathcal{G}}$ for some $i$, set $f(v)=f_i(v)$.
	
	We now define the subset $M\subseteq V$. Every vertex $v\in\tilde{M}$ joins $M$. A vertex $v\in \mathcal{G}_i\setminus \tilde{\mathcal{G}}$ join $M$ if and only if $v\in M_i$ and $Q_i\subseteq \tilde{M}$. Note that for vertices in $\tilde{\mathcal{G}}$, property (\ref{property:RamseyEnhancedFailureProbability}) holds trivially, while for $v\in \mathcal{G}_i\setminus\tilde{\mathcal{G}}$, using the induction hypothesis and union bound
	\[
	\Pr\left[v\notin M\right]\le\Pr\left[v\notin M_{i}\right]+\Pr\left[Q_{i}\nsubseteq\tilde{M}\right]\le\delta\cdot\log2t'+\delta\le\delta\cdot\log2t~.
	\]
	Hence property (\ref{property:RamseyEnhancedFailureProbability}) holds.
	Note that $f$ is clique-preserving as every clique must be contained in either $\tilde{\mathcal{G}}$ or some $\mathcal{G}_i$.	
	Finally, we show that property (\ref{property:RamseyEnhancedDistoriton}) holds. Consider a vertex $u\in M$ and $v\in V$. We proceed by case analysis.
	\begin{itemize}
		\item If a shortest path from $u$ to $v$ goes through a vertex $z\in S$ (this includes the case where either $u$ or $v$ is in $S$). Then for every $u'\in f(u)$ and $v'\in f(v)$, it holds that  $d_{H}(u',v')\le d_{H}(u',f(z))+d_{H}(f(z),v')=d_{G}(u,z)+d_{G}(z,v)=d_{G}(u,v)$.
		\item Else, if both $u,v\in \tilde{G}$, then by \Cref{lem:Ramsey-almost-embeddable}, $\max_{u'\in f(u),v'\in f(v)}d_{H}(u',v')\le\max_{u'\in\tilde{f}(u),v'\in\tilde{f}(v)}d_{\tilde{H}}(u',v')\le d_{\mathcal{\tilde{G}}}(u,v)+\eps D=d_{G}(u,v)+\eps D$.
		\item Else, if there is an $i\in[p]$ such that both $u,v\in \mathcal{G}_i\setminus\tilde{G}$, then by the induction hypothesis
		$\max_{u'\in f(u),v'\in f(v)}d_{H}(u',v')\le\max_{u'\in f_{i}(u),v'\in f_{i}(v)}d_{H_{i}}(u',v')\le d_{\mathcal{G}_{i}}(u,v)+\eps D=d_{G}(u,v)+\eps D$.
		\item Else, if $u\in \tilde{G}$ and there is an $i\in[p]$ such that $v\in \mathcal{G}_i$.
		There is necessarily a vertex $x\in Q_i$ such that there is a shortest path from $u$ to $v$ in $G$ going through $x$. 
		Let $\hat{x}$ be the copy of $x$ used to connect between $\tilde{H}$ and $H_i$. Note that there is an edge between $\hat{x}$ to every copy $v'\in f_i(v)$ in $H_i$. In addition, as $u\in \tilde{M}$, by the second case it holds that $\max_{u'\in f(u)}d_{H}(u',\hat{x})\le\max_{u'\in f(u),x'\in f(x)}d_{H}(u',\hat{x})\le d_{G}(u,x)+\eps D$. We conclude
		\begin{align}
		\max_{u'\in f(u),v'\in f(v)}d_{H}(u',v') & 
		\le\max_{u'\in f(u)}d_{H}(u',\hat{x})+\max_{v'\in f(v)}d_{H}(\hat{x},v')\nonumber\\
		& \le d_{G}(u,x)+\eps D+d_{G}(x,v)=d_{G}(u,v)+\eps D~.\label{eq:Ramsey-stretch}
		\end{align}
		\item  Else, if $v\in \tilde{G}$ and there is an $i\in[p]$ such that $u\in \mathcal{G}_i\setminus\tilde{G}$.
		There is necessarily a vertex $x\in Q_i$ such that there is a shortest path from $u$ to $v$ in $G$ going through $x$. 
		As $u\in M$ it follows that $x\in \tilde{M}\subseteq M$. Let $\hat{x}$ be the copy of $x$ used to connect between $\tilde{H}$ and $H_i$; we observe that inequality (\ref{eq:Ramsey-stretch}) holds in this case.
		\item  Else, there are $i\ne j$ such that $u\in \mathcal{G}_i\setminus\tilde{G}$ and $v\in \mathcal{G}_j\setminus\tilde{G}$. There is necessarily a vertex $x\in Q_i$ such that there is a shortest path from $u$ to $v$ in $G$ going through $x$. As $u\in M$ it follows that $x\in \tilde{M}\subseteq M$. Let $\hat{x}$ be the copy of $x$ used to connect between $\tilde{H}$ and $H_i$. By the forth case, it holds that $\max_{x'\in f(x),v'\in f(v)}d_{H}(x',v')\le d_{G}(x,v)+\eps D$. Thus, \begin{align*}
		\max_{u'\in f(u),v'\in f(v)}d_{H}(u',v') & \le\max_{u'\in f(u)}d_{H}(u',\hat{x})+\max_{v'\in f(v)}d_{H}(\hat{x},v')\\
		& \le d_{G}(u,x)+d_{G}(x,v)+\eps D=d_{G}(u,v)+\eps D~.
		\end{align*}
	\end{itemize}
\end{proof}

%% file: clan.tex
\section[Clan Embedding for Minor-Free Graphs (\Cref{thm:Clan-Embedding})]{Clan Embedding for Minor-Free Graphs}\label{sec:ClanMinorFree}
This section is devoted to proving \Cref{thm:Clan-Embedding} (restated below for convenience).
The proof of \Cref{thm:Clan-Embedding} builds upon a similar approach to \Cref{thm:Ramsey-minor-Free}, however, it is more delicate and considerably more involved. We present the proof here without assuming familiarity with the proof of \Cref{thm:Ramsey-minor-Free}. Nonetheless, 
we recommend the reader to first understand the proof of \Cref{thm:Ramsey-minor-Free} before reading this section.
\Clan*

\begin{remark}
	Note that \Cref{thm:Clan-Embedding} implies a weak version \Cref{thm:Ramsey-minor-Free}, where the distortion guarantee is for pairs $u,v\in M$ rather than than for $u\in M$ and $v\in V$: simply use the chief part $\chi$ as a Ramsey type embedding and set $M=\{v\mid |f(v)|=1\}$. Interestingly, this weaker version is still sufficient for our application to the metric $\rho$-independent set problem (\Cref{thm:isolated}). 
\end{remark}

	We begin with \Cref{lem:Clan-AlmostEmbedable}, which is a special case of nearly-embeddable graphs. Later, we will generalize to minor-free graphs via clique-sums. Specifically, inductively we will use \Cref{lem:Clan-AlmostEmbedable} for each piece, and integrate it to the general embedding. However, for this integration to go through, we will need the intermediate embedding to be clique-preserving. As a consequence, we will not attempt to bound the size of $f$ directly. Instead, for every vertex $v$, $f(v)$ will be the union of two sets $\chi(v)$ and $\psi(v)$. Eventually, for the clan embedding, we will take \emph{one copy from each set}. We will say that the embedding succeeds on a vertex $v$ if $\psi(v)=\emptyset$. (In the following lemma, $\bigcupdot$ denotes the disjoint-union operation.)

	\begin{lemma}\label{lem:Clan-AlmostEmbedable}
	Consider a nearly $h$-embeddable $n$-vertex graph $G=(V,E,w)$ with set of apices $\Psi$, diameter $D$, and parameters $\eps\in(0,\frac14)$, $\delta\in(0,1)$.
	Then there is a distribution over one-to-many, dominating embeddings $f$ into treewidth $O_{h}(\frac{\log n}{\eps\delta})$ graphs, such that for every vertex $v\in V$, $f(v)$ can be partitioned into sets $\chi(v),\psi(v)$ where $\chi(v)\bigcupdot \psi(v)= f(v)$.%\htodo{is $\bigcupdot$ disjoint-union?}\atodo{Yes. I though it should be clear as I said partitioned. Should we say it explicitly?}. 
	It holds that:
	\begin{enumerate}
		\item\label{property:ClanAlmostEmbeddableDistortion} For every pair of vertices $u,v$, \footnote{Note that $\psi(v)$ might be an empty set. A maximum over an empty set is defined to be $\infty$.}
		\begin{equation}
		\min\left\{ \max_{u'\in\chi(u),v'\in\chi(v)}d_{H}(u',v'),\max_{u'\in\psi(u),v'\in\chi(v)}d_{H}(v',u')\right\} \le d_{G}(u,v)+\eps D~.\label{eq:LemmaDominatingStretchGurantee}
		\end{equation}
		\item\label{property:ClanAlmostEmbeddableFailProbability} We say that $f$ fails on a vertex $v$ if $\psi(v)\ne \emptyset$.
		For a clique $Q\subseteq V$, we say that $f$ fails on $Q$ if it fails on some vertex in $Q$. For every clique $Q\subseteq V$, $\Pr[f\mbox{ fails on }Q]\le \delta$. 
		\item\label{property:ClanAlmostEmbeddableCliquePreserve} Consider a clique $Q$, one of the following holds:
		\begin{enumerate}
			\item $f$ succeeds on $Q$. In particular $\chi(Q)$ contains a clique copy of $Q$.
			\item $f$ fails on $Q$, and $\chi(Q)$ contains a clique copy of $Q$. In addition, consider the set\\ $Q^F=\{v\in Q\mid\psi(v)\ne0\}$, then $\psi(Q^F)$ contains a clique copy of $Q^F$.
			\item $f$ fails on $Q$, and $f(Q)$ contains two cliques copies $Q^1,Q^2$ of $Q$ such that for every vertex $v\in Q\setminus \Psi$, both $\chi(v)\cap(Q^1\cup Q^2)$ and $\psi(v)\cap(Q^1\cup Q^2)$ are singletons. In this case, in addition to equation (\ref{eq:LemmaDominatingStretchGurantee}), it also holds that for every $u\in V$ and $v\in Q\setminus \Psi$,
			\begin{equation}
			\min\left\{ \max_{u'\in\chi(u),v'\in\psi(v)}d_{H}(u',v'),\max_{u'\in\psi(u),v'\in\psi(v)}d_{H}(u',v')\right\} \le d_{G}(u,v)+\eps D~.\label{eq:LemmaDominatingStretchGuranteePSI}
			\end{equation}
		\end{enumerate}
	\end{enumerate}
	\end{lemma}
\begin{proof}
	Consider a nearly $h$-embedded graph $G=(V,E,w)$. Assume w.l.o.g. that $D=1$, otherwise we can scale accordingly. 
	We assume that $1/\delta$ is an integer, otherwise we solve for $\delta'$ such that $\frac{1}{\delta'}=\lceil\frac{1}{\delta}\rceil$.
	We will construct $q=\frac{8}{\delta}$ embeddings satisfying \propref{property:ClanAlmostEmbeddableDistortion} of \Cref{lem:Clan-AlmostEmbedable}.	
	The final embedding will be obtained by choosing one of these embeddings uniformly at random.
	Denote by $G'=G[V\setminus\Psi]$ the induced subgraph obtain by removing the apices. 
	In the tree decomposition of $H$ we will construct, the set $\Psi$ will belong to all the bags (with edges towards all the vertices). Thus we can assume that $G'$ is connected, as otherwise we can simply solve the problem on each connected component separately, and combine the solutions by taking the union of all graphs/embeddings.
	
	Let $r\in G'$ be an arbitrary vertex. For $\sigma\in\left\{ 4,\dots,\frac{8}{\delta}\right\}$, 
	set $\mathcal{I}_{-1,\sigma}=[0,\sigma]$, $\mathcal{I}_{-1,\sigma}^{+}=[0,\sigma+2]$,
	and for $j\ge0$, set $\mathcal{I}_{j,\sigma}=\left[\frac{8j}{\delta}+\sigma,\frac{8(j+1)}{\delta}+\sigma\right)$, and
	$\mathcal{I}_{j,\sigma}^{+}=\left[\frac{8j}{\delta}+\sigma-2,\frac{8(j+1)}{\delta}+\sigma+2\right)$.
	Set $U_{j,\sigma}=\left\{ v\in G'\mid d_{G'}(r,v)\in\mathcal{I}_{j,\sigma}\right\} $
	and similarly $U_{j,\sigma}^{+}$ w.r.t. $I_{j,\sigma}^{+}$. 
	Note that by the triangle inequality, for every pair of neighboring vertices $u,v$ it holds that $d_G(u,v)\le D=1$; thus, $u\in U_{j,\sigma}$ implies $v\in U^+_{j,\sigma}$.
	Let $G_{j,\sigma}$ be the graph
	induced by $U^{+}_{j,\sigma}$, plus the vertex $r$.
	In addition, we add edges from the vertex $r$ towards all the vertices with neighbors in $(\cup_{q<j}U^{+}_{q,\sigma})\setminus U^{+}_{j,\sigma}$ (where the weight of a new edge $\left(r,v\right)$ is $d_G(r,v)$).
	Equivalently, $G_{j,\sigma}$  can be constructed by taking the graph induced by $\cup_{j'\le j}U^{+}_{j',\sigma}$ and contracting all the internal edges out of $U_{j,\sigma}^{+}$ into $r$. 
	Note that all the edges towards $r$ have weight at most $D=1$.
	Furthermore, for every vertex $v\in G_{j,\sigma}$, $d_{G_{j,\sigma}}(v,r)< 1+\frac{8}{\delta}+4$. Thus $G_{j,\sigma}$ is a nearly $h$-embedded graph with diameter at most $\frac{16}{\delta}+10=O(\frac{1}{\delta})$ and no apices. See \Cref{fig:GraphLayers} for an illustration.
	
	\begin{figure}[t]
		\centering
		\includegraphics[scale=0.38]{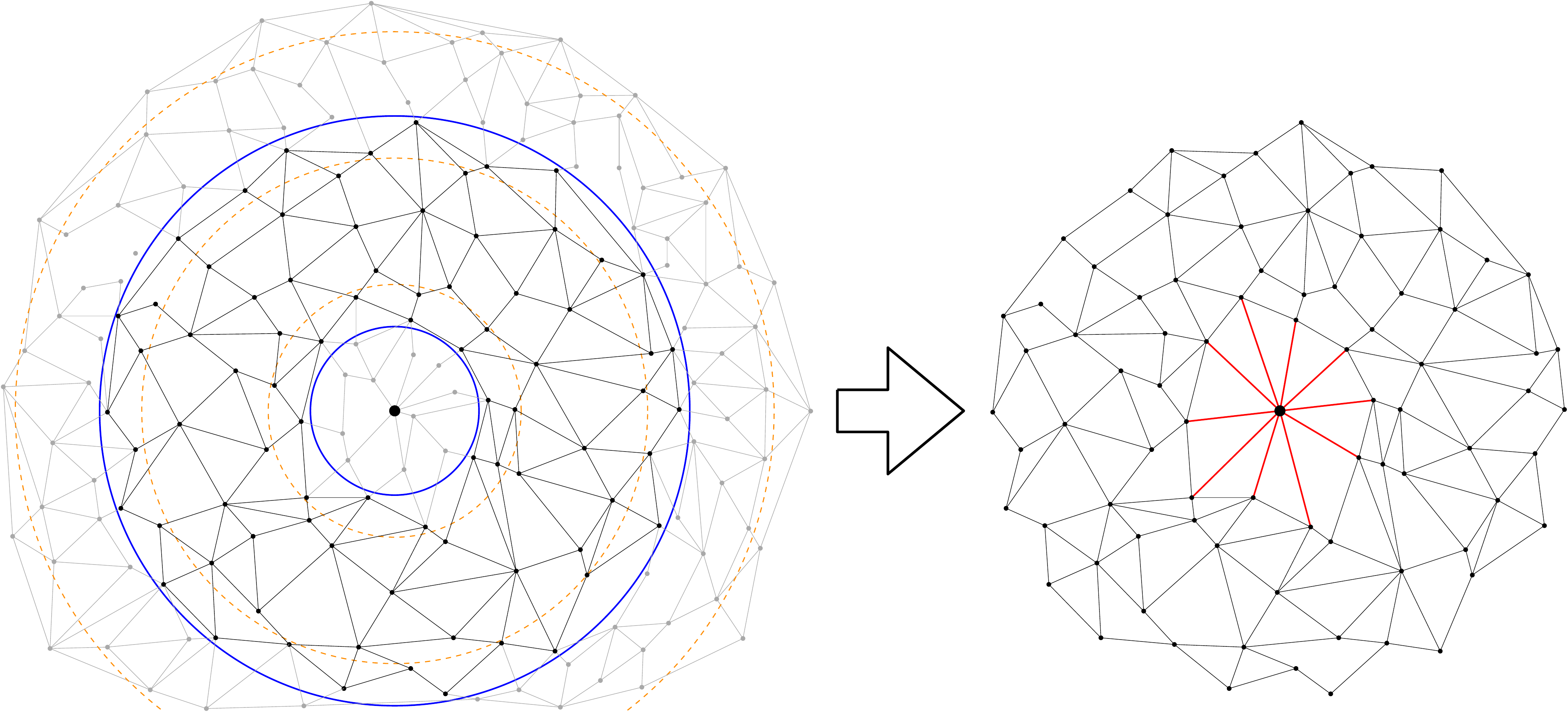}
		\caption{\small On the left is the graph $G'$. $r$ is the big black vertex in the middle. The dashed orange lines separate between the layers of $U_{-1,\sigma},U_{0,\sigma},U_{1,\sigma},\dots$. The two blue lines are the boundaries of $U^+_{0,\sigma}$. All the vertices in $U^+_{0,\sigma}$ (and the edges between them) are black, while all other vertices (and the edges incident on them) are gray. 
		On the right is the graph $G_{0,\sigma}$ with vertex set $U^+_{0,\sigma}\cup\{r\}$, where the edges added from $r$ to vertices with neighbors in $U^+_{-1,\sigma}\setminus U^+_{0,\sigma}$ are marked in red.
		}
		\label{fig:GraphLayers}
	\end{figure}

	Fix some $\sigma$ and $j$. Using \Cref{lm:embed-genus-vortex} with parameter $\Theta(\eps\cdot\delta)$,
	we construct a dominating one-to-many embedding $f_{j,\sigma}$, of $G_{j,\sigma}$ into a graph $H_{j,\sigma}$ with treewidth $O_{h}(\frac{\log n}{\eps\cdot\delta})$, such that $f_{j,\sigma}$ is clique preserving and has additive distortion $\Theta(\eps\cdot\delta)\cdot O(\frac1\delta)=\eps$.
	After the application of \Cref{lm:embed-genus-vortex}, we will add edges from $r$ to all the other vertices (where the weight of a new edge $\left(r,v\right)$ is $d_G(r,v)$). Note that this increases the treewidth by at most $1$. Further, we will assume that there is a bag containing only the vertex $r$ (as we can simply add such a bag). 
	Next, fix $\sigma$. Let $H'_{\sigma}$ be a union of the graphs $\cup_{j\ge -1} H_{j,\sigma}$. We identify the vertex $r$ with itself, but all copies of other vertices that participate in more that a single graph will remain separate. Formally, we define a one-to-many embedding $f_\sigma$, where $f_\sigma(r)$ equals to the unique vertex $r$, and for every other vertex $v\in V\setminus \Psi$, 	$f_\sigma(v)=\bigcup_{j\ge -1} f_{j,\sigma}(v)$. 
	Note that $H'_{\sigma}$ has a tree decomposition of width $O_{h}(\frac{\log n}{\eps\cdot\delta})$, by identifying the bag containing only $r$ in all the graphs. 
	Finally, we create the graph $H_{\sigma}$ by adding the set $\Psi$ with edges towards all the vertices in $H'_{\sigma}$,
	where the weight of a new edge $\left(u',v\right)$ for $u\in f_\sigma(u)$ and $v\in \Psi$ is $d_{G}(u,v)$. For $v\in \Psi$, set $f_\sigma(v)=\{v\}$. 
	As $\Psi=O_h(1)$, $H_{\sigma}$ has treewidth $O_{h}(\frac{\log n}{\eps\cdot\delta})$.
	The one-to-many embedding $f_\sigma$ is dominating. This follows by the triangle inequality as every edge $\{u,v\}$ in the graph has weight $d_G(u,v)$. 
	Finally, the embedding $f$ is chosen to equal $f_\sigma$, for $\sigma$ chosen uniformly at random. This concludes the definition of the embedding $f$.
	
	Next, we define the partition $\chi_{\sigma}(v)\bigcupdot \psi_\sigma(v)$ of $f_\sigma(v)$ for each vertex $v\in V$ as follows:
	\begin{itemize}
		\item If $v\in \Psi\cup\{r\}$, then there is a single copy of $v$ in $f_\sigma$. Set $\chi_\sigma(v)=f_\sigma(v)$ and $\psi_{\sigma}(v)=\emptyset$.
		\item Else, let $j$ be the unique index such that $v\in U_{j,\sigma}$. Set $\chi_\sigma(v)=f_{j,\sigma}(v)$. If there is another index $j'$ such that $v\in U^+_{j',\sigma}$, set $\psi_\sigma(v)=f_{j',\sigma}(v)$, otherwise set $\psi_{\sigma}(v)=\emptyset$.
	\end{itemize} 
	Clearly, as there are at most $2$ indices $j$ such that $v\in U^{+}_{j,\sigma}$,  $\chi_{\sigma}(v)\bigcupdot \psi_\sigma(v)=f_\sigma(v)$. 
	
	Next, we prove \propref{property:ClanAlmostEmbeddableDistortion}- the stretch bound. Consider a pair of vertices $u,v\in V$. If $v\in\Psi\cup\{r\}$ then $f_\sigma(v)$ is a singleton with an edge towards every copy of $u$, thus \propref{property:ClanAlmostEmbeddableDistortion} holds. The same argument holds also if $u\in \Psi\cup\{r\}$.
	Otherwise, if $d_{G'}(u,v)>d_{G}(u,v)$, then the shortest path between $u$ to $v$ in $G$ goes through an apex vertex $z\in \Psi$. In particular, $f_\sigma(z)$ is a singleton with an edge towards every other vertex. It follows that  in $H_\sigma$, the distance between every two copies in $f_{\sigma}(v)$ and $f_{\sigma}(u)$ is exactly $d_{G}(u,z)+d_{G}(z,v)=d_{G}(u,v)$.	
	Else, $d_{G'}(u,v)=d_{G}(u,v)$.
	Let $j$ be the unique index such that $v\in U_{j,\sigma}$, then $u\in U^{+}_{j,\sigma}$.
	Furthermore, $d_{G_{j,\sigma}}(u,v)=d_{G'}(u,v)$ since the entire shortest path between them is in $U^{+}_{j,\sigma}$.
	By \Cref{lm:embed-genus-vortex},
	\begin{align*}
	& \min\left\{ \max_{u'\in\chi_{\sigma}(u),v'\in\chi_{\sigma}(v)}d_{H_{\sigma}}(u',v'),\max_{u'\in\psi_{\sigma}(u),v'\in\chi_{\sigma}(v)}d_{H_{\sigma}}(v',u')\right\} \\
	& \qquad\le\max_{u'\in f_{j,\sigma}(u),v'\in f_{j,\sigma}(v)}d_{H_{j,\sigma}}(u',v')~\le~ d_{G_{j,\sigma}}(u,v)+\eps D~=~d_{G'}(u,v)+\eps D~=~d_{G}(u,v)+\eps D~.
	\end{align*}
	Next we argue \propref{property:ClanAlmostEmbeddableFailProbability}- the failure probability of a clique. Recall that $f,\chi,\psi$ will equal to $f_\sigma,\chi_\sigma,\psi_\sigma$ for $\sigma\in \{1,\dots,\frac8\delta\}$ chosen uniformly at random.  
	Consider some clique $Q$, we can assume w.l.o.g. that $Q$ does not contain any apex vertices (as $f$ never fails on apex vertex). Let $s_Q,t_Q\in Q$ be the closest and farthest vertices from $r$ in $G'$, respectively.
	Then $d_{G'}(r,t_Q)-d_{G'}(r,s_Q)\le d_{G'}(s_Q,t_Q)\le D=1$. $f_\sigma$ fails on $Q$ iff there is a non-empty intersection between the interval $[d_{G'}(r,t_Q),d_{G'}(r,s_Q))$ (of length at most $1$) and interval $[\frac{8j}{\delta}+\sigma-2,\frac{8j}{\delta}+\sigma+2)$ for some $j$. Note that there are at most $5$ choices of $\sigma$ on which this happens. We conclude that $\Pr[f\mbox{ fails on }Q]\le\frac{5}{\nicefrac{8}{\delta}-3}\le\delta$.

	\begin{figure}[t]
		\floatbox[{\capbeside\thisfloatsetup{capbesideposition={left,top},capbesidewidth=7.3cm}}]{figure}[\FBwidth]
		{\caption{\footnotesize Illustration of the different cases in  \propref{property:ClanAlmostEmbeddableCliquePreserve}. The green area marks all the  vertices in $U^+_{j,\sigma}$. The vertices in $U_{j,\sigma}$ are enclosed between the two black semicircles. The vertices in $U^+_{j,\sigma}\cap U^+_{j+1,\sigma}$ (resp. $U^+_{j-1,\sigma}\cap U^+_{j,\sigma}$) are enclosed between the red (resp. orange) dashed semicircles.
				In the first case \textbf{(a)}, all the vertices of $Q$ are in $U_{j,\sigma}$ and no vertex failed. 
				In the second case \textbf{(b)}, all the vertices of $Q$ are in $U_{j,\sigma}$ and some  vertices failed.
				In the third case \textbf{(c)}, the vertices of $Q$ non-trivially partitioned between $U_{j,\sigma}$ and $U_{j+1,\sigma}$, and all of them failed.
			}\label{fig:ClanLemmaCliqueCases}}
		{
			\includegraphics[width=8.4cm]{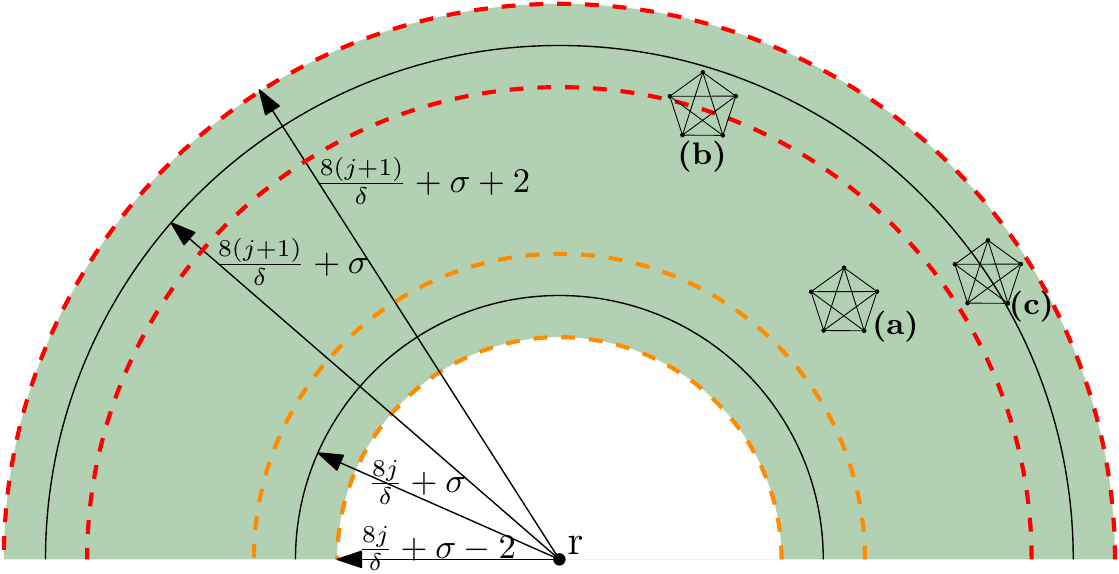}
		}
	\end{figure}
	Finally, we prove \propref{property:ClanAlmostEmbeddableCliquePreserve}- clique preservation. Consider a clique $Q$, note that we can assume that $Q\subseteq G'$, as $f_\sigma$ will not fail on any apex. Furthermore, if $r\in Q$, then no vertex in $Q$ fails as $Q\subseteq B_{G'}(r,1)\subseteq U_{-1,\sigma}\setminus U^+_{0,\sigma}$. Thus we can assume that $r\notin Q$. We proceed by case analysis; the cases are illustrated in \Cref{fig:ClanLemmaCliqueCases}.
	\begin{itemize}
		\item[(a)] if $f_\sigma$ succeeds on $Q$, then $f_\sigma(Q)=\chi_\sigma(Q)$. In particular there is a unique $j$ such that $Q\subseteq  U_{j,\sigma}$.  As $f_{j,\sigma}$ is clique-preserving, it contains a clique copy of $Q$. In particular,   $\chi_\sigma(Q)$ contain a clique copy of $Q$.
	\end{itemize}	
		Otherwise, $f_\sigma$ fails on $Q$. Then, there is a unique index $j$ such that the intersection of $Q$ with both $U^+_{j,\sigma}$ and $U^+_{j+1,\sigma}$ is non-empty.
	\begin{itemize}
		\item[(b)]  First, consider the case that $Q\subseteq U_{j,\sigma}$ (the case $Q\subseteq U_{j+1,\sigma}$ is symmetric). Here $\chi_\sigma(Q)=f_{j,\sigma}(Q)$, and $\psi_\sigma(Q)=\psi_\sigma(Q^F_\sigma)=f_{j+1,\sigma}(Q^F_\sigma)$, where $Q^F_\sigma=\{v\in Q\mid\psi_\sigma(v)\ne0\}$. As $f_{j,\sigma}$ and $f_{j+1,\sigma}$ are clique-preserving, $\chi_\sigma(Q)$ contain a clique copy of $Q$, while $\psi_\sigma(Q^F_\sigma)$ contains a clique copy of $Q^F_\sigma$.
		\item[(c)]	Finally, consider the case where $Q$ intersect both $U_{j,\sigma}$ and $U_{j+1,\sigma}$. It holds that $d_{G'}(r,s_Q)<\frac{8(j+1)}{\delta}+\sigma\le d_{G'}(r,t_Q)$, hence $\frac{8(j+1)}{\delta}+\sigma-1\le d_{G'}(r,s_Q)$ and $d_{G'}(r,t_Q)< \frac{8(j+1)}{\delta}+\sigma+1$ (here $s_Q,t_Q\in Q$ are the closest and farthest vertices from $r$, respectively).
		Necessarily, $Q\subseteq U^+_{j,\sigma}\cap U^+_{j+1,\sigma}$. In particular, as $f_{j,\sigma}(Q)$, and $f_{j+1,\sigma}(Q)$ are clique-preserving, they contain clique copies $Q_1,Q_2$ of $Q$ (respectively). Furthermore, $Q_1,Q_2\subseteq f_\sigma(Q)$, and for every vertex $v\in Q$, both $\chi(v)\cap(Q^1\cup Q^2)$ and $\psi(v)\cap(Q^1\cup Q^2)$ are singletons.\\	
		It remains to prove the additional stretch guarantee. Consider a vertex $v\in Q$, suppose that $v\in  U_{j,\sigma}$ (the case $v\in  U_{j+1,\sigma}$ is symmetric).
		Here $\chi_\sigma(v)=f_{j,\sigma}(v)$ and $\psi_\sigma(v)=f_{j+1,\sigma}(v)$.
		Consider some vertex $u\in V$,
		in similar manner to the general distortion argument, if either $u\in\Psi\cup\{r\}$, or the shortest path from $u$ to $v$ in $G$ goes through $\Psi\cup\{r\}$,  then the distance between every two copies in $f_\sigma(v)$ and $f_\sigma(u)$ is exactly $d_{G}(u,v)$, and hence equation (\ref{eq:LemmaDominatingStretchGuranteePSI}) holds. 
		Else, $d_{G'}(u,v)=d_{G}(u,v)$, and it holds that $d_{G'}(r,u)\ge d_{G'}(r,v)-d_{G'}(u,v)\ge d_{G'}(r,s_Q)-1\ge\frac{8(j+1)}{\delta}+\sigma-2$, thus $u\in U^+_{j+1,\sigma}$.
		Furthermore, $d_{G_{j+1,\sigma}}(u,v)=d_{G'}(u,v)$ (as the entire shortest path between them is in $U^{+}_{j+1,\sigma}$).
		By \Cref{lm:embed-genus-vortex},
		\begin{align*}
		& \min\left\{ \max_{u'\in\chi(u),v'\in\psi(v)}d_{H}(u',v'),\max_{u'\in\psi(u),v'\in\psi(v)}d_{H}(u',v')\right\} \\
		& \quad\le\max_{u'\in f_{j+1,\sigma}(u),v'\in f_{j+1,\sigma}(v)}d_{H_{j+1,\sigma}}(u',v')\,\le\,d_{G_{j+1,\sigma}}(u,v)+\epsilon D\,=\,d_{G'}(u,v)+\epsilon D\,=\,d_{G}(u,v)+\eps D~.
		\end{align*}
	\end{itemize}
\end{proof}

Consider a $K_r$-minor-free graph $G$, and let $\mathbb{T}$ be its clique-sum decomposition. That is
$G = \cup_{(G_i,G_j) \in E(\mathbb{T})}G_i \oplus_h G_j$
where each $G_i$ is a nearly $h(r)$-embeddable graph.  
We call the clique involved in the clique-sum of $G_i$ and $G_j$ the \emph{joint set} of the two graphs.
Let $\phi_h$ be some function depending only on $h$ such that the treewidth of the graphs constructed in \Cref{lem:Clan-AlmostEmbedable} is bounded by $\phi_h\cdot \frac{\log n}{\eps\cdot\delta}$.
The embedding of $G$ is defined recursively, where some vertices from former levels will be added to future levels as apices. In order to control the number of such apices, we will use the concept of enhance minor-free graphs introduced in \Cref{def:enhancedMinorFree} in \Cref{sec:RamseyMinorFree}.
We will prove the following lemma by induction on $t$:

\begin{lemma}\label{lem:Clan-enhanced-minor-free}
	Given an $(r,s,t)$-enhanced minor-free graph $G$ of diameter $D$ with a specified set $S$ of elevated vertices, and parameters $\eps\in(0,\frac14)$,$\delta\in(0,1)$, 
	there is a distribution over one-to-many, clique-preserving, dominating embeddings $f$ into graphs of treewidth $\phi_{h(r)}\cdot\frac{\log n}{\eps\cdot\delta} + s+h(r)\cdot \log t$, 
	such that for every vertex $v\in V$, $f(v)$ can be partitioned into sets $g_1(v),g_2(v),\dots$ where $\bigcupdot_{j\ge1} g_j(v)= f(v)$. Furthermore, 
	\begin{enumerate}
		\item\label{prop:enhancedClmSparsity} For every $v\in V$, let $q_v$ be the maximal index $j$ such that $g_{j}(v)\ne\emptyset$, then
	 $\mathbb{E}[q_v]\le (1+\delta)^{\log2t}$. 
		In addition, if $v\in S$ then $|f(v)|=1$ and thus $q_v=1$.
		\item\label{prop:enhancedClmDistortion} 
		For every pair of vertices $u,v$, $\min_{j}\max_{u'\in g_{j}(u),v'\in g_{1}(v)}d_{H}(u',v')\le d_{G}(u,v)+\eps D$. 	
	\end{enumerate}
\end{lemma}

Assuming \Cref{lem:Clan-enhanced-minor-free}, \Cref{thm:Clan-Embedding} easily follows.
\begin{proof}[Proof of \Cref{thm:Clan-Embedding}]
	Note that every $K_r$-minor-free graph is  $(r,0,n)$-enhanced minor-free. We apply \Cref{lem:Clan-enhanced-minor-free} using parameters $\eps$ and $\delta'=\frac{\delta}{2\log 2n}$.
	For every vertex $v\in V$, let $g(v)\subseteq f(v)$ be a set containing a single copy from each non empty set $g_j(v)$. Let $\chi(v)=g(v)\cap g_1(v)$ be the copy in $g(v)$ from $g_1(v)$.
	The distortion guarantee is straightforward to verify. 
	The treewidth of the resulting graph is $\phi_{h(r)}\cdot\frac{\log n}{\eps\cdot\delta'}+0+h(r)\cdot\log n=O_{r}(\frac{\log^{2}n}{\epsilon^{2}})$.
	Finally, for every vertex $v\in V$, it holds that $\mathbb{E}[|g(v)|]\le (1+\frac{\delta}{2\log 2n})^{\log2n}< e^{\frac\delta2}<1+\delta$.
\end{proof}
The rest of the section is devoted to proving \Cref{lem:Clan-enhanced-minor-free}.
\begin{proof}[Proof of \Cref{lem:Clan-enhanced-minor-free}]
	The claim is proved by induction on $t$.
	It follows from \Cref{lem:Clan-AlmostEmbedable} that \Cref{lem:Clan-enhanced-minor-free} holds for the base case $t=1$; the treewidth will be $\phi_{h(r)}\cdot\frac{\log n}{\eps\cdot\delta} + s$ since we add all elevated vertices to every bag.

	We now turn  to the induction step. Consider an $(r,s,t)$-enhanced minor-free graph $G$. Let $G'$ be a $K_r$-minr-free graph obtained from $G$ by removing the set $S$ (of size at most $s$). Let $\mathbb{T}$ be the clique-sum decomposition of $G'$ with $t$ pieces. 	
	Using \Cref{lm:tree-sep}, choose a central piece $\tilde{G}\in\mathbb{T}$ of $\mathbb{T}$.	
	Let $G_1,\dots,G_p$ be the neighbors of $\tilde{G}$ in $\mathbb{T}$. Note that $\mathbb{T}\setminus \tilde{G}$ contains $p$ connected components $\mathbb{T}_1,\dots,\mathbb{T}_p$, where $G_i\in \mathbb{T}_i$, and  $\mathbb{T}_i$ contains at most $|\mathbb{T}|/2=t/2$ pieces. 
	Let $Q_i$ be the clique used in the clique-sum of $G_i$ with $\tilde{G}$ in $\mathbb{T}$. 
	For every $i$, we will add edges between vertices of $Q_i$ to all the vertices in $\mathbb{T}_i$; this is equivalent to making $Q_i$ into apices. Every new edge $\{u,v\}$ will have weight $d_G(u,v)$.
	Let $\mathcal{G}_i$ be the graph induced on the vertices of $\mathbb{T}_i\cup S$ (and the newly added edges). Note that $\mathcal{G}_i$ is an $(r,s',t')$-enhanced minor-free graph for $t'\le \frac t2$ and $|s'|\le s+|Q_i|\le s+h(r)$. Further, for every $u,v\in  \mathcal{G}_i$, it holds that $d_{\mathcal{G}_i}(u,v)=d_{G}(u,v)$, and thus $\mathcal{G}_i$ has diameter at most $D$.
	Applying the inductive hypothesis to $\mathcal{G}_i$, we sample a dominating embedding $f_i$ into $H_i$, such that for every $v\in\mathcal{G}_i$ we have $f_i(v)=\bigcupdot_{j\ge1} g_{i,j}(v)$.
	We denote by $q_v^i$ the maximal index such that $g_{i,q_v^i}(v)\ne\emptyset$.
	Note that properties (\ref{prop:enhancedClmSparsity}) and (\ref{prop:enhancedClmDistortion}) hold and furthermore, $H_i$ has treewidth $\phi_{h(r)}\cdot\frac{\log n}{\eps\cdot\delta} + s'+h(r)\cdot \log 2t'\le \phi_{h(r)}\cdot\frac{\log n}{\eps\cdot\delta} + s+h(r)\cdot \log 2t$.
	In addition, for a vertex $v\in S\cup Q_i$, $|f_i(v)|=1$ (thus $q_v^i=1$), while for every vertex $v\in V$, $\mathbb{E}[q_v^i]\le (1+\delta)^{\log2t'}\le (1+\delta)^{\log t}$. 
	
	Let $\tilde{\mathcal{G}}$ be the graph induced on $\tilde{G}\cup S$. We apply \Cref{lem:Clan-AlmostEmbedable} to $\tilde{\mathcal{G}}$ to sample a dominating one-to-many embedding $\tilde{f}$ into $\tilde{H}$, such that for each vertex $v\in \tilde{\mathcal{G}}$, $\tilde{f}(v)$ is partitioned into $\tilde{\chi}(v)$ and $\tilde{\psi}(v)$.
	$\tilde{H}$ has treewidth $\phi_{h(r)}\cdot\frac{\log n}{\eps\cdot\delta}+s$ (this is by  \Cref{lem:Clan-AlmostEmbedable}, we first remove all apices and then add them back). Note also that properties (\ref{property:ClanAlmostEmbeddableDistortion}), (\ref{property:ClanAlmostEmbeddableFailProbability}), and (\ref{property:ClanAlmostEmbeddableCliquePreserve}) hold.
	
	We next describe how to combine the different parts into a single  embedding.
	The graph (and the induced embedding) will be created by identifying some vertices in $\tilde{H}$ with vertices in each $H_i$. Some of the graphs $H_i$ will be duplicated and we will have two copies of them (depending on whether $Q_i$ failed in $\tilde{f}$).
	Note that the set $S$ has a single copy everywhere, and thus for every $v\in S$, we will simply identify all the vertices $\tilde{f}(v),f_{1}(v),\dots,f_{p}(v)$.
	$$\mbox{For a vertex}~ v\in\tilde{\mathcal{G}}\mbox{, set } \quad g_1(v)=\tilde{\chi}(v) \quad \mbox{ and } \quad g_2(v)=\tilde{\psi}(v)~.$$
	Consider some $i\in[p]$. Note that the clique $Q_i$ belongs to $S_i$. In particular, for every vertex $v\in Q_i$, $f_i(v)$ is a singleton, and $f_i(Q_i)$ is a clique.
	We continue w.r.t. the $3$ cases in \Cref{lem:Clan-AlmostEmbedable} (see \Cref{fig:ClanLemmaCliqueCases} for an illustration of the cases):
	\begin{itemize}
		\item \textbf{$\tilde{f}$ succeeds on $Q_i$}: Here $\tilde{\psi}(Q_i)=\emptyset$, and $\tilde{\chi}(Q_i)=g_q(Q_i)$ contains a clique copy $Q^1_i\subseteq \tilde{\chi}(Q_i)$ of $Q_i$.
		We simply identify each vertex in $f_i(Q_i)$ with the corresponding copy in $Q^1_i$. We will abuse notation and refer to $H_i$ as $H_i^1$, to $f_i$ as $f_i^1$, and to $g_{i,j}$ as $g^1_{i,j}$.
		$$\mbox{For a vertex}~ v\in \mathcal{G}_i\setminus\tilde{\mathcal{G}},\quad\mbox{ for every } j\ge 1 \quad\mbox{ set } \quad g_j(v)=g^1_{i,j}(v)~.$$
		\item \textbf{$\tilde{f}$ fails on $Q_i$, and $\tilde{\chi}(Q_i)$ contains a clique copy of $Q_i$}:
		Denote by $Q_i^1\subseteq \tilde{\chi}(Q_i)$ the promised clique copy of $Q_i$.
		In addition, $\tilde{\psi}(Q^F_i)$ is guaranteed to contain a clique copy $Q_i^2$ of $Q^F_i=\{v\in Q_i\mid\tilde{\psi}(v)\ne\emptyset\}$.
		We duplicate $H_i$ into two graphs $H_i^1$ and $H_i^2$ with respective duplicate embeddings $f_i^1,f_i^2$. However, the vertices of $Q_i\setminus Q_i^F$ are removed from $H_i^2$ and $f_i^2$. We combine $\tilde{H}$ with $H_i^1$ (resp. $H_i^2$) by combining a clique copy from $\tilde{\chi}(Q_i)$ (resp. $\tilde{\psi}(Q^F_i)$) with the corresponding vertices from $f_i^1(Q_i)$ (resp. $f_i^2(Q^F_i)$) (recall that they are apices and thus have a single copy). 
		\begin{itemize}
			\item For every vertex $v\in \mathcal{G}_i\setminus\tilde{\mathcal{G}}$ where $q_v^i$ is the maximal index $j$ such that $g_{i,j}(v)\ne\emptyset$. 
			For every $j\in[1,q_v^i]$, set $g_j(v)=g^1_{i,j}(v)$ to be the corresponding copies from $f^1_i(v)$, and $g_{q_v^i+j}(v)=g^2_{i,j}(v)$ be the corresponding copies from $f^2_i(v)$.
		\end{itemize}
	
		\item \textbf{$\tilde{f}$ fails on $Q_i$, and $\tilde{f}(Q_i)$ contains two clique copies $Q_i^1,Q_i^2$ of $Q_i$ such that for every $v\in Q_i$, $Q_i^1\cup Q_i^2$ intersects both $\tilde{\chi}(v)$ and $\tilde{\psi}(v)$}:
		We duplicate $H_i$ into two graphs $H_i^1$ and $H_i^2$ with respective duplicate embeddings $f_i^1,f_i^2$. We combine $\tilde{H}$ with $H_i^1$ (resp. $H_i^2$) by identifying $Q_i^1$ (resp. $Q_i^2$) with $f_i^1(Q_i)$ (resp. $f_i^2(Q_i)$) (recall that they are apices and thus have a single copy). 
		\begin{itemize}
			\item For every vertex $v\in \mathcal{G}_i\setminus\tilde{\mathcal{G}}$ where $q_v^i$ is the maximal index $j$ such that $g_{i,j}(v)\ne\emptyset$. 
			For every $j\in[1,q_v^i]$, set $g_j(v)=g^1_{i,j}(v)$ be the corresponding copies from $ f^1_{i,j}(v)$, and $g_{q_v^i+j}(v)=g^2_{i,j}(v)$ be the corresponding copies from $f^2_{i,j}(v)$.
		\end{itemize}
	\end{itemize}

	We claim next that $f,g_1,g_2,\dots$ fulfill all the required properties.
	First, note that $f$ is clique-preserving as every clique must be contained in either $\tilde{\mathcal{G}}$ or some $\mathcal{G}_i$.
	Second, clearly $f$ is dominating as the weight of every edge between a vertex in $f(v)$ and $f(u)$ is $d_G(u,v)$.
	Third, as we only identify between cliques, the graph $H$ has treewidth
	\[
	\max\left\{ \phi_{h(r)}\cdot\frac{\log n}{\eps\cdot\delta}+s+h(r)\cdot\log2t\quad,\quad\phi_{h(r)}\cdot\frac{\log n}{\eps\cdot\delta}+s\right\} \quad=\quad\phi_{h(r)}\cdot\frac{\log n}{\eps\cdot\delta}+s+h(r)\cdot\log2t
	\]
	Forth, it holds by definition that for every vertex $v\in V$, $f(v)=\bigcupdot_jg_j(v)$.
	
	Next, we prove \propref{prop:enhancedClmSparsity}. Clearly, for a vertex $v\in S$, we identify between all its copies and thus $f(v)$ is a singleton. Consider a vertex $v\in V$, if $v\in \tilde{\mathcal{G}}$, then by \Cref{lem:Clan-AlmostEmbedable} $$\mathbb{E}[q_v]=1+\Pr\left[\tilde{f}\text{ fails on }v\right]\le1+\delta~.$$ 
	Else, consider $v\in \mathcal{G}_i\setminus \tilde{\mathcal{G}}$ for some $i$, and denote by $q_{v}^{i}$ the maximal index $j$ such that $g_{i,j}$ is non-empty. We have 
	\begin{align*}
	\mathbb{E}[q_{v}] & =\mathbb{E}[q_{v}^{i}]\cdot\Pr\left[\tilde{f}\text{ succeeds on }Q_{i}\right]+\mathbb{E}[2q_{v}^{i}]\cdot\Pr\left[\tilde{f}\text{ fails on }Q_{i}\right]\\
	& =\mathbb{E}[q_{v}^{i}]\cdot\left(1+\Pr\left[\tilde{f}\text{ fails on }Q_{i}\right]\right)\\
	& \le(1+\delta)^{\log 2\cdot\frac t2}\cdot(1+\delta)=(1+\delta)^{\log2t}~,
	\end{align*}
	where the first equality is because we have two copies of $H_i$ iff $\tilde{f}$ fails on $Q_i$. The second equality is because $\Pr\left[f\text{ succeeds on }Q_{i}\right]=1-\Pr\left[f\text{ fails on }Q_{i}\right]$. The final inequality follows by the induction hypothesis and \Cref{lem:Clan-AlmostEmbedable}.
	
	Finally, we prove \propref{prop:enhancedClmDistortion}. Consider a pair of vertices $u,v\in V$. We proceed by case analysis.
	\begin{itemize}
		\item \textbf{If a shortest path from $u$ to $v$ goes through a vertex $z\in S$} (this includes the case where either $u$ or $v$ are in $S$): Then 
		\begin{align*}
		& \min_{j}\max_{u'\in g_{j}(u),v'\in g_{1}(v)}d_{H}(u',v')\quad\le\quad\max_{u'\in f(u),v'\in f(v)}d_{H}(u',v')\\
		& \qquad\quad\le\quad\max_{u'\in f(u),v'\in f(v)}d_{H}(u',f(z))+d_{H}(f(z),v')\quad=\quad d_{G}(u,z)+d_{G}(z,v)\quad=\quad d_{G}(u,v)~.
		\end{align*}
	\end{itemize}	
		For the remaining cases, we assume that $d_{G'}(u,v)=d_G(u,v)$ (recall that $G'=G[V\setminus S]$).
	\begin{itemize}	
		\item \textbf{Else, if both $u,v\in \tilde{\mathcal{G}}$}: Then by \Cref{lem:Clan-AlmostEmbedable}, 
		\begin{align*}
		\min_{j}\max_{u'\in g_{j}(u),v'\in g_{1}(v)}d_{H}(u',v') & \le\min\left\{ \max_{u'\in\chi(u),v'\in\chi(v)}d_{H}(u',v'),\max_{u'\in\psi(u),v'\in\chi(v)}d_{H}(u',v')\right\} \\
		& \le d_{G}(u,v)+\eps D~.
		\end{align*}
		
		\item \textbf{Else, if $u\in \tilde{\mathcal{G}}$ and there is an $i\in[p]$ such that $v\in \mathcal{G}_i\setminus\tilde{\mathcal{G}}$}:
		There is necessarily a vertex $x\in Q_i$ such that there is a shortest path from $u$ to $v$ in $G$ going through $x$. Note by the construction that (a) the copy $g_1(v)$ belongs to $H^1_i$ (a copy of $H_i$), (b) there is an edge from $g_1(v)$ to a copy of $x$ in $f_i^1(Q_i)$ and (c)  a clique copy $Q_i^1\subseteq \tilde{f}(Q)$  of $Q$  is identified with $f_i^1(Q_i)$ (a set of singletons). We continue by case analysis:
		\begin{itemize}
			\item If either $\tilde{f}$ succeeds on $Q_i$, or $Q_i^1\subseteq \tilde{\chi}(Q_i)$.
			Then there is a copy $\hat{x}$ of $x$ in $g_1(x)\cap Q_i^1$. It holds that 
			\begin{align}
			\min_{j}\max_{u'\in g_{j}(u),v'\in g_{1}(v)}d_{H}(u',v') & \le\min_{j}\left(\max_{u'\in g_{j}(u)}d_{H}(u',\hat{x})+\max_{v'\in g_{1}(v)}d_{H}(\hat{x},v')\right)\nonumber\\
			& \le d_{G}(u,x)+\eps D+d_{G}(x,v)=d_{G}(u,v)+\eps D~.\label{eq:DominatingStretchCase3}
			\end{align}
			where the second inequality follows by the second case (as $x\in\tilde{\mathcal{G}}$), and the fact that there is an edge in $H$ between $\hat{x}$ to every vertex in $g_1(v)$.
			
			\item Else, $\tilde{f}(Q_i)$ contains two clique copies $Q_i^1,Q_i^2$ of $Q_i$. Note that $\hat{x}$ can belong to either $g_1(x)=\tilde{\chi}(x)$ or $g_2(x)=\tilde{\psi}(x)$. Nevertheless, by using either equation (\ref{eq:LemmaDominatingStretchGurantee}) or (\ref{eq:LemmaDominatingStretchGuranteePSI}) we have that $\min_{j}\max_{u'\in g_{j}(u)}d_{H}(u',\hat{x})\le d_{G}(u,x)+\eps D$.
			As there is edge in $H$ between $\hat{x}$ to every vertex in $g_1(v)$, we conclude that equation (\ref{eq:DominatingStretchCase3}) holds.
		\end{itemize}

		\item  \textbf{Else, if $v\in \tilde{\mathcal{G}}$ and there is an $i\in[p]$ such that $u\in \mathcal{G}_i\setminus{\mathcal{G}}$}:
		There is necessarily a vertex $x\in Q_i$ such that there is a shortest path from $u$ to $v$ in $G$ going through $x$.
		By the second case, there is an index $j'$ such that $\max_{x'\in g_{j'}(x),v'\in g_{1}(v)}d_{H}(x',u')\le d_{G}(x,v)+\eps D$. 
		As $x\in\tilde{\mathcal{G}}$, $j'\in\{1,2\}$. In any case, a copy of $H_i$ was assigned to $\tilde{H}$ by identifying clique vertices. In particular, some vertex $\hat{x}\in g_{j'}(x)$ was identified with the apex vertex $f_i(x)$ (from the relevant copy). Therefore there is an index $j''$ such that $\hat{x}$ has edges towards all the vertices in $g_{j''}(u)$. We conclude,
		\begin{align*}
		\min_{j}\max_{u'\in g_{j}(u),v'\in g_{1}(v)}d_{H}(u',v') & \le\max_{u'\in g_{j''}(u),v'\in g_{1}(v)}d_{H}(u',\hat{x})+d_{H}(\hat{x},v')\\
		& \le d_{G}(u,x)+\max_{x'\in g_{j'}(x),v'\in g_{1}(v)}d_{H}(x',v')\\
		& \le d_{G}(u,x)+d_{G}(x,v)+\eps D=d_{G}(u,v)+\eps D~.
		\end{align*}
		
		\item \textbf{Else, if there is an $i\in[p]$ such that $u,v\in \mathcal{G}_i\setminus{\mathcal{G}}$}:
		There is a copy of $H_i$ which embedded as is into $H$ and contains all the vertices in $g_1(v)$. By the induction hypothesis 
		\[
		\min_{j}\max_{u'\in g_{j}(u),v'\in g_{1}(v)}d_{H}(u',v')\le\min_{j}\max_{u'\in g_{i,j}(u),v'\in g_{i,1}(v)}d_{H_{i}}(u',v')\le d_{\mathcal{G}_{i}}(u,v)+\eps D= d_{G}(u,v)+\eps D~.
		\]
			
		\item  \textbf{Else, there are $i\ne i'\in[p]$ such that $u\in \mathcal{G}_i\setminus{\mathcal{G}}$ and $v\in \mathcal{G}_{i'}\setminus{\mathcal{G}}$}:
		There are necessarily vertices $y\in Q_i$ and $x\in Q_{i'}$ such that there is a shortest path from $u$ to $v$ in $G$ going through $y$ and $x$. 
		Note that the copy $H^1_{i'}$ of $H_{i'}$ containing $g_1(v)$ was added to $H$ by identifying $f^1_{i'}(Q_{i'})$ with a clique copy $Q_{i'}^1$ of $Q_{i'}$. In particular, there is a copy $\hat{x}\in Q_{i'}^1$ of $x$ which has edges towards all the vertices in $g_1(v)$. There are two cases:
		\begin{itemize}
			\item If $\hat{x}\in g_1(x)$, then by the third case there is an index $j$ such that $\max_{u'\in g_{j}(u)}d_{H}(u',\hat{x})\le\max_{u'\in g_{j}(u),x'\in g_{1}(x)}d_{H}(u',x')\le d_{G}(u,x)+\eps D$. As there is an edge from $\hat{x}$ to every copy of $v$ in $g_1(v)$, we conclude that $\max_{u'\in g_{j}(u),v'\in g_{1}(v)}d_{H}(u',v')\le\max_{u'\in g_{j}(u)}d_{H}(u',\hat{x})+\max_{v'\in g_{1}(x)}d_{H}(\hat{x},v')\le d_{G}(u,x)+\eps D+d_{G}(x,v)=d_{G}(u,v)+\eps D$.
			\item Else, $\hat{x}\in g_2(x)$. Necessarily $\tilde{f}$ failed on $Q_{i'}$ and $\tilde{f}(Q_{i'})$ contains two clique copies $Q_{i'}^1,Q_{i'}^2$ of $Q_{i'}$. 
			It holds that $g_2(x)=\tilde{\psi}(x)$, thus by \Cref{lem:Clan-AlmostEmbedable} (case 3.(c))
			there is an index $j\in\{1,2\}$ such that $\max_{y'\in g_{j}(y)}d_{H}(y',\hat{x})\le\max_{y'\in g_{j}(y),x'\in \tilde{\psi}(x)}d_{H}(y',x')\le d_{G}(x,y)+\eps D$.
			Let $\hat{y}\in Q_i^j\subseteq g_j(Q_i)$ be the copy of $y$ from the corresponding clique copy.
			Note that there is an edge from $\hat{x}$ to every copy of $v$ in $g_1(v)$. Farther, there is an index $j''$ such that $\hat{y}$ has edges towards all the vertices in $g_{j''}(u)$. We conclude,
			\begin{align*}
			\min_{j}\max_{u'\in g_{j}(u),v'\in g_{1}(v)}d_{H}(u',v') & \le\max_{u'\in g_{j''}(u)}d_{H}(u',\hat{y})+d_{H}(\hat{y},\hat{x})+\max_{v'\in g_{1}(v)}d_{H}(\hat{x},v')\\
			& \le d_{G}(u,y)+\max_{y'\in g_{j}(y),x'\in g_{2}(x)}d_{H}(y',x')+d_{G}(x,v)\\
			& \le d_{G}(u,y)+d_{G}(y,x)+\eps D+d_{G}(x,v)=d_{G}(u,v)+\eps D~.
			\end{align*}
		\end{itemize}
	\end{itemize}
\end{proof}

\begin{remark}\label{remark:ClanImpliesWeakRamsey}
	The clan embedding in \Cref{thm:Clan-Embedding} directly implies a weaker version of \Cref{thm:Ramsey-minor-Free}, where the only difference is that the distortion is only for pairs where both $u,v\in M$ and not only $u\in M$. Note that this weaker version is still strong enough for our application to the $\rho$-independent set problem in \Cref{thm:isolated}.\\
	Sketch: sample a clan embedding $(f,\chi)$ using \Cref{thm:Clan-Embedding}.  Return $g=\chi$ with the set $M=\{v\in V~\mid~|f(v)|=1\}$. The weaker distortion guarantee and failure probability are straightforward.
\end{remark}

%% file: applications.tex
\section{Applications}\label{sec:applications}
Organization: in Sections \ref{subsec:indepdendent}, \ref{subsec:dominating} and \ref{subsec:routing} we provide the algorithms (and proofs) to our QPTAS $^{\ref{foot:approximationSchemes}}$ for metric $\rho$-independent set problem, QPTAS for metric $\rho$-dominating set problem, and compact routing scheme, respectively. 

We begin with a discussion on approximation schemes for metric $\rho$-dominating/independent set problems in bounded treewidth graphs.
In the $(k,r)$-center problem we are given a graph $G=(V,E,w)$, and the goal is to find a set $S$ of centers of cardinality at most $r$ such that every vertex $v\in V$ is at distance at most $r$ from some center $u\in S$.
Katsikarelis, Lampis and Paschos \cite{KLP19} provided a PTAS $^{\ref{foot:approximationSchemes}}$ for the $(k,r)$-center problem in treewidth $\tw$ graphs using a dynamic programming approach. Specifically, for any parameters $k,r\in\N$ and $\eps\in(0,1)$, they provided an algorithm running in $O(\frac{\tw}{\eps})^{\tw}\cdot\poly(n)$ time that either returns a solution to the $(k,(1+\eps)r)$-center problem, or (correctly) declares that there is no valid solution to the $(k,r)$-center problem in $G$.
This dynamic programming can be easily generalized to the case where there is a measure $\mu:V\rightarrow \R^+$, and terminal set ${\cal K}\subset V$. Specifically, the algorithm will either return a set $S$ of measure $\mu(S)\le k$, such that every vertex $v\in {\cal K}$ is at distance at most $(1+\eps)r$ from $S$, or will declare there is no set $S$ of measure at most $k$ at distance at most $r$ from every vertex in ${\cal K}$.

As was observed by Fox-Epstein \etal \cite{FKS19}, using \cite{KLP19} one can construct a bicriteria PTAS for the metric $\rho$-dominating set problem in treewidth $\tw$ graphs with $O(\frac{\tw}{\eps})^{\tw}\cdot\poly(n)$ running time. \cite{FKS19} studied the basic version (with uniform measure and ${\cal K}=V$), however this observation holds for the general case as well.
In a follow-up paper, Katsikarelis \etal \cite{KLP20} constructed a similar dynamic programming for the $\rho$-independent problem with the same  $O(\frac{\tw}{\eps})^{\tw}\cdot\poly(n)$ running time. It could also be generalized to work with a measure $\mu$. This dynamic programming was also promised to appear in the full version of \cite{FKS19}.
We conclude this discussion:
\begin{theorem}[\cite{KLP19,KLP20}]\label{thm:treewidthPTAS}
	There is a bicriteria polynomial approximation scheme (PTAS) for both metric $\rho$-independent set and $\rho$-dominating set problems in treewidth $\tw$ graphs with running time $O(\frac{\tw}{\eps})^{\tw}\cdot\poly(n)$.
\end{theorem}

\subsection[Metric $\rho$-Independent Set (\Cref{thm:isolated})]{QPTAS for the $\rho$-Independent Set Problem in Minor-Free Graphs}\label{subsec:indepdendent}
This subsection is devoted to proving the following theorem:
\isolated*

\begin{proof}
	Create a new graph $G'$ from $G$ by adding a single vertex $\psi$ at distance $\frac34\rho$ from all the other vertices. 
	$G'$ is $K_{r+1}$-minor free. Note that for every $u,v\in Y$, it holds that $d_{G'}(u,v)=\min\{\frac32\rho,d_{G}(u,v)\}$. Thus $G'$ has diameter at most $\frac{3}{2}\rho$. Furthermore, for every $\rho'\in(0,\frac32\rho)$, a set $S\subseteq V$ is a $\rho'$-independent set in $G$ if and only if $S$ is a $\rho'$-independent set in $G'$.
	Using \Cref{thm:Ramsey-minor-Free} with parameters $\eps'=\frac\eps2$ and $\delta=\frac\eps4$, let $g$ be an embedding of $G'$ into a treewidth-$O_r(\frac{\log^2n}{\eps^2})$ graph $H$ with a set $M\subseteq V\cup\{\psi\}$ such that (1) for every $u,v\in M$, $d_H(g(u),g(v))\le d_{G'}(u,v)+\frac{\eps}{2}\cdot\frac32\rho<d_{G'}(u,v)+\eps\rho$, and (2) for every $v\in V$, $\Pr[v\in M]\ge 1-\frac\eps4$. 
	
	Define a new measure $\mu_H$ in $H$, where for each $v\in G'$,
	\[
	\mu_{H}(v')=\begin{cases}
	0 & v'\notin g(V\cap M)\\
	\mu(v) & \text{else, }g(v)=v'\text{ for some }v\in M\setminus\{\psi\}
	\end{cases}\qquad.
	\]
	In particular, $\mu_H(g(\psi))=0$.
	Using \Cref{thm:treewidthPTAS}, we find a $(1-\frac\eps2)\rho$-independent set $S_H$ w.r.t. $\mu_H$, such that for every $\rho$-independent set $\tilde{S}$ in $H$ it holds that $\mu_H(S_H)\ge(1-\frac\eps2)\mu_H(\tilde{S})$.
	We can assume that $S_H\subseteq g(M)$, as the measure of all vertices out of $g(M)$ is $0$. We will return $S=g^{-1}(S_H)$; note that $S\subseteq M$. 
	First, we argue that $S$ is a $(1-\eps)\rho$-independent set. 
	For every $u,v\in S$, $g(u),g(u)\in S_H$ thus
	\[
	(1-\frac{\eps}{2})\rho\le d_{H}(g(u),g(v))\le d_{G'}(u,v)+\frac{\eps}{2}\rho\le d_{G}(u,v)+\frac{\eps}{2}\rho~,
	\]
	implying $d_{G}(u,v)\ge(1-\eps)\rho$.
	
	\sloppy Let $S_{\opt}$ be a $\rho$-independent set  w.r.t. $d_G$ of maximal measure.
	As $g$ is dominating embedding, $g(S_{\opt}\cap M)$ is a $\rho$-independent set in $H$. By linearity of expectation 
	${\mathbb{E}[\mu(S_{\opt}\setminus M)]=\sum_{v\in S_{\opt}}\mu(v)\cdot\Pr\left[v\notin M\right]\le\frac{\eps}{4}\cdot\mu(S_{\opt})}$. Using Markov inequality 
	\[
	\Pr\left[\mu(S_{\opt}\cap M)<(1-\frac{\eps}{2})\mu(S_{\opt})\right]=\Pr\left[\mu(S_{\opt}\setminus M)\ge\frac{\eps}{2}\mu(S_{\opt})\right]\le\frac{\mathbb{E}[\mu(S_{\opt}\setminus M)]}{\frac{\eps}{2}\mu(S_{\opt})}\le\frac{1}{2}~.
	\]
	Thus, with probability at least $\frac12$, $H$ contains a $\rho$-independent set $g(S_{\opt}\cap M)$ of measure ${\mu_{H}(g(S_{\opt}\cap M))=\mu(S_{\opt}\cap M)\ge(1-\frac{\eps}{2})\mu(S_{\opt})}$. If this event indeed occurs, the independent set $S_H$ returned by \cite{FKS19} algorithm will be of measure greater than 	$(1-\frac{\eps}{2})(1-\frac{\eps}{2})\mu(S_{\opt})>(1-\eps)\mu(S_{\opt})$.
	High probability could be obtained by repeating the above algorithm $O(\log n)$ times and returning the independent set of maximal  cardinality among the observed solutions.	 
\end{proof}
\begin{remark}\label{rem:IsolatedDeterministic}
	The algorithm above can be derandomized as follows: first note that the algorithm from \Cref{thm:treewidthPTAS} is deterministic. Next, during the construction in the proof of \Cref{thm:Ramsey-minor-Free}, each time we execute \Cref{lem:Ramsey-almost-embeddable} we pick $\sigma\in O(\frac{1}{\delta})$ uniformly at random, where $\delta=\Theta(\frac{\eps}{\log n})$.
	As we bound the probability of $\Pr[v\notin M]$ using a simple union bound, it will still hold if we pick the same $\sigma$ in all the executions of \Cref{lem:Ramsey-almost-embeddable}. We conclude that we can sample the embedding of \Cref{thm:Ramsey-minor-Free} from a distribution with support size $O(\frac{\log n}{\eps})$. A derandomization follows.
\end{remark}

\subsection[Metric $\rho$-Dominating Set (\Cref{thm:dominatingSet})]{QPTAS for the $\rho$-Dominating Set Problem in Minor-Free Graphs}\label{subsec:dominating}
We restate the main theorem in this section for convenience. 
\dominating*
\begin{proof}
	Similarly to \Cref{thm:isolated}, we start by constructing an auxiliary graph $G'$ from $G$ by adding a single vertex $\psi$ at distance $2\rho$ from all the other vertices. 
	Extend the measure $\mu$ to $\psi$ by setting $\mu(\psi)=\infty$.
	For every $u,v\in V$ it holds that $d_{G'}(u,v)=\min\{4\rho,d_{G}(u,v)\}$. It follows that $G'$ is a $K_{r+1}$-minor-free graph with diameter bounded by $4\rho$.
	In particular, for every $\rho'\in(0,2\rho)$, a set $S\subseteq V$ is $\rho'$-dominating set (w.r.t. ${\cal K}$) in $G$ if and only if $S$ is $\rho'$ dominating set in $G'$ (w.r.t. ${\cal K}$).
	Using \Cref{thm:Clan-Embedding} with parameters $\eps'=\frac{\eps}{12}$ and $\delta=\frac\eps 6$, let $(f,\chi)$ be a clan embeddings of $G'$ into a treewidth $O_r(\frac{\log^2n}{\eps^2})$ graph $H'$ with additive distortion $\eps'\cdot4\rho=\frac\eps3\rho$. 
	Define a new measure $\mu_H$ in $H$, where for each $v'\in H$,
	\[
	\mu_{H}(v')=\begin{cases}
	\infty & v'\notin f(V)\\
	\mu(v) & v'\in f(v)
	\end{cases}
	\]
	Set also ${\cal K}_H=\chi({\cal K})\subseteq H$ to be our set of terminals.
	Using \Cref{thm:treewidthPTAS}, 
	we find a $(1+\frac{\eps}{3})^2\rho$-dominating set $A_H$, such that for every $\chi(v)\in {\cal K}_H$, $d_H(\chi(v),A_H)\le(1+\frac\eps3)^2\rho$, and for every $(1+\frac\eps3)\rho$-dominating set $\tilde{A}$ w.r.t. ${\cal K}_H$ it holds that $\mu_H(A_H)\le(1+\frac\eps3)\mu_H(\tilde{A})$.
	We can assume that $A_H$ contains only vertices from $f(V)$ (as all other vertices have measure $\infty$, while ${\cal K}_H$ itself is legal solution of finite measure). We will return $A=f^{-1}(A_H)=\{u\in V\mid f(u)\cap A_H\ne\emptyset\}$. 
	
	First, we argue that $A$ is a $(1+\eps)\rho$-dominating set. For every vertex $v\in {\cal K}$, $\chi(v)\in {\cal K}_H$. Therefore there is a vertex $\hat{u}\in A_H$ such that $d_H(\chi(v),\hat{u})\le (1+\frac\eps3)^2\rho$. In particular, our solution $A$ contains the vertex $u$ such that $\hat{u}\in f(u)$. As $(f,\chi)$ is 
	dominating embedding 
	we conclude 
	\[
	d_{G}(u,v)\le\min_{u'\in f(u)}d_{H}(u',\chi(v))\le d_{H}(\hat{u},\chi(v))\le(1+\frac{\eps}{3})^{2}\rho<(1+\eps)\rho~.
	\]
	
	Second, we argue that $A$ has nearly optimal measure. Let $A_{\opt}$ be a $\rho$-dominating set in $G$ w.r.t. ${\cal K}$ of minimal measure.
	As $(f,\chi)$ has additive distortion $\frac\eps3\rho$, 
	$f(A_{\opt})$ is a $(1+\frac\eps 3)\rho$-dominating set in $H$ (w.r.t. ${\cal K}_H$). Indeed, consider a vertex $\chi(v)\in {\cal K}_H$ (for $v\in {\cal K}$). There is a vertex $u\in A_{\opt}$ such that $d_G(u,v)\le \rho$. It holds that 
	\[
	d_{H}(f(A_{\opt}),\chi(v))\le\min_{u'\in f(u)}d_{H}(u',\chi(v))\le d_{G}(u,v)+\frac{\eps}{3}\rho\le (1+\frac{\eps}{3})\rho
	\]
	By \Cref{thm:treewidthPTAS}, we will find a $(1+\frac\eps3)^2\rho$-dominating set of measure at most $(1+\frac\eps3)\mu_H(f(A_\opt))$ in $H$.
	By linearity of expectation,
	\[
	\mathbb{E}\left[\mu_{H}(f(A_{\opt})\right]=\sum_{u\in A_{\opt}}\mu(u)\cdot\mathbb{E}\left[\left|f(u)\right|\right]\le(1+\frac{\eps}{6})\cdot\mu(A_{\opt})~.
	\]
	From the other hand, $\mu_{H}(f(A_{\opt}))\ge\mu_{H}(\chi(A_{\opt}))=\mu(A_{\opt})$. 
	Using Markov inequality,
	\begin{align*}
	\Pr\left[\mu_{H}(f(A_{\opt}))\ge(1+\frac{\eps}{3})\cdot\mu(A_{\opt})\right] & =\Pr\left[\mu_{H}(f(A_{\opt}))-\mu(A_{\opt})\ge\frac{\eps}{3}\mu(A_{\opt})\right]\\
	& \le\frac{\mathbb{E}[\mu_{H}(f(A_{\opt}))-\mu(A_{\opt})]}{\frac{\eps}{3}\mu(A_{\opt})}\le\frac{\frac{\eps}{6}}{\frac{\eps}{3}}=\frac{1}{2}~.
	\end{align*}
	Thus with probability at least $\frac12$, $H$ contains  $(1+\frac\eps3)\rho$-dominating set of measure $(1+\frac\eps3)\mu(A_{\opt})$. If this event indeed occurs, the dominating set $A_H$ returned by \Cref{thm:treewidthPTAS} will be of measure at most $(1+\frac{\eps}{3})^2\mu(A_{\opt})<(1+\eps)\mu(A_{\opt})$.
	High probability could be obtained by repeating the algorithm above $O(\log n)$ times and returning the set of minimum measure among the observed dominating sets.	 
\end{proof}

\subsection[Compact Routing Scheme (\Cref{thm:route})]{Compact Routing Scheme}\label{subsec:routing}
We restate the main theorem of this subsection for convenience.
We begin by presenting a result of Thorup and Zwick \cite{TZ01b} regarding routing in a tree.
\begin{theorem}[\cite{TZ01b}]\label{thm:tree-routh}
	For any tree $T=(V,E)$ (where $|V|=n$), there is a routing scheme with stretch $1$ that has routing tables of size $O(1)$ and labels (and headers) of size $O(\log n)$.
\end{theorem}
Recall that we measure space in machine words, where each word is $\Theta(\log n)$ bits.
We stress out the extremely short routing table size obtained in \cite{TZ01b}. Note that when a vertex receives a packet with a header, it makes the routing decision based only on the routing table, and do not require any knowledge of the label of itself. In particular, the routing table contains a unique identifier of the vertex.

Additional ingredient that our construction will require is that of a \emph{distance labeling scheme} for trees. 
A \emph{distance labeling},
assigns to each point $x\in X$ a label $l(x)$, 
and there is an algorithm $\mathcal{A}$ (oblivious to $(X,d)$) that provided labels $l(x),l(y)$ of arbitrary $x,y \in X$, can compute $d_G(u,v)$.
Specifically, a distance labeling is said to have \emph{stretch} $t\ge 1$ if
\[
\forall x,y\in X,
\qquad
d(x,y)\le \mathcal{A}\left(l(x),l(y)\right) \le t\cdot d(x,y).
\]
We refer to \cite{FGK20} for an overview of distance labeling schemes in different regimes (and comparison with metric embedding, see also \cite{Peleg00Labling,GPPR04,TZ05,EFN18}). 
Exact distance labeling on an $n$-vertex tree requires  $\Theta(\log n)$ words \cite{AGHP16}, which is already larger than the routing table size we are aiming for. 
Nonetheless, Freedman  \etal \cite{FGNW17} (improving upon \cite{AGHP16,GKKPP01}) showed that for any $n$-vertex unweighted tree, and $\eps\in(0,1)$, one can construct an $(1+\eps)$-labeling scheme with labels of size $O(\log\frac1\eps)$ words.
\begin{theorem}[\cite{FGNW17}]\label{thm:tree-label}
	For any $n$-vertex tree $T=(V,E)$ with polynomial aspect ratio $^{\ref{foot:aspectRatio}}$, and parameter $\eps\in(0,1)$, there is a distance labeling scheme with stretch $1+\eps$, and $O(\log\frac1\eps)$ label size.
\end{theorem}
We will use \Cref{thm:tree-label} for fixed $\eps$. For this case the theorem can be extended to weighted trees with polynomial aspect ratio (by subdividing edges).
\begin{proof}[Proof of \Cref{thm:route}]
	We combine \Cref{thm:ClanSpanningTree} with \Cref{thm:tree-routh} and  \Cref{thm:tree-label} to construct a compact routing scheme.
	We begin by sampling a spanning clan embedding $(f,\chi)$ of $G$ into a tree $T$ with distortion $O(k\log\log n)$ such that for every vertex $v\in V$,  $\mathbb{E}[|f(v)|]\le n^{1/k}$.
	Using \Cref{thm:tree-label}, we construct a distance labeling scheme for $T$ with stretch at most $2$ and $O(1)$ label size. That is, each vertex $v'\in T$ has a label $l_{\rm dl}(v')$ of constant size, such that for every pair $v',u'\in T$, $d_T(v',u')\le \mathcal{A}\left(l_{\rm dl}(v'),l_{\rm dl}(u')\right) \le 2\cdot d_T(v',u')$ (${\rm dl}$ stands for distance labeling).
	
	Using \Cref{thm:tree-routh}, we construct a compact routing scheme for $T$, such that each $v'\in T$ has a label $\ell_{\rm crs}(v')$ of size $O(\log |T|)=O(\log n)$, and routing table $\tau_{\rm crs}(v')$ of size $O(1)$ (${\rm crs}$ stands for compact routing scheme).
	We construct a compact routing scheme for $G$ as follows: for every vertex $v\in V$, its label defined to be $\ell_{G}(v)=\left(\ell_{\rm crs}(\chi(v)),l_{\rm dl}(\chi(v))\right)$, and its table $\tau_{G}(v)$ to be the concatenation of $\left\{\left(\tau_{\rm crs}(v'),l_{\rm dl}(v')\right)\right\}_{v'\in f(v)}$.
	In words, the label $\ell_{G}(v)$ consist of the routing label $\ell_{\rm crs}(\chi(v))$, and distance label $l_{\rm dl}(\chi(v))$, of the chief $\chi(v)$ in $T$, while the routing table $\tau_{G}(v)$ consist of the routing table $\tau_{\rm crs}(v')$, and distance label $l_{\rm dl}(v')$, of all the copies $v'$ in the clan $f(v)$.
	Clearly, the size of the label is $O(\log n)+O(1)=O(\log n)$, while the expected size of the routing table is $\mathbb{E}[\sum_{v'\in f(v)}O(1)]=O(1)\cdot \mathbb{E}[|f(v)|]=O(n^{\frac1k})$.
	
	Consider a node $v$ that wants so send a package to a node $u$, while possessing the routing label $\ell_{G}(u)$ of $u$. $v$ will go over all the copies $v'\in f(v)$, and choose the copy $v_u$ that minimized the estimated distance $\mathcal{A}\left(l_{\rm dl}(v'),l_{\rm dl}(\chi(u))\right)$. Then, using the routing table $\tau_{\rm crs}(v_u)$ of $v_u$, $v$ will make a routing decision and transfer the package to the first vertex $z'\in T$ on the shortest path from $v_u$ to $\chi(u)$ in $T$. $v$ will transfer this package with a header consisting of the label of $u$ and the name of $z'$.
	This somewhat longer routing decision process occurs only when a delivery is initiated.
	In any other step, a node $z$ receives a package with a header containing the routing label of the destination $\ell_{G}(u)$ and a name of a copy $z'\in f(z)$. Then $z$ uses the routing table $\tau_{\rm crs}(z')$ of $z'$ to make a routing decision and transfer the package to the first vertex $q'\in T$ on the shortest path from $z'$ to $\chi(u)$ in $T$. As previously, $z$ will  transfer the package with a header consisting of the label of $u$ and the name of $q'$.
	Clearly the size of the header is $O(\log n)$. Note that other than the first decision, each decision is made in constant time (while the first decision is made in expected $O(n^{\frac1k})$ time). 
	Finally, when routing a package starting at $v$ towards $u$, the path corresponds exactly to a path in $T$ from a copy $v_u\in f(v)$ to $\chi(u)$. The length of this path is bounded by
	\begin{align*}
	d_{T}(v_{u},\chi(u)) & \le\mathcal{A}\left(l_{{\rm dl}}(v_{u}),l_{{\rm dl}}(\chi(u))\right)=\min_{v'\in f(v)}\mathcal{A}\left(l_{{\rm dl}}(v'),l_{{\rm dl}}(\chi(u))\right)\\
	& \le\min_{v'\in f(v)}2\cdot d_{T}(v',\chi(u))=O(k\log\log n)\cdot d_{G}(v,u)~.
	\end{align*}
\end{proof}

%% file: AppendixPathDistortion.tex
\section{Path Distortion of Clan embeddings into ultrametrics}\label{app:PathDistortion}
In this section we provide briefly the modification and missing details required to obtain the path distortion property for our clan embedding into ultrametrics.
\begin{definition}[Path-distortion]	
	We say that the one-to-many embedding $f:X\rightarrow 2^Y$ between $(X,d_X)$ to $(Y,d_Y)$ has \emph{path-distortion} $t$ if for every sequence $\left(x_0,x_1,\dots,x_m\right)$ in $X$ there is a sequence $v'_0,\dots,v'_m$ in $Y$ where $x'_i\in f(x_i)$, such that $\sum_{i=0}^{m-1}d_Y(x'_i,x'_{i+1})\le t \cdot \sum_{i=1}^{m-1}d_X(x_i,x_{i+1})$.
\end{definition}

To obtain a clan embedding $(f,\chi)$ as in \Cref{lem:clanTreeMeasure}, the only modification required is to use the following strengthen version of \Cref{clm:ClanTreePartition} (the proof of which appears bellow).
\begin{claim}\label{clm:ClanTreePartitionAlt}
	\sloppy There is a point $v\in X$ and radius $R\in(0,\frac{\diam(X)}{2}]$,
	such that the sets ${P=B_{X}(v,R+\frac{1}{8(k+1)}\cdot\diam(X))}$,
	$Q=B_{X}(v,R)$, and $\bar{Q}=X\setminus Q$ satisfy the following properties:
	\begin{enumerate}
		\item $\min\{\mu(P),\mu(\bar{Q})\}\le\frac{2}{3}\cdot\mu(X)$, and $\diam(P)\le\frac12\cdot\diam(X)$.
		\item $\mu(P)\le\mu(Q)\cdot\left(\frac{\mu^{*}(X)}{\mu^{*}(P)}\right)^{\frac{1}{k}}$.		
	\end{enumerate} 
\end{claim}
As a result, the distortion gurantee  we will obtain will be $16(k+1)$ instead of $16 k$.
However, we will be guaranteed that recursively one of the two created clusters has measure at most $\frac{2}{3}\mu (X)$, and also that the diameter in the first cluster is bounded by half the diameter of $X$.
These are the only properties used in the proof of \cite{BM04multi} to obtain the path distortion gurantee. In particular, the exact same argument as in \cite{BM04multi} will imply the following result:
\begin{lemma}\label{lem:clanTreeMeasureAlt}
	Given an $n$-point metric space $(X,d_{X})$ with aspect ratio$^{\ref{foot:aspectRatio}}$ $\Phi$, $(\ge1)$-measure $\mu:X\rightarrow\mathbb{R}_{\ge1}$,
	and integer parameter $k\ge1$, there is a clan embedding $(f,\chi)$ into an ultrametric with multiplicative distortion $16(k+1)$, path distortion $O(k)\cdot\min\{\log n,\log \Phi\}$, and
	such that $\mathbb{E}_{x\sim\mu}[|f(x)|]\le\mu(X)^{1+\frac1k}$.
\end{lemma}
Using the exact same arguments we had to obtain \Cref{thm:ClanUltrametric} from \Cref{lem:clanTreeMeasure}, we  conclude:
\begin{restatable}[Clan embedding into ultrametric]{theorem}{ClanUltrametricAlt}
	\label{thm:ClanUltrametricAlt}
	Given an $n$-point metric space $(X,d_{X})$ with aspect ration $\Phi$, and parameter $\epsilon\in(0,1]$, there is a uniform distribution $\mathcal{D}$ over $O(n\log n/\epsilon^2)$ clan embeddings $(f,\chi)$ into ulrametrics	with multiplicative distortion $O(\frac{\log n}{\epsilon})$, path distortion $O(\frac{\log n}{\epsilon})\cdot\min\{\log n,\log \Phi\}$, and such that for every point $x\in X$, $\mathbb{E}_{f\sim\mathcal{D}}[|f(x)|]\le1+\epsilon$.
	
	In addition, for every $k\in \N$, there is a uniform distribution $\mathcal{D}$ over $O(n^{1+\frac{2}{k}}\log n)$ clan embeddings $(f,\chi)$
	into ulrametrics with multiplicative distortion $16(k+1)$, path distortion $O(k)\cdot\min\{\log n,\log \Phi\}$, and such that for every point $x\in X$, $\mathbb{E}_{f\sim\mathcal{D}}[|f(x)|]= O(n^{\frac1k})$.
\end{restatable}

\begin{remark}
	The spanning clan embedding construction for \Cref{thm:ClanSpanningTree} actually provides path-distortion gurantee without modification. This is as in the \texttt{create-petal} procedure (\Cref{alg:create-petal}), we always create a petal (cluster) with measure at most $\frac{1}{2}\mu(Y)$ (and bounded radius). 
\end{remark}

\begin{proof}[Proof of \Cref{clm:ClanTreePartitionAlt}]
	Let $v$ be the point minimizing the ratio $\frac{\mu\left(B_{X}(v,\frac{\diam(X)}{4})\right)}{\mu\left(B_{X}(v,\frac{\diam(X)}{8})\right)}$.
	Set $\rho=\frac{\diam(X)}{8(k+1)}$, and for $i\in[0,k]$
	let $Q_{i}=B_{X}(v,\frac{\diam(X)}{8}+i\cdot\rho)$. Let $i'\in[0,k-1]$
	be the index minimizing $\frac{\mu(Q_{i+1})}{\mu(Q_{i})}$. Then, 
	\[
	\left(\frac{\mu(Q_{k+1})}{\mu(Q_{0})}\right)^{\frac1k}\ge \left(\frac{\mu(Q_{k})}{\mu(Q_{0})}\right)^{\frac1k}=\left(\frac{\mu(Q_{1})}{\mu(Q_{0})}\cdot\frac{\mu(Q_{2})}{\mu(Q_{1})}\cdots\frac{\mu(Q_{k})}{\mu(Q_{k-1})}\right)^{\frac1k}\ge\left(\frac{\mu(Q_{i'+1})}{\mu(Q_{i'})}\right)^{k\cdot\frac{1}{k}}=\frac{\mu(Q_{i'+1})}{\mu(Q_{i'})}~.
	\]
	If $\mu(Q_{i'+1})\le\frac{2}{3}\mu(X)$ or $\mu(Q_{i'})\ge\frac{1}{3}\mu(X)$, fix $i=i'$. Otherwise, fix $i=i'+1$. Note that $i\in [0,k]$.
	Set $R=\frac{\diam(X)}{8}+i\cdot\rho$, and $P=B_{X}(v,R+\rho)$, $Q=B_{X}(v,R)$, $\bar{Q}=X\setminus Q$. 
	Note that $\diam(P)\le2\cdot(\frac{\diam(X)}{8}+(k+1)\cdot\rho)=\frac{\diam(X)}{2}$.
	
	If $i=i'$, then clearly $\frac{\mu(P)}{\mu(Q)}\le\left(\frac{\mu(Q_{k+1})}{\mu(Q_{0})}\right)^{\frac{1}{k}}$ and $\min\{\mu(P),\mu(\bar{Q})\}\le\frac{2}{3}\cdot\mu(X)$. 
	Otherwise, $i=i'+1$, thus $\mu(Q_{i'+1})>\frac{2}{3}\mu(X)$ and $\mu(Q_{i'})<\frac{1}{3}\mu(X)$, implying that $\frac{\mu(Q_{i'+1})}{\mu(Q_{i'})}>2$ and thus 
	\[
	\frac{\mu(P)}{\mu(Q)}=\frac{\mu(Q_{i+1})}{\mu(Q_{i})}=\frac{\mu(Q_{i'+2})}{\mu(Q_{i'+1})}\le\frac{\mu(X)}{\frac{2}{3}\mu(X)}=\frac32<\frac{\mu(Q_{i'+1})}{\mu(Q_{i'})}\le\left(\frac{\mu(Q_{k+1})}{\mu(Q_{0})}\right)^{\frac{1}{k}}~.
	\]
	Furthermore, $\mu(\bar{Q})=\mu(X)-\mu(Q_{i'+1})<\frac{1}{3}\mu(X)$. 
	In both cases we obtain that $\min\{\mu(P),\mu(\bar{Q})\}\le\frac{2}{3}\cdot\mu(X)$ and $\frac{\mu(P)}{\mu(Q)}\le\left(\frac{\mu(Q_{k+1})}{\mu(Q_{0})}\right)^{\frac{1}{k}}$. It remains to prove the second required property.
	
	%	\sloppy 

	Let $u_{P}$ be the point defining $\mu^{*}(P)$, that is $\mu^{*}(P)	=\mu\left(B_{P}(u_{P},\frac{\diam(P)}{4}\right)\le\mu\left(B_{P}(u_{P},\frac{\diam(X)}{8}\right)$.
	Using the minimality of $v$, it holds that
	\[
	\frac{\mu(P)}{\mu(Q)}\le\left(\frac{\mu(Q_{k})}{\mu(Q_{0})}\right)^{\frac{1}{k}}=\left(\frac{\mu\left(B_{X}(v,\frac{\diam(X)}{4})\right)}{\mu\left(B_{X}(v,\frac{\diam(X)}{8})\right)}\right)^{\frac{1}{k}} \stackrel{(*)}{\le} \left(\frac{\mu\left(B_{X}(u_{P},\frac{\diam(X)}{4})\right)}{\mu\left(B_{X}(u_{P},\frac{\diam(X)}{8})\right)}\right)^{\frac{1}{k}}\le\left(\frac{\mu^{*}\left(X\right)}{\mu^{*}\left(P\right)}\right)^{\frac{1}{k}}~.
	\]
	where $(*)$ is due to the choice of $v$.
\end{proof}

%% file: CKM19.tex
\section[Local Search Algorithms]{Local Search Algorithms for Metric Becker Problems}\label{sec:LocalSearch}
In this section we present PTAS's $^{\ref{foot:approximationSchemes}}$ for the metric $\rho$-dominating/independent set problems under the uniform measure. Both algorithms are local search algorithms. The analysis of the algorithm for the metric $\rho$-dominating set problem was presented in~\cite{LeThesis18}.
%communicated to us by Vincent Cohen-Addad, who generously allowed us to publish it. 
This analysis uses techniques similar to the ones used in \cite{CKM19} to construct PTAS for the $k$-means and $k$-median problems in minor-free graphs. The analysis for the metric $\rho$-independent set problem is original (even though similar).

In both proofs we will use $r$-divisions. The following theorem follows from \cite{Frederickson87,AST90} (see \cite{CKM19} for details).
\begin{theorem}[\cite{Frederickson87,AST90}]\label{thm:r-division}
	For every graph $H$, there is an absolute constant $c_H$ such that every $r\in\N$, and every $n$-vertex $H$-minor-free graph $G=(V,E)$, the vertices of $G$ can be divided into clusters $\mathcal{R}$ such that:
	\begin{enumerate}
		\item For every edge $\{u,v\}\in E$, there is a cluster $C\in\mathcal{R}$ such that $u,v\in C$.
		\item For every $C\in\mathcal{R}$, $|C|\le r^2$.
		\item Let $\mathcal{B}$ be the set of vertices appearing in more than a single cluster, called boundary vertices, then $\sum_{C\in\mathcal{R}}|C\cap\mathcal{B}|\le c_H\cdot\frac nr$.
	\end{enumerate} 
\end{theorem}

\subsection{Local search for $\rho$-dominating set under uniform measure}

We state and prove the theorem here when the set of terminals $\mathcal{K}=V$, however it can be easily accommodated to deal with a general terminal set.

\begin{restatable}[]{theorem}{CKM19}
	\label{thm:CKM19}
	There is a polynomial approximation scheme (PTAS) for the metric $\rho$-dominating set problem in $H$-minor-free graphs under the uniform measure.\\
	Specifically, given a weighted $n$-vertex $H$-minor-free graph $G=(V,E,w)$,  and parameters $\eps\in(0,\frac12)$, $\rho>0$, in $n^{O_{|H|}(\eps^{-2})}$ time, one can find a $\rho$-dominating set $S\subseteq V$ such that for every $\rho$-dominating set $\tilde{S}$,
	$|S|\le (1+\eps)|\tilde{S}|$.
\end{restatable}

\begin{algorithm}[t]
	\caption{\texttt{Local search algorithm for metric $\rho$-dominating set}}\label{alg:CKM19}
	\DontPrintSemicolon
	\SetKwInOut{Input}{input}\SetKwInOut{Output}{output}
	\Input{$n$ vertex graph $G=(V,E,w)$, parameters $\rho,s$}
	\Output{$\rho$-dominating set $S$}
	\BlankLine
	$S\leftarrow V$\;
	\While{$\exists$ $\rho$-dominating set $S'\subseteq V$ s.t. $|S'|<|S|$ and $|S\setminus S'|+|S'\setminus S|\le s$}{
		$S\leftarrow S'$\;
	}

	\Return $S$\;
\end{algorithm}

\hspace{-17pt}\emph{Proof.}\ \ 
Set $r=\frac{4c_H}{\eps}$ where $c_H$ is the constant from \Cref{thm:r-division} w.r.t. $H$. 
Let $S$ be the set returned by the local search \Cref{alg:CKM19} with parameters $\rho$, and $s=r^2=O_H(\frac{1}{\eps^2})$.
Clearly $S$ is a $\rho$-dominating set.
The running time of each step of the while loop is at most ${n\choose s}^2\cdot\poly(n)=n^{O_{|H|}(\eps^{-2})}$, as there are at most $n$ iterations, the running time follows.
Let $S_{\opt}$ be the $\rho$-dominating set of minimum cardinality, it remains to prove that $|S|\le(1+\eps)|S_{\opt}|$.

Let $\tilde{V}=S\cup S_{\opt}$, and let $\mathcal{P}$ be a partition of the vertices in $V$ w.r.t. the Voronoi cells with $\tilde{V}$ as centers.
Specifically, for each vertex $u\in V$, $u$ joins the cluster $P_v$ of a vertex $v\in\tilde{V}$ at minimal distance $\min_{v\in\tilde{V}}d_G(u,v)$.  \footnote{For simplicity, we will assume that all the pairwise distances are unique. Alternatively, one can break ties in a consistent way (i.e. w.r.t. some total order).}
Let $\tilde{G}$ be the graph obtained from $G$ by contracting the internal edges in each Voronoi cell (and keeping only a single copy of each edge). Alternatively, one can define $\tilde{G}$ with $\tilde{V}$ as vertex set such that $v,u\in\tilde{V}$ are adjacent iff there is an edge in $G$ between a vertex in $P_u$ to a vertex in $P_v$. Note that $\tilde{G}$ is a minor of $G$, and hence is $H$-minor-free.

Next, we use \Cref{thm:r-division} on $\tilde{G}$ to obtain $r$-division  $\mathcal{R}$, with $\mathcal{B}$ as boundary vertices.
Consider a cluster $C\in \mathcal{R}$, and let $C'=C\cap(\mathcal{B}\cup S_{\opt})$. Fix $S'=(S\setminus C)\cup C'$. 
\begin{claim}
	$S'$ is a $\rho$-dominating set.
\end{claim}
	\hspace{-17pt}\emph{Proof.}\ \  
	Consider a vertex $u\in V$, We will argue that $u$ is at distance at most $\rho$ from some vertex in $S'$.
	Let $v_1\in S$ (resp. $v_2\in S_\opt$) be the closest vertex to $u$ in $S$ (resp. in $S_\opt$). It holds that $d_G(u,v_1),d_G(u,v_2)\le\rho$.
	If either $v_1\notin C$, $v_1\in C\cap\mathcal{B}$, or $v_2\in C$ then $S'$ contains at least one of $v_1,v_2$ and we are done. Thus we can assume that  $v_1\in C\setminus\mathcal{B}$ and $v_2\notin C$. 
	Let $\Pi=\{v_1=z_0,z_1,\dots,z_a,u,w_0,w_1,\dots,w_b=v_2\}$ be the unique shortest path from $v_1$ to $v_2$ that goes through $u$ (the thick black line in illustration on the right). 
	
\begin{wrapfigure}{r}{0.25\textwidth}
	\begin{center}
		\vspace{-25pt}
		\includegraphics[scale=0.7]{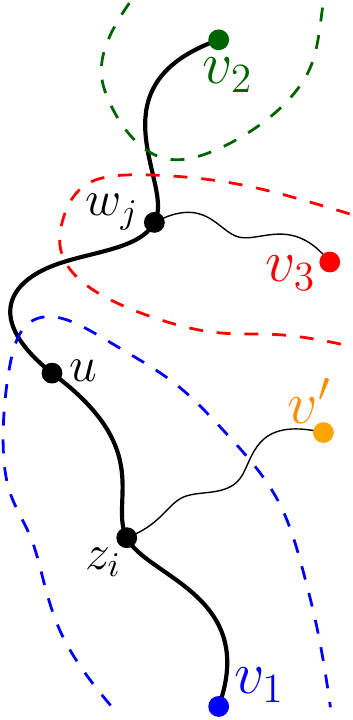}
		\vspace{-15pt}
	\end{center}
%	\vspace{-10pt}
\end{wrapfigure}
	
	Assume first that $u$ belongs to the Voroni cell $P_{v_1}$ of $v_1$ (encircled by a blue dashed line). For every $i$ and $v'\in\tilde{V}$ it holds that $d_{G}(v',z_{i})\ge d_{G}(v',u)-d_{G}(u,z_{i})>d_{G}(v_{1},u)-d_{G}(u,z_{i})=d_{G}(v_{1},z_{i})$. It follows that all the vertices $\{z_0,z_1,\dots,z_a\}$ belong to the Voronoi cell $P_{v_1}$. 
	As $v_1\in C\setminus\mathcal{B}$, and $v_2\notin C$, there must be some index $j$ such that $w_j$ belongs to the Voronoi cell $P_{v_3}$ of $v_3\in C\cap \mathcal{B}$ (as otherwise there will be an edge in $\tilde{G}$ between a vertex in $C\setminus\mathcal{B}$ to a vertex not in $C$). 
	It holds that 
	\[
	d_{G}(u,v_{3})\le d_{G}(u,w_{j})+d_{G}(w_{j},v_{3})\le d_{G}(u,w_{j})+d_{G}(w_{j},v_{2})=d_{G}(u,v_{2})\le\rho~,
	\]
	where the first inequality follows by triangle inequality, the second as $w_j\in P_{v_3}$, and the equality as $w_j$ lays on the shortest path from $u$ to $v_2$. 
	As $v_3\in C\cap \mathcal{B}$ it follows that $v_3\in S'$, thus we are done.
	The case $u\in P_{v_2}$ is symmetric.
	\qed
	
	\vspace{7pt}
	
It holds that $|S'\setminus S|+|S\setminus S'|\le |C|\le r^2=s$. Thus, $|S'|\ge |S|$ since  otherwise, \Cref{alg:CKM19} would've not returned the set $S$. Hence  $|C\cap(\mathcal{B}\cup S_{\opt})|=|C'|\ge|C\cap S|$. As the same argument could be applied on every cluster $C\in\mathcal{R}$, we conclude that,
\[
|S|=\sum_{C\in\mathcal{R}}|C\cap S|\le\sum_{C\in\mathcal{R}}|C\cap(\mathcal{B}\cup S_{\opt})|\le|S_{\opt}|+\sum_{C\in\mathcal{R}}|C\cap\mathcal{B}|\le|S_{\opt}|+c_H\cdot\frac{|\tilde{V}|}{r}\le|S_{\opt}|+2c_H\cdot\frac{|S|}{r}~.
\]
But this implies 
$|S_{\opt}|\ge(1-\frac{2c_H}{r})|S|=(1-\frac{\eps}{2})|S|$,
thus $|S|\le\frac{1}{1-\frac\eps2}|S_{\opt}|\le(1+\eps)|S_{\opt}|$.
\qed

\subsection{Local search for $\rho$-independent set under uniform measure}
\begin{restatable}[]{theorem}{CKM19}\label{thm:Local-Sarch-Independent}
	There is a polynomial approximation scheme (PTAS) for the metric $\rho$-independent set problem in $H$-minor-free graphs under uniform measure.\\
	Specifically, given a weighted $n$-vertex $H$-minor-free graph $G=(V,E,w)$,  and parameters $\eps\in(0,\frac12)$, $\rho>0$, in $n^{O_{|H|}(\eps^{-2})}$ time, one can find a $\rho$-independent set $S\subseteq V$ such that for every $\rho$-independent set $\tilde{S}$,
	$|S|\ge (1-\eps)|\tilde{S}|$.
\end{restatable}

\begin{algorithm}[t]
	\caption{\texttt{Local search algorithm for metric $\rho$-independent set}}\label{alg:LocalSearchIndependentSet}
	\DontPrintSemicolon
	\SetKwInOut{Input}{input}\SetKwInOut{Output}{output}
	\Input{$n$ vertex graph $G=(V,E,w)$, parameters $\rho,s$}
	\Output{$\rho$-independent set $S$}
	\BlankLine
	$S\leftarrow \emptyset$\;
	\While{$\exists$ $\rho$-independent set $S'\subseteq V$ s.t. $|S'|>|S|$ and $|S\setminus S'|+|S'\setminus S|\le s$}{
		$S\leftarrow S'$\;
	}
	
	\Return $S$\;
\end{algorithm}
\begin{proof}
	Set $r=\frac{4c_H}{\eps}$ where $c_H$ is the constant from \Cref{thm:r-division} w.r.t. $H$.
	Let $S$ be the set returned by the local search \Cref{alg:LocalSearchIndependentSet} with parameters $\rho$, and $s=r^2=\frac{16c_H^2}{\eps^2}=O_H(\frac{1}{\eps^2})$.
	Clearly $S$ is a $\rho$-independent set.
	The running time of each step of the while loop is at most ${n\choose s}^2\cdot\poly(n)=n^{O_{|H|}(\eps^{-2})}$, as there are at most $n$ iterations, the running time follows.
	Let $S_{\opt}$ be the $\rho$-independent set of maximum cardinality, it remains to prove that $|S|\ge(1-\eps)|S_{\opt}|$.
	
	Construct a graph $\tilde{G}$ with $\tilde{V}=S\cup S_\opt$ as a vertex set. We add an edge an edge between $u,v\in \tilde{V}$ iff $d_G(u,v)<\rho$. 
	Clearly all the edges are from $S\times S_\opt$ (as both $S,S_\opt$ are $\rho$-independent sets). 
	Note that $\tilde{V}$ is a minor of $G$. This is because if we take all the shortest paths $P_{u,v}$ for $\{u,v\}\in E'$ they will not intersect. To see this, assume for contradiction that there are different pairs $u,u'\in S_\opt$ , $v,v'\in S$ such that $\{u,v\},\{u',v'\}\in E$, and there is some vertex $z$ such that $z\in P_{u,v}\cap P_{u',v'}$.
	W.l.o.g. assume that $d_G(u,z)+d_G(u',z)\le d_G(z,v)+d_G(z,v')$. Using the triangle inequality it follows that 
	\begin{align*}
	d_{G}(u,u')\le d_{G}(u,z)+d_{G}(u',z) & \le\frac{1}{2}\cdot\left(d_{G}(u,z)+d_{G}(z,v)+d_{G}(u',z)+d_{G}(z,v')\right)\\
	& =\frac{1}{2}\cdot\left(d_{G}(u,v)+d_{G}(u',v')\right)<\rho~,
	\end{align*}
	a contradiction.
	
	Next, we apply \Cref{thm:r-division} to $\tilde{G}$ to obtain $r$-division  $\mathcal{R}$, with $\mathcal{B}$ as boundary vertices.
	Consider a cluster $C\in \mathcal{R}$, and let $C'=(C\cap S_{\opt})\setminus \mathcal{B}$. Fix $S'=(S\setminus C)\cup C'$. 
	\begin{claim}
		$S'$ is a $\rho$-independent set.
	\end{claim}
	\begin{proof}
		Consider a pair of vertices $u,v\in S'$, we will show that $d_G(u,v)\ge\rho$.
		If both $u,v$ belong to $S$, then since $S$ is a $\rho$-independent set, it follows that $d_G(u,v)\ge\rho$. The same argument holds if both $u,v$ belong to $S_\opt$. We thus can assume w.l.o.g. that $u\in S\setminus S_\opt$ and $v\in S_\opt\setminus S$.
		It follows that $u\notin C$ while $v\in C$.
		However, as $v\in C\cap S'$, necessarily $v\notin\mathcal{B}$. The only vertices in $C$ with edges towards vertices out of $C$ are in  $\mathcal{B}$. It follows that  $\{u,v\}$ is not an edge of $\tilde{G}$, implying $d_G(u,v)\ge\rho$.
	\end{proof}
	It holds that $|S'\setminus S|+|S\setminus S'|\le |C|\le r^2=s$. Thus, $|S'|\le |S|$, as otherwise \Cref{alg:LocalSearchIndependentSet} would have not returned the set $S$. Hence, $|(C\cap S_{\opt})\setminus \mathcal{B}|=|C'|\le|C\cap S|$. As the same argument could be applied on every cluster $C\in\mathcal{R}$, we conclude that,
	\[
	|S|=\sum_{C\in\mathcal{R}}|C\cap S|\ge\sum_{C\in\mathcal{R}}|(C\cap S_{\opt})\setminus\mathcal{B}|\ge|S_{\opt}|-\sum_{C\in\mathcal{R}}|C\cap\mathcal{B}|\ge|S_{\opt}|-c_{H}\cdot\frac{|\tilde{V}|}{r}\ge|S_{\opt}|-2c_{H}\cdot\frac{|S|}{r}~.
	\]
	But this implies that
	$|S_{\opt}|\le(1+\frac{2c_{H}}{r})|S|=(1+\frac{\eps}{2})|S|$,
	thus $|S|\ge\frac{1}{1+\frac{\eps}{2}}|S_{\opt}|\ge(1-\eps)|S_{\opt}|$.
\end{proof}